\documentclass[12pt,a4paper]{amsart}

\usepackage{amsmath,amssymb,amsthm,mathrsfs}
\usepackage[margin=1in]{geometry}
\usepackage{enumitem}
\usepackage{hyperref}
\usepackage{tikz-cd}
\usepackage{xcolor}
\usepackage{esint}
\usepackage{comment}
\usepackage{tikz}
\usetikzlibrary{decorations.pathmorphing}

\theoremstyle{plain}
\newtheorem{theorem}{Theorem}[section]
\newtheorem{conjecture}[theorem]{Conjecture}
\theoremstyle{definition}
\newtheorem{corollary}[theorem]{Corollary}

\newtheorem{proposition}[theorem]{Proposition}
\newtheorem{definition}[theorem]{Definition}
\newtheorem{example}[theorem]{Example}
\newtheorem{construction}[theorem]{Construction}
\theoremstyle{remark}
\newtheorem{remark}[theorem]{Remark}

\newcommand{\Zrt}{Z_{\mathrm{RT}}}
\newcommand{\Zbv}{Z_{\mathrm{BV}}}
\newcommand{\Loc}{\mathrm{Loc}}
\newcommand{\GC}{G_\CC}
\newcommand{\sd}{\mathrm{sd}}

\newcommand{\Rep}{\mathrm{\mathbf{Rep}}}
\newcommand{\Obs}{\mathrm{Obs}}
\newcommand{\EE}{\mathbb{E}}
\newcommand{\ZZ}{\mathbb{Z}}
\newcommand{\CC}{\mathbb{C}}
\newcommand{\QQ}{\mathbb{Q}}
\newcommand{\RR}{\mathbb{R}}

\newcommand{\fg}{\mathfrak{g}}
\newcommand{\Flat}{\mathrm{Flat}}
\newcommand{\Quant}{\mathrm{Quant}}
\newcommand{\BV}{\mathrm{BV}}

\newcommand{\CS}{\mathrm{CS}}
\newcommand{\MTC}{\mathscr{C}}
\newcommand{\Bord}{\mathrm{\mathbf{Bord}}}
\newcommand{\Vect}{\mathrm{\mathbf{Vect}}}
\newcommand{\id}{\mathrm{id}}
\newcommand{\op}{\mathrm{op}}

\newcommand{\IndCoh}{\mathrm{\mathbf{IndCoh}}}
\newcommand{\QCoh}{\mathrm{\mathbf{QCoh}}}
\DeclareMathOperator{\Hom}{Hom}

\DeclareMathOperator{\Map}{Map}

\DeclareMathOperator{\Tr}{Tr}
\DeclareMathOperator{\HH}{HH}
\DeclareMathOperator{\MCG}{MCG}
\DeclareMathOperator{\Ad}{Ad}

\DeclareMathOperator{\rk}{rk}
\DeclareMathOperator{\Ext}{Ext}

\DeclareMathOperator{\Perf}{Perf}
\DeclareMathOperator{\Coh}{Coh}
\newtheorem{question}[theorem]{Question}
\newtheorem{principle}[theorem]{Principle}
\newtheorem{proposal}[theorem]{Proposal}

\usepackage{mathbbol}

\title[From BV-BFV Quantization to RT Invariants]{From BV-BFV Quantization to\\ Reshetikhin--Turaev Invariants}
\author[N. Moshayedi]{Nima Moshayedi}
\address{Institut f\"ur Mathematik\\ Universit\"at Z\"urich\\
Winterthurerstrasse 190
CH-8057 Z\"urich}
\email[N. Moshayedi]{nima.moshayedi@math.uzh.ch}
\date{}
\thanks{This paper was prepared with the assistance of AI tools for writing, mathematical exposition, 
literature search, and typesetting. All content has been 
verified by the author to the best of his ability. The 
author takes full responsibility for all mathematical 
content, including the formulation of conjectures, the 
correctness of attributed results, and any errors.}

\begin{document}
\maketitle

\begin{abstract}
We propose a program for bridging the gap between the perturbative BV-BFV quantization of Chern--Simons theory and the non-perturbative Reshetikhin--Turaev (RT) invariants of 3-manifolds, passing through factorization homology of $\EE_n$-algebras and the derived algebraic geometry of character stacks. We conjecture that the modular tensor category underlying the RT construction arises as the $\EE_2$-category from BV-BFV quantization of Chern--Simons theory on the disk, with the derived character stack $\Loc_{\GC}(\Sigma)$ of the complexified group and its shifted symplectic structure mediating the proposed identification. A recurring theme is the mismatch between the compact gauge group of Chern--Simons theory and the reductive algebraic group of derived algebraic geometry: we work throughout with $\GC$ and indicate where a Riemann--Hilbert type comparison is needed to return to the compact theory. We formulate a main conjecture asserting natural equivalence of the BV-BFV and RT constructions as $(3\text{-}2\text{-}1)$-extended topological quantum field theories, together with several auxiliary conjectures and two questions, develop a proof strategy via deformation quantization of shifted symplectic stacks, and clarify the role of $\EE_n$-Koszul duality, which relates the perturbative expansion around a flat connection to the local geometry of the character stack. We relate the perturbative expansions to the sheaf of vanishing cycles on the character stack of a $3$-manifold, propose that a stratification of the character stack by stabilizer type organizes the trans-series decomposition, and propose that the Stokes data of the perturbative expansion should be read as a stability datum in the sense of Kontsevich--Soibelman, with the Chern--Simons functional playing the role of the central charge. Supporting evidence is examined in the abelian, low-genus, and Seifert fibered cases. Connections to resurgence, categorification, and the quantum geometric Langlands program are discussed as further motivation. Throughout, we distinguish between results in the literature, conjectures, and heuristics.
\end{abstract}

\tableofcontents

\section{Introduction}\label{sec:intro}

\subsection{Historical context and motivation}

The Reshetikhin--Turaev (RT) invariants \cite{RT1990} constitute one of the central achievements of quantum topology. Starting from a modular tensor category $\MTC$ (prototypically the semisimple quotient of representations of a quantum group $U_q(\fg)$ at $q$ a root of unity), one constructs a $(2{+}1)$-dimensional topological quantum field theory (TQFT) assigning vector spaces to closed surfaces and linear maps to cobordisms. Applied to a closed 3-manifold presented via surgery on a framed link, the construction yields a complex-valued topological invariant $\Zrt(M)$.

Since any competing construction must reproduce them, we record at the outset what these
invariants detect. For $\fg =
\mathfrak{sl}_2$ the associated link invariants specialize to the Jones polynomial and its
coloured relatives, so the RT theory already detects chirality of knots. At the level of
3-manifolds, the RT invariants distinguish lens spaces that are homotopy equivalent but not
homeomorphic \cite{Jeffrey1992}; their large-level asymptotics recover the Chern--Simons
invariants and Reidemeister torsions of the flat connections \eqref{eq:asymptotic}; the first
non-trivial term in the expansion of the $SU(2)$ invariant about $q = 1$ is the Casson
invariant of a homology sphere \cite{Ohtsuki1996}; and the associated mapping class group
representations are asymptotically faithful (Theorem \ref{thm:Andersen}), so no non-trivial
mapping class is invisible at all levels. On the other hand they are not complete invariants,
and it is not known whether the collection of all RT invariants distinguishes all 3-manifolds;
the extent to which they see geometric, as opposed to homotopy-theoretic, information is the
content of the volume conjecture and remains open.

The construction proceeds in two stages. First, Reshetikhin and Turaev \cite{RT1990} defined invariants of framed oriented links in $S^3$ colored by representations of a quantum group, generalizing the Jones polynomial. Second, using the Kirby calculus \cite{Kirby1978} (which relates different surgery presentations of the same 3-manifold via a finite set of moves; see also \cite{FennRourke1979} for the refinement to rational surgery and \cite{GompfStipsicz1999} for a textbook account), they showed that an appropriate normalization of the colored link invariant, summed over all irreducible representations, yields a topological invariant of the resulting 3-manifold. Turaev \cite{Turaev1994} later axiomatized the algebraic input as a modular tensor category, abstracting away from the specific quantum group origin (see \cite{BK2001,EGNO2015} for textbook treatments).

On the field-theoretic side, Witten \cite{Witten1989} identified these invariants with the partition function of Chern--Simons gauge theory at level $k$ (the level is related to the root of unity by $q = e^{i\pi/(k+h^\vee)}$; see below). For a compact simple Lie group $G$ with Lie algebra $\fg$, and a closed oriented 3-manifold $M$, the Chern--Simons partition function is formally
\begin{equation}\label{eq:CS-path-integral}
Z_{\CS}(M) = \int_{\mathcal{A}/\mathcal{G}} \exp\!\left(\frac{ik}{4\pi} \int_M \Tr\left(A \wedge dA + \tfrac{2}{3} A \wedge A \wedge A\right)\right) \mathcal{D}A,
\end{equation}
where $\mathcal{A}$ is the space of connections on a principal $G$-bundle over $M$, $\mathcal{G}$ is the gauge group, $\Tr$ denotes an invariant bilinear form on $\fg$ normalized so that the action is $2\pi$-periodic in $k$, and $\mathcal{D}A$ is the (formal) path integral measure. Witten argued, using a combination of canonical quantization and path integral methods, that this produces the Jones polynomial for knot complements and the RT invariant for closed 3-manifolds, with the identification $q = e^{2\pi i/(k + h^\vee)}$, where $h^\vee$ is the dual Coxeter number of $\fg$. (In the quantum group literature, one often writes $q = e^{i\pi/(k+h^\vee)}$ (e.g. \cite{Turaev1994}), so that $q^2$ equals Witten's parameter; we adopt this latter convention from Section \ref{sec:RT} onward.)

Making this path integral rigorous in the perturbative regime has been the subject of extensive work by Axelrod--Singer \cite{AS1994}, Kontsevich \cite{Kontsevich1994}, Bar-Natan \cite{BarNatan1995}, Bott--Cattaneo \cite{BC1998}, and many others. The modern culmination of this line of work is the BV-BFV program of Cattaneo, Mnev, and Reshetikhin \cite{CMR2014,CMR2018,CMR2020}, which provides a general framework for the perturbative quantization of gauge theories on manifolds with boundary, assigning to the bulk a BV structure and to the boundary a BFV structure, subject to a compatibility condition (the modified quantum master equation). Applied to Chern--Simons theory, this yields well-defined perturbative invariants as formal power series in a parameter $\hbar$.

Here $k$ is a positive integer, whereas perturbation theory requires a formal variable. The
perturbative regime is $k \to \infty$, and the standard device is to treat $1/(k+h^\vee)$ as a
formal expansion parameter: one expands around $k = \infty$ and only afterwards asks whether
the resulting series can be evaluated at the physical values. Accordingly we write
\[
\hbar \;=\; \frac{2\pi i}{k+h^\vee},
\]
and regard $\hbar$ as a \emph{formal} variable, so that the BV-BFV output is an element of
$\CC[\![\hbar]\!]$; the substitution of an actual complex number for $\hbar$, in particular of
the value corresponding to an integer level, is the ``de-formalization'' problem that this
paper is about (see Remark \ref{rmk:perturbative} and Remark \ref{rmk:q-parameter}). The
quantum parameter is then $q = e^{\hbar/2}$, consistently with $q = e^{i\pi/(k+h^\vee)}$.

\subsection{The main question}

The question that motivates this paper is:

\medskip
\noindent\textbf{Main Question:} \emph{How are the perturbative invariants produced by BV-BFV quantization related to the exact, non-perturbative RT invariants?}
\medskip

The standard expectation, going back to Witten and elaborated by many authors \cite{Freed1995,Jeffrey1992,Garoufalidis2004,Gukov2005,Marino2004}, is that the RT invariant admits an asymptotic expansion as $k \to \infty$ of the form
\begin{equation}\label{eq:asymptotic}
\Zrt(M,k) \;\sim\; \sum_{[A] \in \Flat_G(M)} e^{2\pi i k \,\CS(A)} \cdot k^{\sd(A)/2} \cdot \tau_A(M) \cdot \left(1 + \sum_{n=1}^\infty a_n(A)\, k^{-n}\right),
\end{equation}
where the sum runs over gauge-equivalence classes of flat $G$-connections on $M$ (critical points of the Chern--Simons functional), $\CS(A)$ is the classical Chern--Simons invariant of the flat connection \cite{ChernSimons1974,Freed1995},
\[
\sd(A) \;=\; \dim H^0(M; \Ad_A) - \dim H^1(M; \Ad_A)
\]
is a \emph{spectral dimension} encoding the size of the stabilizer and of the deformation space,
$\Ad_A$ denotes the local system on $M$ with fibre $\fg$ associated to the flat connection $A$
through the adjoint representation $\Ad: G \to \mathrm{GL}(\fg)$ (equivalently, the flat bundle
$\widetilde{M} \times_{\rho_A} \fg$ determined by the holonomy $\rho_A: \pi_1(M) \to G$), and
$H^\bullet(M;\Ad_A)$ is its twisted de Rham (equivalently, twisted group) cohomology. Further,
$\tau_A(M)$ is the Reidemeister--Ray--Singer analytic torsion \cite{RS1971} (a regularized
determinant of the twisted de Rham operator), and the coefficients $a_n(A)$ are the Feynman
diagram contributions at $(n+1)$ loops, which are the quantities computed by the BV-BFV
formalism.

We note two points about \eqref{eq:asymptotic}. First, the notation
$\sum_{[A]}$ is schematic: the moduli space $\Flat_G(M)$ of flat connections is in general not a
finite set, and when it has positive-dimensional components the corresponding contribution to
the asymptotic expansion is an integral over that component (with respect to a measure built
from the torsion) rather than a term in a finite sum; the discrete sum is correct only when the
flat connections are isolated and non-degenerate, which is the case for instance on rational
homology spheres with irreducible $SU(2)$ representations that are acyclic. Second, we write
$\sd(A)$ rather than the more common $d_A$, since $d_A$ denotes the covariant derivative
$d + [A,-]$ throughout.

Several obstructions prevent a direct passage from \eqref{eq:asymptotic} to an identification of $\Zrt$ with $\Zbv$:

\begin{enumerate}[label=(\alph*)]
\item \textbf{Divergence of the perturbative series.} The power series $\sum_n a_n(A) k^{-n}$ is generically divergent (zero radius of convergence). Making sense of the right-hand side of \eqref{eq:asymptotic} requires Borel resummation or more general resurgent analysis. An arbitrary formal power series does \emph{not} admit a convergent Borel transform: for the Borel transform $\sum_n a_n \zeta^n/n!$ to have positive radius of convergence one needs a growth condition on the coefficients, namely $|a_n| \le C\, \sigma^n\, n!$ for some constants $C, \sigma > 0$ (the Gevrey-$1$ condition), and for Borel--Laplace summation one needs in addition analytic continuation and controlled growth along the relevant ray. Chern--Simons perturbation theory is \emph{expected} to satisfy such bounds, and this is verified in examples \cite{GMP2016,GGM2021}, but in general it is an assumption; see Definition \ref{def:resurgent}.

\item \textbf{Stokes phenomena.} The Borel resummations of the perturbative series around different flat connections are not independent. They are related by Stokes automorphisms, which relate the contributions of different critical points. The Stokes constants carry information invisible to any single perturbative expansion.

\item \textbf{Change of description.} The BV-BFV formalism works within gauge theory (connections on bundles), while the RT construction works within algebra (modular tensor categories, surgery calculus). These are different mathematical frameworks, and matching them requires identifying the algebraic structures that emerge from quantization.

\item \textbf{Phase and normalization.} Even where the asymptotic expansion is verified, matching the overall phase and normalization requires careful treatment of framings, $\eta$-invariants, and gravitational Chern--Simons terms.

\item \textbf{Functoriality.} Both frameworks produce TQFT-like structures (assignment of state spaces to surfaces, gluing), but matching these structures requires identifying not just the partition functions for closed manifolds but the entire functorial apparatus.
\end{enumerate}

\subsection{Our approach}

In this paper, we propose an approach to the Main Question that avoids the analytic difficulties of asymptotic resummation entirely, instead passing through algebraic and categorical structures. Our thesis is:

\medskip
\noindent\textbf{Main Thesis.} \emph{The modular tensor category $\MTC$ used in the RT construction should be equivalent, as an $\EE_2$-algebra, to the braided monoidal category arising from BV-BFV quantization of Chern--Simons theory on the disk. If established, this equivalence, mediated by the derived geometry of character stacks, would provide a direct identification of the RT functor with the (properly completed) BV-BFV partition function, bypassing the need for resurgent resummation. This paper formulates a precise program toward this identification and assembles supporting evidence, but significant technical gaps remain open.}

\medskip

We do not claim that this equivalence is canonical. In the strict sense, canonicity would mean a
natural equivalence of functors defined uniformly in the data, for instance on the groupoid of
pairs $(G,k)$ or of quantizable $1$-shifted symplectic stacks. Since we do not know how to
formulate the BV-BFV side as such a functor (see Section \ref{sec:outlook}), we claim only an
equivalence of $\EE_2$-algebras for each $G$ and $k$; Conjecture \ref{conj:E2-equiv} adds that it
can be chosen compatibly with the $\CC[\![\hbar]\!]$-linear structure and with the
Etingof--Kazhdan quantization.

\medskip

The strategy involves four layers of structure:

\begin{enumerate}[label=\textbf{Layer \arabic*.},leftmargin=3.5em]
\item \textbf{Classical:} The derived character stack $\Loc_{\GC}(\Sigma)$, with its shifted symplectic structure (PTVV \cite{PTVV2013}), provides the classical phase space common to both frameworks, after the compact group $G$ of the gauge theory has been replaced by its complexification $\GC$ (Section \ref{subsec:compact-vs-complex}).
\item \textbf{Quantum-algebraic:} Deformation quantization of $\Loc_{\GC}(\Sigma)$ produces the quantum group $U_q(\fg)$ and its representation category $\Rep_q(\GC)$, which at roots of unity yields the modular tensor category. The relevant \emph{shifted} deformation quantization machinery is that of CPTVV \cite{CPTVV2017} together with Safronov's identification of Poisson--Lie with $1$-shifted Poisson structures \cite{Safronov2017b} and Poisson additivity \cite{Safronov2018}; the identification of the output with quantum groups goes back to Drinfeld \cite{Drinfeld1987} and Etingof--Kazhdan \cite{EK1996}, and the integration over surfaces is BBJ \cite{BBJ2018}. For quantum groups at roots of unity and the resulting modular tensor categories, see Lusztig \cite{Lusztig1990}, Andersen \cite{AndersenH1992}, Chari--Pressley \cite{CP1994}, Bakalov--Kirillov \cite{BK2001} and \cite{EGNO2015}. This layer rests on established theorems, though the compatibility with BV-BFV quantization (Layer 4) remains conjectural.
\item \textbf{Topological:} Factorization homology \cite{AF2015,AFT2019,Francis2013,Lurie_HA} integrates the $\EE_2$-category $\Rep_q(\GC)$ over surfaces to produce quantum character varieties (BBJ \cite{BBJ2018}), which at roots of unity yield the RT state spaces as distinguished objects (AKZ \cite{AKZ2017}). These are established results; the conjectural step is connecting them to the BV-BFV side.
\item \textbf{Field-theoretic:} The BV-BFV quantization, reformulated in derived algebraic geometry, should coincide with the deformation quantization of Layer 2. This is the main open problem and the subject of Conjecture \ref{conj:quant-agree}.
\end{enumerate}

\begin{figure}[ht]
\centering
\begin{tikzpicture}[
    box/.style={rounded corners=6pt, minimum width=3.8cm, minimum height=1.4cm, align=center, font=\small\bfseries, text=white, draw=none},
    arr/.style={-{stealth}, thick},
    darr/.style={-{stealth}, thick, dashed},
    every node/.style={font=\small},
    scale=0.95
]
% Four pillars
\node[box, fill=cyan!60!black] (bv) at (-5,2.5) {BV-BFV\\[-2pt]\fontsize{7}{9}\selectfont perturbative, non-extended\\[-1pt]\fontsize{7}{9}\selectfont $\Zbv: \Bord_3^{\mathrm{or}} \to \Vect_\hbar$};
\node[box, fill=magenta!60!black] (rt) at (5,2.5) {Reshetikhin--Turaev\\[-2pt]\fontsize{7}{9}\selectfont non-perturbative, extended\\[-1pt]\fontsize{7}{9}\selectfont $\Zrt: \Bord_3^{\mathrm{ext}} \to 2\text{-}\Vect$};
\node[box, fill=green!60!black] (fh) at (5,-2.5) {Factorization\\Homology\\[-2pt]\fontsize{7}{9}\selectfont $\iint_\Sigma \mathcal{A}$, $\otimes$-excision\\[-1pt]\fontsize{7}{9}\selectfont BBJ, AKZ};
\node[box, fill=orange!60!black] (dag) at (-5,-2.5) {Derived Character\\Stack\\[-2pt]\fontsize{7}{9}\selectfont shifted symplectic\\[-1pt]\fontsize{7}{9}\selectfont PTVV, Calaque};
% Center object
\node[draw=black!70, thick, rounded corners=4pt, fill=white, inner sep=6pt, align=center] (center) at (0,0) {$\Loc_{\GC}(\Sigma)$\\[-2pt]\fontsize{7}{9}\selectfont\color{black!60} $(2{-}d)$-shifted symplectic};
% Solid arrows (established or conjectural connections)
\draw[arr, orange!60!black] (dag) -- node[left, align=right, text=black, font=\footnotesize] {Layer 1\\[-2pt]\fontsize{7}{9}\selectfont classical\\[-1pt]\fontsize{7}{9}\selectfont phase space} (bv);
\draw[arr, orange!60!black] (dag) -- node[below, text=black, font=\footnotesize, align=center] {Layer 2: deformation quantization\\[-2pt]\fontsize{7}{9}\selectfont CPTVV, Safronov, Etingof--Kazhdan} (fh);
\draw[arr, green!60!black] (fh) -- node[right, align=left, text=black, font=\footnotesize] {Layer 3\\[-2pt]\fontsize{7}{9}\selectfont semisimplify\\[-1pt]\fontsize{7}{9}\selectfont $+$ AKZ} (rt);
% Dashed arrow (the main conjecture)
\draw[darr, magenta!60!black, line width=1.2pt] (bv) -- node[above, text=black, font=\footnotesize\bfseries] {Main Conjecture} node[below, text=black!60, font=\fontsize{7}{9}\selectfont] {promote to $(3\text{-}2\text{-}1)$-ext.\ and identify} (rt);
% Dashed arrow (Layer 4)
\draw[darr, cyan!60!black] (-4.8,1.7) -- node[right=4pt, align=left, text=black, font=\footnotesize] {Layer 4\\[-2pt]\fontsize{7}{9}\selectfont BV-BFV $=$\\[-1pt]\fontsize{7}{9}\selectfont def.\ quant.?} (-4.8,-1.5);
% Connections to center
\draw[thick, black!30] (center) -- (bv);
\draw[thick, black!30] (center) -- (rt);
\draw[thick, black!30] (center) -- (fh);
\draw[thick, black!30] (center) -- (dag);
\end{tikzpicture}
\caption{Schematic overview of the program. The BV-BFV functor (top left) is an ordinary, perturbative TQFT valued in $\Vect_\hbar$. The RT construction (top right) is a $(3\text{-}2\text{-}1)$-extended TQFT valued in $2\text{-}\Vect$. The Main Conjecture asserts that the BV-BFV data can be promoted to a $(3\text{-}2\text{-}1)$-extended TQFT that agrees with RT. Solid arrows indicate connections supported by existing results; dashed arrows indicate conjectural identifications. The derived character stack $\Loc_{\GC}(\Sigma)$ sits at the center, mediating between all four frameworks.}
\label{fig:big-picture}
\end{figure}
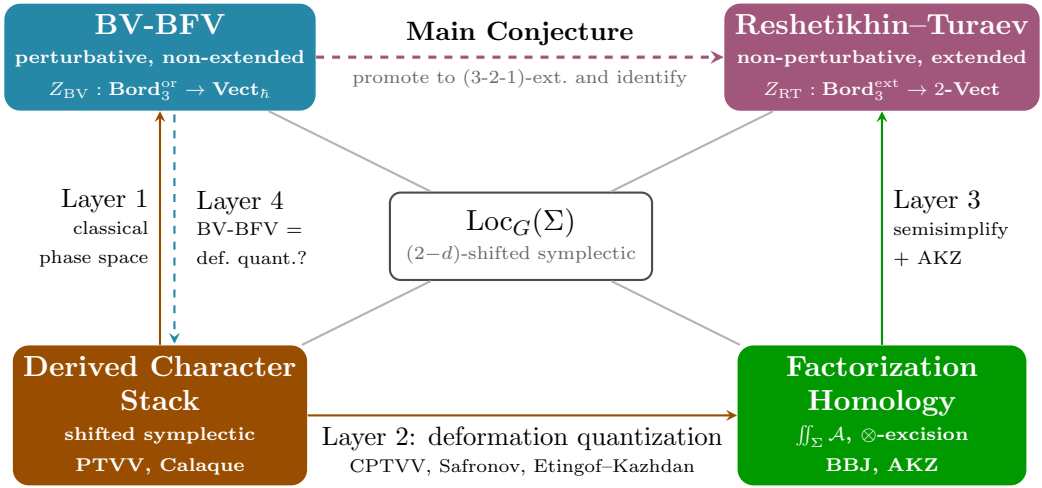

\subsection{Compact versus complex gauge group}\label{subsec:compact-vs-complex}

Chern--Simons theory, as a physical theory and as it is treated in the BV-BFV literature,
takes as input a \emph{compact} simple Lie group $G$ and an integer level $k$. Derived
algebraic geometry, as it is used in the PTVV/CPTVV/BBJ circle of ideas, takes as input a
\emph{reductive algebraic group over $\CC$}. These inputs are different, and the passage
between them is one of the most delicate points of this paper. We fix our conventions here.

\medskip
\noindent\textbf{Convention.} From Section \ref{sec:derived-character} onwards, $\GC$ denotes
the complex reductive algebraic group obtained by complexifying the compact group $G$ (so
$G = SU(n)$ gives $\GC = SL_n(\CC)$, $G = U(1)$ gives $\GC = \CC^\times$). All derived
character stacks, all shifted symplectic structures, all deformation quantizations, and all of
the resurgence discussion of Section \ref{sec:resurgence} are taken with respect to $\GC$. When
we write $\Loc_{\GC}(\Sigma) = \Map(\Sigma_B, B\GC)$ we mean the algebraic object; when we
speak of flat $G$-connections we mean the differential-geometric object; these are not the same.

\medskip

Three distinct issues should be kept apart.

\begin{enumerate}[label=(\alph*)]
\item \emph{The stacks are different.} For $\Sigma = S^1$ the differential-geometric object is
the quotient groupoid $[G/G]$ of the compact group acting on itself by conjugation, which lives
in differential geometry (or in derived $C^\infty$-geometry), while $\Loc_{\GC}(S^1) =
[\GC/\GC]$ is an algebraic stack. These may be compared after passing to underlying topological
stacks (since $G \hookrightarrow \GC$ is a homotopy equivalence, $BG \simeq B\GC$ as topological
stacks), but they are not equivalent as stacks in their respective geometric contexts, and
quasi-coherent sheaves on them differ; for instance, the adjoint quotients of $U(n)$ and of
$GL_n(\CC)$ have different rings of functions. A similar subtlety distinguishes the Betti and
de Rham versions of the geometric Langlands program; compare Section \ref{subsec:langlands}.

\item \emph{Flatness is a Riemann--Hilbert statement.} Even after complexification, the
identification of the algebraic stack $\Loc_{\GC}(\Sigma) = \Map(\Sigma_B, B\GC)$ with a moduli
space of flat $\GC$-connections is not a definition but a theorem: one must relate maps out of
the Betti stack (representations of $\pi_1$, up to the adjoint action) with $\GC$-bundles with
flat connection. For a compact manifold this is the Riemann--Hilbert correspondence, and it is
an equivalence of \emph{stacks in topological spaces}; whether it upgrades to an equivalence of
\emph{algebraic} stacks requires further hypotheses. We indicate where it is used (see Remark
\ref{rmk:RH} and Proposition \ref{prop:step1}).

\item \emph{The level is an integer only for compact $G$.} For compact $G$ the Chern--Simons
invariant $\CS(A)$ takes values in $\RR/\ZZ$ and the level is quantized; for $\GC$ the
Chern--Simons functional is complex-valued and there is no such quantization. This distinction
is invisible at the level of formal power series but becomes important in Section
\ref{sec:resurgence}: the Borel singularities of the perturbative expansion sit at differences
$\CS(B) - \CS(A)$ of critical values, and for compact $G$ these all lie on a single line in the
Borel plane, so that the ray structure that drives resurgence degenerates. The resurgence
literature for Chern--Simons theory \cite{Witten2011,GMP2016,GGM2021} accordingly works with the
complex group, and so do we.
\end{enumerate}

\medskip

\noindent\textbf{Conventions.} (i) All derived algebraic geometry is done for $\GC$.
(ii) Statements that concern the compact theory (integrality of the level, unitarity, positivity
in the Verlinde formula) are stated for $G$, and their relation to the $\GC$-picture is stated
explicitly. (iii) Wherever a Riemann--Hilbert-type comparison is required, we say so. (iv) The
dictionary of Section \ref{subsec:BV-BFV-dictionary} carries an explicit
column recording, for each line, whether the two sides are literally identified or only
identified after complexification.

\begin{remark}[The derived-differential-geometric alternative]\label{rmk:DDG}
An alternative route is to work in derived \emph{differential} geometry (derived $C^\infty$-stacks, derived manifolds, or
the $\infty$-categorical formalisms of \cite{Spivak2010,Joyce2012,CarchediSteffens2019}), where
the compact group $G$ and the moduli stack of flat $G$-connections are the natural objects, and where analogues of shifted symplectic geometry are available
\cite{PridhamDDG}. In that setting Chern--Simons theory can be discussed without
complexifying, and one may hope that the analogues of PTVV and of the CPTVV quantization
theorem hold. The obstacle is that the results we need on the algebraic side (BBJ, AKZ, the Etingof--Kazhdan
theorem, and the theory of $\Theta$-stratifications of Section \ref{subsec:theta-strat}) have
not been developed in derived differential geometry. We regard the
derived-differential-geometric formulation of this program as an important open problem
(Section \ref{sec:outlook}) and not as a substitute for the algebraic one.
\end{remark}

\subsection{Structure of the paper}

In Section \ref{sec:BV-BFV}, we give a detailed review of the BV-BFV formalism. We begin with the classical BV formalism (graded symplectic geometry, the classical master equation, and the BRST operator) and the BFV formalism for boundaries (the boundary phase space, BFV charge, and physical state space). We then specialize to Chern--Simons theory, giving the explicit BV fields, the superfield formulation, the BFV boundary data with the Atiyah--Bott symplectic form, and a proof that the BV-BFV axiom is satisfied. We develop the gauge-fixing procedure (Lorenz gauge), the perturbative expansion around flat connections, the structure of Feynman diagram contributions as configuration space integrals, the modified quantum master equation, and the gluing formula that gives the BV-BFV partition function the structure of a symmetric monoidal functor to $\Vect_\hbar$. We also describe the cellular (simplicial) version of the construction.

In Section \ref{sec:RT}, we review the Reshetikhin--Turaev construction. We define ribbon categories and their graphical calculus, modular tensor categories with the modularity condition on the $S$-matrix, and give the quantum group example $\overline{\Rep}_q(\GC)$ with the explicit Kac--Peterson formula for the $S$-matrix. We present the RT-TQFT functor, including the Verlinde formula for state space dimensions, the surgery formula for closed 3-manifolds, invariance under Kirby moves, and the extended TQFT structure assigning the MTC to circles and conformal blocks to punctured surfaces.

In Section \ref{sec:fact-homology}, we develop factorization homology. We define $\EE_n$-operads and algebras, explain the hierarchy from $\EE_0$ (pointed objects) through $\EE_2$ (braided monoidal) to $\EE_\infty$ (commutative), and define factorization homology as a colimit over disk embeddings. We state the $\otimes$-excision theorem, compute factorization homology of $S^1$ as Hochschild homology and of $S^1$ at the categorical level as the Drinfeld center, and state the fundamental theorems of Ben-Zvi--Brochier--Jordan (identifying factorization homology of braided tensor categories with quantum character varieties) and of Ai--Kong--Zheng (identifying the distinguished object of factorization homology of an MTC on a closed surface with the RT state space). We connect this to the Cobordism Hypothesis.

In Section \ref{sec:derived-character}, we introduce derived algebraic geometry (derived stacks, cotangent complexes), define the derived character stack $\Loc_{\GC}(\Sigma)$, and develop the PTVV theory of shifted symplectic structures. We explain the AKSZ mechanism producing these structures, the Lagrangian structure on restriction maps (Calaque), and the precise dictionary between BV-BFV data and derived geometry. We then develop the deformation quantization of shifted symplectic stacks (combining CPTVV, Safronov, and Etingof--Kazhdan) and the BBJ identification of quantized character stacks with factorization homology, distinguishing the two levels of the theory: quantum character varieties at generic $q$ and the collapse to $\Vect$ with a distinguished object at roots of unity.

In Section \ref{sec:bridge}, we state our main conjectures. We construct the $\EE_2$-category from BV-BFV via vertical and horizontal composition, state the $\EE_2$-equivalence conjecture and the Main Conjecture (natural equivalence of functors), and develop the four-step proof strategy: classical identification, quantization comparison (with three detailed sub-steps), factorization homology transport, and cobordism map matching via quantized Lagrangian correspondences.

In Section \ref{sec:koszul} we recall $\EE_n$-Koszul duality and use it to relate the perturbative expansion around a flat connection to the formal geometry of $\Loc_{\GC}(M)$ at that point (Proposition \ref{prop:koszul-classical}). Koszul duality is local, and the global structure comes from the geometry of the stack. We discuss two sources of it: the perverse sheaf of vanishing cycles on the $(-1)$-shifted symplectic stack $\Loc_{\GC}(M)$, which we conjecturally relate to the perturbative expansion and, through a conjecture of Gunningham--Jordan--Safronov, to skein modules (Conjecture \ref{conj:koszul}); and a stratification of the character stack by stabilizer type, modelled on the $\Theta$-stratifications of Halpern-Leistner \cite{HL2015,HL2015b} and illustrated on the solid torus.

In Section \ref{sec:resurgence}, we introduce resurgent series, alien derivatives, and Stokes automorphisms in the Chern--Simons context (for the complex group $\GC$, as explained in Section \ref{subsec:compact-vs-complex}), present the Gukov--Mari\~no--Putrov evidence for the Poincar\'e homology sphere, and relate the picture to the existing literature connecting Stokes phenomena to $t$-structures and stability conditions \cite{BTL2012,BTL2013,KS2008,KS2013,Bridgeland2019,GMN2013,Boalch2014}. We propose that the Chern--Simons functional plays the role of the central charge, so that the Stokes data of Chern--Simons theory is formally a Kontsevich--Soibelman stability datum on a graded Lie algebra, and we formulate the resurgence--Koszul correspondence as a \emph{question}. Since $\CS$ is constant on connected components of the character stack, Stokes data between sectors with different values of $\CS$ are not visible on the stack alone, and the question is formulated in terms of flow lines between flat connections together with the stalks of the vanishing-cycles sheaf.

In Section \ref{sec:evidence}, we test the conjectural picture in the abelian case ($G = U(1)$, where everything is exact), genus zero and one (including the pillowcase computation and $SL_2(\ZZ)$ representations), mapping class group representations (including Andersen's asymptotic faithfulness theorem), lens spaces (with explicit formulas and Jeffrey's verification), and Seifert manifolds.

In Section \ref{sec:outlook}, we discuss the remaining gaps toward a proof, higher-categorical extensions (the fully extended TQFT conjecture), 4-dimensional implications (Crane--Yetter theory and Khovanov homology categorification), and connections to the \emph{quantum} geometric Langlands program via Betti, de Rham, and Dolbeault variants of the character stack.

\medskip
\noindent\textbf{Status of the statements.} Since this paper combines established results
with proposals, we indicate the status of each statement by its label. \emph{Theorem} and
\emph{Proposition} are used only for
statements that are either proved here or proved in a cited reference, with the reference given
in the statement itself. \emph{Conjecture} is used for statements that we believe, that are
precise enough to be true or false, and for which we can describe what a proof would have to
supply. \emph{Question} is used for statements that we suspect are true but cannot at present
formulate precisely enough to call conjectures. Some constructions are explicitly described as
\emph{heuristic} in the text.

\medskip
\textbf{Acknowledgements.} The author thanks P. Mnev, P. Safronov and K. Wernli for valuable
comments and remarks, and O. Gwilliam for a detailed correspondence on the resurgence section
which led directly to Section \ref{subsec:theta-strat}, to the reformulation of Question
\ref{q:resurgence-koszul}, and to the discussion of the r\^ole of the central charge in Section
\ref{subsec:stokes-t-structures}. The author is grateful to the anonymous referee for a careful and detailed report which
improved the paper considerably, in particular the treatment of the compact/complex mismatch in
Section \ref{subsec:compact-vs-complex}.

\section{The BV-BFV Formalism for Chern--Simons Theory}\label{sec:BV-BFV}

\begin{quote}\small
\emph{This section is expository.} Nothing in Sections \ref{sec:BV-BFV}--\ref{sec:derived-character}
is new; these sections summarize a contemporary perspective on four bodies of theory
(BV-BFV quantization, the Reshetikhin--Turaev construction, factorization homology, and shifted
symplectic geometry) in the common notation that Sections
\ref{sec:bridge}--\ref{sec:evidence} require. The reader who knows any of this material should
skip the corresponding section and consult it only for conventions. For the BV formalism we
draw on Batalin--Vilkovisky \cite{BV1981}, Schwarz \cite{Schwarz1993}, Cattaneo--Felder
\cite{CF2001} and Mnev's lectures \cite{Mnev2019}; for the BV-BFV formalism on manifolds with
boundary the sources are Cattaneo--Mnev--Reshetikhin \cite{CMR2014,CMR2018,CMR2020}; for
perturbative Chern--Simons theory, Axelrod--Singer \cite{AS1994}, Kontsevich
\cite{Kontsevich1994}, Bott--Cattaneo \cite{BC1998} and Bar-Natan \cite{BarNatan1995}; and for
the factorization-algebra formulation, Costello--Gwilliam \cite{CG2017,CG2021}. We refer the
interested reader to these sources for more examples and more detail.
\end{quote}

\subsection{Classical BV formalism}

The classical Batalin--Vilkovisky formalism \cite{BV1981,Schwarz1993} provides a systematic cohomological framework for handling gauge symmetry in field theories. Before incorporating boundary structures, we recall the essential ingredients for theories on closed manifolds.\\

The starting point is a $\ZZ$-graded manifold $\mathcal{F}$, called the \emph{space of BV fields}, which collects the physical fields (in degree $0$), the ghosts encoding infinitesimal gauge transformations (in degree $+1$), the antifields serving as sources for BRST variations (in degree $-1$), and, for reducible gauge theories, ghosts-for-ghosts and higher ghosts and their antifields in further degrees. This graded space carries a symplectic form $\omega$ of cohomological degree $-1$, known as the \emph{BV symplectic form} or \emph{antibracket}, which pairs fields with their corresponding antifields. The induced Poisson bracket $\{-,-\}_{\BV}$ has degree $+1$.

\begin{definition}[BV manifold]\label{def:BV-manifold}
A \emph{classical BV manifold} is a quadruple $(\mathcal{F}, \omega, S, Q)$ consisting of a $\ZZ$-graded manifold $\mathcal{F}$ with a degree $(-1)$ symplectic form $\omega$, a degree $0$ function $S \in C^\infty(\mathcal{F})$ called the \emph{BV action}, and a degree $+1$ cohomological vector field\footnote{This is for the finite-dimensional setting, where the odd symplectic form is non-degenerate and the Hamiltonian vector field of $S$ is globally defined. One has to be more careful in the infinite-dimensional setting.} $Q = \{S, -\}_{\BV}$. These are subject to the \emph{classical master equation} (CME):
\begin{equation}\label{eq:CME}
\{S, S\}_{\BV} = 0,
\end{equation}
which implies $Q^2 = 0$ by the graded Jacobi identity.
\end{definition}

The CME encodes the full consistency of the gauge symmetry. It simultaneously contains the gauge invariance of the classical action, the closure of the gauge algebra, and all higher coherence conditions. The resulting cohomology $H^0(Q)$ computes the algebra of classical observables modulo gauge equivalence, while the higher cohomology groups $H^k(Q)$ detect obstructions and redundancies in the gauge structure.

At the quantum level, the BV formalism introduces an additional piece of structure: a second-order differential operator $\Delta_{\BV}$ of degree $+1$, the \emph{BV Laplacian}, satisfying $\Delta_{\BV}^2 = 0$. This operator measures the failure of the antibracket to be a derivation and is related to $\omega$ by the identity \[\{f, g\}_{\BV} = (-1)^{|f|}\Delta_{\BV}(fg) - (-1)^{|f|}\Delta_{\BV}(f)\, g - f\, \Delta_{\BV}(g).\]

\begin{definition}[BV quantization]\label{def:BV-quant}
A \emph{quantum BV theory} is a classical BV manifold equipped with a BV Laplacian\footnote{This is canonically defined for the finite-dimensional case. In the infinite-dimensional case one has to introduce a suitable regularization.} $\Delta_{\BV}$ as above, such that the BV action satisfies the \emph{quantum master equation} (QME):
\begin{equation}\label{eq:QME}
\Delta_{\BV}\, e^{iS/\hbar} = 0 \quad\Longleftrightarrow\quad \{S, S\}_{\BV} - 2i\hbar\, \Delta_{\BV} S = 0.
\end{equation}
\end{definition}

The QME is a quantum deformation of the CME: at $\hbar = 0$ it reduces to $\{S,S\}_{\BV} = 0$, while the correction term $i\hbar\Delta_{\BV} S$ accounts for the anomalies that can arise upon quantization. When the QME is satisfied, one defines the \emph{partition function} by choosing a Lagrangian submanifold $\mathcal{L} \subset \mathcal{F}$. This is the BV formulation of gauge-fixing. When we integrate, we get
\begin{equation}\label{eq:BV-Z-closed}
Z = \int_{\mathcal{L}} e^{iS/\hbar}\, \mu_{\mathcal{L}},
\end{equation}
where $\mu_{\mathcal{L}}$ is the half-density on $\mathcal{L}$ induced by the BV symplectic form. A key consequence of the QME is that $Z$ does not depend on the choice of gauge-fixing Lagrangian $\mathcal{L}$, up to appropriate cohomological equivalence. This gauge-independence is the BV analogue of the BRST Ward identities and is the fundamental reason the formalism produces well-defined invariants.

\subsection{The BFV formalism for boundaries}

When the manifold $M$ has a non-empty boundary $\partial M = \Sigma$, the BV formalism as described above breaks down. The integration-by-parts argument that ensures gauge-independence of the partition function produces boundary contributions that do not vanish. The Batalin--Fradkin--Vilkovisky (BFV) formalism addresses this by equipping the boundary with its own Hamiltonian structure, compatible with the bulk BV data.\\

The idea is that the boundary inherits a phase space from the restriction of the bulk fields, and this phase space carries a degree $0$ symplectic structure (in contrast to the degree $(-1)$ symplectic structure in the bulk). On the boundary, the cohomological data are encoded by a degree $+1$ Hamiltonian on an even symplectic graded manifold.

\begin{definition}[BFV manifold]\label{def:BFV-manifold}
A \emph{BFV manifold} is a quadruple $(\mathcal{F}^\partial, \omega^\partial, S^\partial, Q^\partial)$ where $\mathcal{F}^\partial$ is a $\ZZ$-graded manifold (the \emph{boundary phase space}), $\omega^\partial$ is a symplectic form of degree $0$, $S^\partial \in C^\infty(\mathcal{F}^\partial)$ is a degree $+1$ function called the \emph{BFV charge} satisfying $\{S^\partial, S^\partial\}^\partial = 0$, and $Q^\partial = \{S^\partial, -\}^\partial$ is the resulting cohomological vector field\footnote{We have $(Q^\partial)^2=0$.}.
\end{definition}

The condition $\{S^\partial, S^\partial\}^\partial = 0$ is the BFV analogue of the classical master equation. It ensures that $Q^\partial$ squares to zero and thus defines a cohomology theory on the boundary. Quantization of the boundary data proceeds by a formal analogue of geometric quantization of
the symplectic graded manifold $(\mathcal{F}^\partial, \omega^\partial)$. This produces a graded
vector space $\mathcal{H}_\Sigma$ which can be interpreted as a space of boundary states, but
which does \emph{not} come equipped automatically with a Hilbert space structure: no inner
product, and no completion, is part of the output. The construction thus resembles geometric
quantization without fitting fully within that framework. In the perturbative setting of
\cite{CMR2018} what one obtains is a cochain complex, constructed order by order in $\hbar$, and
we will always mean this when we write $\mathcal{H}_\Sigma$. The BFV charge $S^\partial$
quantizes to a degree $+1$ operator $\hat{\Omega}_\Sigma$ on $\mathcal{H}_\Sigma$ with
$\hat{\Omega}_\Sigma^2 = 0$, and the \emph{physical state space} is its cohomology
$H^\bullet(\mathcal{H}_\Sigma,\hat{\Omega}_\Sigma)$. This is the BFV analogue of BRST
cohomology, and it selects the gauge-invariant states on the boundary.\\

The central construction of the CMR program is the coupling of the bulk BV structure to the boundary BFV structure through a compatibility condition.

\begin{definition}[BV-BFV data]\label{def:BV-BFV}
A \emph{BV-BFV pair} for a manifold $M$ with $\partial M = \Sigma$ consists of bulk BV data $(\mathcal{F}_M, \omega_M, S_M, Q_M)$ and boundary BFV data\\ $(\mathcal{F}_\Sigma^\partial, \omega^\partial_\Sigma=\delta\alpha^\partial_\Sigma, S^\partial_\Sigma, Q^\partial_\Sigma)$, together with a surjective submersion $\pi: \mathcal{F}_M \to \mathcal{F}_\Sigma^\partial$ (restriction of bulk fields to the boundary), subject to the \emph{BV-BFV axiom}:
\begin{equation}\label{eq:BV-BFV-axiom}
\iota_{Q_M} \omega_M = \delta S_M + \pi^* \alpha^\partial_\Sigma.
\end{equation}
An equivalent and often more useful formulation is that the failure of the classical master equation in the bulk is controlled by the pullback of the boundary BFV charge:
\begin{equation}\label{eq:BV-BFV-axiom-2}
\frac{1}{2}\{S_M, S_M\}_{\BV} = \pi^*S^\partial_\Sigma.
\end{equation}
\end{definition}

\begin{figure}[ht]
\centering
\begin{tikzpicture}[scale=1.0]

% === Fill the body as a single closed path ===
\fill[blue!8]
    % start bottom-left, go up the left wall
    (-3,-2) .. controls (-3.2,-0.5) and (-3,1) .. (-2.5,2)
    % across top-left boundary (back arc)
    arc(180:0:0.8 and 0.3)
    % down through crotch
    .. controls (-0.5,0.8) and (0.5,0.8) ..
    % up to top-right boundary
    (0.9,2)
    % across top-right boundary (back arc)
    arc(180:0:0.8 and 0.3)
    % down the right wall
    .. controls (3,1) and (3.2,-0.5) .. (3,-2)
    % across bottom (back arc)
    arc(0:-180:3 and 0.8)
    -- cycle;

% === Draw all outlines on top ===

% Left wall
\draw[blue!60!black, thick]
    (-3,-2) .. controls (-3.2,-0.5) and (-3,1) .. (-2.5,2);
% Right wall
\draw[blue!60!black, thick]
    (3,-2) .. controls (3.2,-0.5) and (3,1) .. (2.5,2);
% Inner crotch
\draw[blue!60!black, thick]
    (-0.9,2) .. controls (-0.5,0.8) and (0.5,0.8) .. (0.9,2);

% Boundary circles
\fill[red!40!black, opacity=0.2] (0,-2) ellipse (3 and 0.8);
\draw[red!60!black, very thick]
    (-3,-2) arc[start angle=180,end angle=360,x radius=3,y radius=0.8];
\draw[red!60!black, very thick, dashed]
    (3,-2) arc[start angle=0,end angle=180,x radius=3,y radius=0.8];

\draw[red!60!black, very thick] (-1.7,2) ellipse (0.8 and 0.3);
\draw[red!40!black, fill, opacity=0.2] (-1.7,2) ellipse (0.8 and 0.3);

\draw[red!60!black, very thick] (1.7,2) ellipse (0.8 and 0.3);
\draw[red!40!black, fill, opacity=0.2] (1.7,2) ellipse (0.8 and 0.3);

% === Labels ===
\node[blue!60!black, font=\large\bfseries] at (2,0) {$M$};

\node[red!60!black, font=\small\bfseries, above] at (-1.7,2.35) {$\Sigma_1$};
\node[red!60!black, font=\small\bfseries, above] at (1.7,2.35) {$\Sigma_2$};
\node[red!60!black, font=\small\bfseries] at (0,-3.15) {$\Sigma_{\mathrm{in}}$};

% BV label box
\node[draw=blue!50, thick, rounded corners=3pt, fill=blue!5,
      inner sep=4pt, font=\scriptsize, align=center] at (-6,0)
      {BV data on $M$:\\[1pt]
       $(\mathcal{F}_M, \omega_M, S_M, Q_M)$\\[1pt]
       deg.\ $-1$ symplectic};
\draw[-{stealth}, thin, blue!40] (-4.3,0) -- (-3.3,0);

% BFV label box
\node[draw=red!50, thick, rounded corners=3pt, fill=red!5,
      inner sep=4pt, font=\scriptsize, align=center] at (6,0)
      {BFV data on $\partial M$:\\[1pt]
       $(\mathcal{F}_\Sigma^\partial, \omega^\partial_\Sigma, S^\partial_\Sigma, Q^\partial_\Sigma)$\\[1pt]
       deg.\ $0$ symplectic};
\draw[-{stealth}, thin, red!40] (4.3,0) -- (3.3,0);

% Restriction map
\node[violet!70!black, font=\scriptsize\bfseries] at (4,1.1) {$\pi$: restrict};
\draw[-{stealth}, thick, violet!60] (3.3,1.4) -- (2.7,1.9);

% === BV-BFV axiom box ===
\node[draw=black!60, thick, rounded corners=5pt, fill=yellow!8,
      inner sep=7pt, font=\small, align=center] at (0,-4.8)
      {BV-BFV axiom:\; $\dfrac{1}{2}\{S_M, S_M\}_{\mathrm{BV}} = \pi^*S^\partial_\Sigma$};
\draw[-{stealth}, thick, black!40] (0,-3.5) -- (0,-4);

\end{tikzpicture}
\caption{The BV-BFV framework, illustrated for a cobordism
$M: \Sigma_{\mathrm{in}} \to \Sigma_1 \sqcup \Sigma_2$. The bulk 3-manifold $M$ (blue shading) carries BV data with a degree $(-1)$ symplectic structure; each boundary component (red circles) carries BFV data with a degree $0$ symplectic structure. The restriction map $\pi$ connects bulk and boundary fields, and the BV-BFV axiom states that the failure of the classical master equation in the bulk is controlled by the BFV charge on the boundary.}
\label{fig:BV-BFV}
\end{figure}
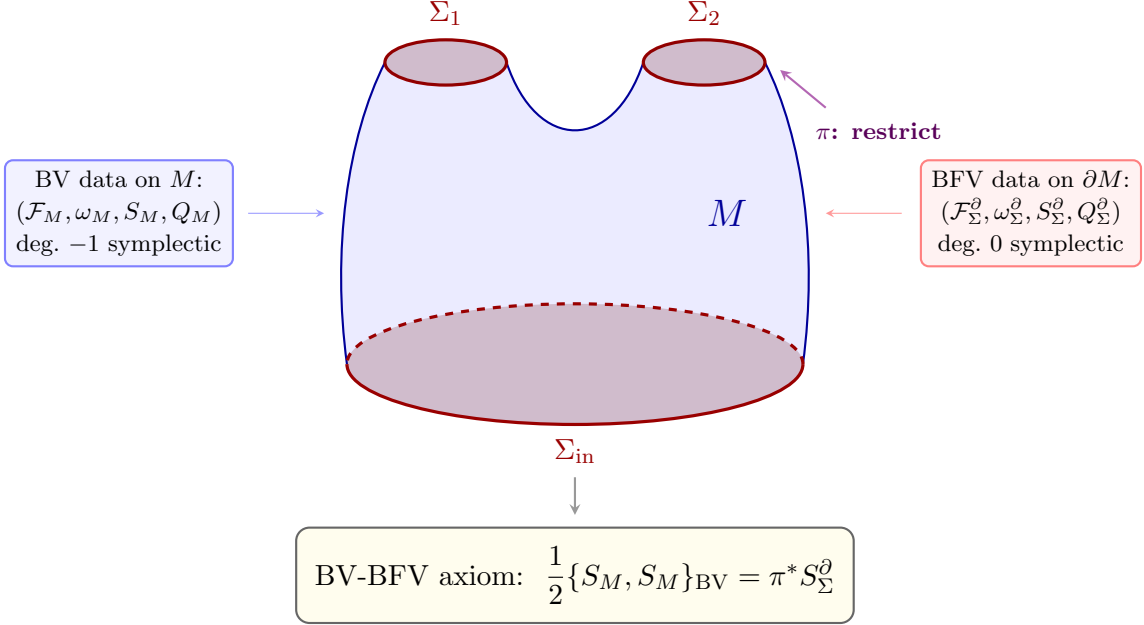

The BV-BFV axiom has a transparent physical meaning. On a closed manifold, the classical master equation $\{S,S\}_{\BV} = 0$ holds because all boundary terms from integration by parts vanish. When the manifold has a boundary, these terms survive, and the axiom \eqref{eq:BV-BFV-axiom-2} states that they are captured by the BFV charge $S^\partial_\Sigma$. This is the cohomological refinement of the familiar variational principle. The variation of the action on a manifold with boundary takes the form $\delta S = (\text{equations of motion}) + \int_\Sigma \theta$, where $\theta$ is the pre-symplectic current on the boundary. The BFV structure organizes this boundary contribution into a consistent Hamiltonian framework, ensuring that the gauge symmetry of the bulk theory is compatible with the symplectic geometry of the boundary phase space.

\subsection{Application to Chern--Simons theory}\label{subsec:CS-BVBFV}

We now specialize the general BV-BFV framework to Chern--Simons theory with a compact simple Lie group $G$ and Lie algebra $\fg$, on an oriented 3-manifold $M$ with boundary $\Sigma = \partial M$. This is the central example for our program, and we describe the BV-BFV data in full detail. The classical Chern--Simons action for $A\in\Omega^1(M,\fg)$ is given by
\[
\mathrm{CS}(A)=\int_M \Tr\left(A\land dA+\tfrac23A\land A\land A\right).
\]

\subsubsection{The BV space of fields}

The key observation underlying the BV formulation of Chern--Simons theory is that the entire field content (physical fields, ghosts, and antifields) can be packaged into a single inhomogeneous differential form. The space of BV fields on $M$ is the shifted de Rham complex
\begin{equation}
\mathcal{F}_M = \Omega^\bullet(M, \fg)[1] = \bigoplus_{k=0}^{3} \Omega^k(M, \fg)[1-k],
\end{equation}
and a general element is written as a \emph{superfield} $\mathbf{A} = c + A + A^+ + c^+$, where the components and their roles are summarized in Table \ref{tab:fields}.

\begin{table}[h]
\centering
\begin{tabular}{c|c|c|c}
\textbf{Field} & \textbf{Form degree} & \textbf{Ghost number} & \textbf{Role} \\ \hline
$c$ & 0 & $+1$ & Ghost (gauge parameter) \\
$A$ & 1 & $0$ & Connection (physical field) \\
$A^+$ & 2 & $-1$ & Antifield of $A$ \\
$c^+$ & 3 & $-2$ & Antifield of $c$
\end{tabular}
%\caption{Component fields of the Chern--Simons BV superfield.}
\label{tab:fields}
\end{table}

The physical field is the connection 1-form $A$; the ghost $c$ is a $\fg$-valued 0-form encoding infinitesimal gauge transformations; the antifields $A^+$ and $c^+$ are differential forms of degrees 2 and 3, respectively, serving as the BV duals of $A$ and $c$. The BV symplectic form of degree $-1$,
\begin{equation}
\omega_M = \frac{1}{2}\int_M \Tr\,(\delta\mathbf{A} \wedge \delta\mathbf{A})=\int_M\Tr\,(\delta A\wedge \delta A^++\delta c\wedge \delta c^+),
\end{equation}
pairs each field with its antifield via the trace form on $\fg$ and the wedge product on $M$.\\

In the superfield formulation, the BV action has the same functional form as the classical Chern--Simons action, applied to the full superfield:
\begin{equation}\label{eq:CS-BV-action}
S_M^{\BV} = \int_M \Tr\left(\mathbf{A} \wedge d\mathbf{A} + \tfrac{2}{3}\, \mathbf{A} \wedge \mathbf{A} \wedge \mathbf{A}\right).
\end{equation}
Expanding this in terms of the component fields and sorting by ghost number, one obtains:
\begin{equation}\label{eq:CS-BV-action-expanded}
S_M^{\BV} = \int_M \Tr\Big( A \wedge dA + \tfrac{2}{3}\, A \wedge A \wedge A \nonumber+ \tfrac12 A^+ \wedge d_A c +\tfrac12 c\wedge d_AA^++ c^+ \wedge \tfrac{1}{2}[c, c] \Big),
\end{equation}
where $d_A = d + [A, -]$ is the covariant derivative. The first line is the classical Chern--Simons action functional; the second line encodes the gauge structure. The term $A^+ \wedge d_A c$ couples the antifield of the connection to the linearized gauge transformation and $c^+ \wedge \frac{1}{2}[c,c]$ encodes the closure of the gauge algebra (the structure constants of $\fg$).

\subsubsection{The BFV boundary data}

Restricting the bulk fields to the boundary\\ $\Sigma = \partial M$ produces the BFV phase space. Since $\Sigma$ is a surface, the shifted de Rham complex truncates to forms of degree at most 2, giving
\begin{equation}\label{eq:CS-BFV-fields}
\mathcal{F}_\Sigma^\partial = \Omega^\bullet(\Sigma, \fg)[1],
\end{equation}
with boundary fields $(\mathbb{c}, \mathbb{A}, \mathbb{A}^+)$ in ghost-number degrees $(+1, 0, -1)$, obtained by restricting the bulk superfield $\mathbf{A}|_\Sigma = \mathbb{c} + \mathbb{A} + \mathbb{A}^+$. Here $\mathbb{A}$ is the boundary connection, $\mathbb{c}$ is the boundary ghost, and $\mathbb{A}^+$ is the boundary momentum (the antifield conjugate to the curvature constraint).

The BFV symplectic form of degree $0$ is the Atiyah--Bott form, extended to include the ghost--momentum pairing:
\begin{equation}\label{eq:AB-symplectic}
\omega^\partial_\Sigma =\frac12\int_\Sigma\Tr(\delta\mathbf{A}\vert_\Sigma\land \delta\mathbf{A}\vert_\Sigma)=\int_\Sigma \Tr\,(\delta\mathbb{A} \wedge \delta\mathbb{A}) + \int_\Sigma \Tr\,(\delta\mathbb{c} \wedge \delta\mathbb{A}^+).
\end{equation}
The first term is the standard symplectic structure on the space of connections on $\Sigma$, familiar from the work of Atiyah and Bott; the second is the canonical cotangent pairing between the ghost and its conjugate momentum. The BFV charge that generates both the constraint and the gauge symmetry is:
\begin{equation}\label{eq:BFV-charge}
S^\partial_\Sigma = \int_\Sigma \Tr\left(\mathbb{c}\wedge F_\mathbb{A}+\tfrac12\mathbb{A}^+\wedge[\mathbb{c},\mathbb{c}]\big)\right),
\end{equation}
where $F_{\mathbb{A}} = d\mathbb{A} + \tfrac12\mathbb{A} \wedge \mathbb{A}$ is the curvature of the boundary connection. On the constraint surface $\mathbb{A}^+ = 0$, the condition $S^\partial_\Sigma = 0$ imposes nothing, but the Hamiltonian flow of $S^\partial_\Sigma$ generates the flatness constraint $F_{\mathbb{A}} = 0$ together with the infinitesimal gauge transformations on all boundary fields.

\begin{proposition}\label{prop:CS-BV-BFV-verified}
The Chern--Simons BV-BFV data defined above satisfies the BV-BFV axiom \eqref{eq:BV-BFV-axiom}.
\end{proposition}

\begin{proof}
We prove the BV--BFV identity in the form
\[
\iota_{Q_M}\omega_M=\delta S_M^{\BV}+\pi^*\alpha^\partial_\Sigma,
\]
from which it follows that
\[
\frac12\{S_M^{\BV},S_M^{\BV}\}_{\BV}=\pi^*S^\partial_\Sigma,
\]
where $S^\partial_\Sigma$ is the BFV charge on the boundary. Let us consider the BV action\footnote{Throughout we use $\tfrac16\Tr(A\wedge[A,A]) = \Tr(A\wedge A\wedge A)$, which follows from invariance of $\Tr$ and graded antisymmetry of the wedge product.}
\[
S_M^{\BV}
=
\int_M \Tr\!\left(
\tfrac12 \mathbf A\wedge d\mathbf A+\tfrac16 \mathbf A\wedge[\mathbf A,\mathbf A]
\right),
\]
where $\mathbf A=c+A+A^++c^+ \in \Omega^\bullet(M,\fg)[1]$ is the BV superfield. Its variation is
\[
\delta S_M^{\BV}
=
\int_M \Tr\left(
\delta\mathbf A\wedge
\bigl(d\mathbf A+\tfrac12[\mathbf A,\mathbf A]\bigr)
\right)
+
\int_{\partial M}\Tr\left(\tfrac12\,\mathbf A\wedge\delta\mathbf A\right).
\]
Thus, the bulk Euler--Lagrange expression is
\[
d\mathbf A+\frac12[\mathbf A,\mathbf A],
\]
and the boundary term is the pullback of the boundary one-form
\[
\alpha^\partial_\Sigma=
\frac12\int_\Sigma \Tr\!\left(\mathbf A\vert_\Sigma\wedge\delta\mathbf A\vert_\Sigma\right).
\]
The bulk BV symplectic form is
\[
\omega_M
=
\int_M \Tr\bigl(\delta A\wedge\delta A^+ + \delta c\wedge\delta c^+\bigr),
\]
or, equivalently, in superfield notation,
\[
\omega_M=\frac12\int_M \Tr(\delta\mathbf A\land\delta\mathbf A).
\]
Hence, the Hamiltonian vector field $Q_M$ generated by $S_M^{\BV}$ is determined by
\[
\iota_{Q_M}\omega_M=\delta S_M^{\BV}+\pi^*\alpha^\partial_\Sigma,
\]
that is,
\[
Q_M \mathbf A = d\mathbf A+\frac12[\mathbf A,\mathbf A].
\]
Restricting to the boundary gives the induced cohomological vector field
\[
Q^\partial_\Sigma \mathbf A\vert_\Sigma
=
d\mathbf A\vert_\Sigma+\frac12\Big[\mathbf A\vert_\Sigma,\mathbf A\vert_\Sigma\Big].
\]

Writing
\[
\mathbf A\vert_\Sigma=\mathbb c+\mathbb A+\mathbb A^+,
\]
with
\[
\mathbb c\in\Omega^0(\Sigma,\fg),\qquad
\mathbb A\in\Omega^1(\Sigma,\fg),\qquad
\mathbb A^+\in\Omega^2(\Sigma,\fg),
\]
one obtains
\[
Q^\partial_\Sigma \mathbb A = d_{\mathbb A}\mathbb c,
\qquad
Q^\partial_\Sigma \mathbb c = \frac12[\mathbb c,\mathbb c],
\qquad
Q^\partial_\Sigma \mathbb A^+ = F_{\mathbb A}+[\mathbb A^+,\mathbb c],
\]
up to the overall sign convention for the BV/BFV bracket.

The boundary BFV symplectic form is
\[
\omega^\partial_\Sigma
=
\int_\Sigma \Tr\bigl(\delta\mathbb A\land\delta\mathbb A
+\delta\mathbb c\land\delta\mathbb A^+\bigr),
\]
and the corresponding BFV charge is the Hamiltonian generating
$Q^\partial_\Sigma$, namely
\[
S^\partial_\Sigma
=
\int_\Sigma \Tr\!\left(
\mathbb c\land F_{\mathbb A}
+\tfrac12\,\mathbb A^+\land[\mathbb c,\mathbb c]
\right),
\]
again up to the global sign convention.

Therefore, the boundary defect of the bulk BV action is exactly the pullback of the BFV data on $\Sigma$, and the BV--BFV compatibility identity holds:
\[
\iota_{Q_M}\omega_M=\delta S_M^{\BV}+\pi^*\alpha^\partial_\Sigma.
\]
Equivalently, we get
\[
\frac12\{S_M^{\BV},S_M^{\BV}\}_{\BV}=\pi^*S^\partial_\Sigma.
\]
This proves the claim.
\end{proof}

\subsection{Gauge-fixing and the perturbative expansion}

To extract concrete numbers from the BV partition function \eqref{eq:BV-Z-closed}, one must choose a gauge-fixing Lagrangian $\mathcal{L} \subset \mathcal{F}_M$. Although Chern--Simons theory is topological, the gauge-fixing procedure requires an auxiliary Riemannian metric $g$ on $M$, which will not appear in the final answer.

The perturbative approach proceeds by expanding around a classical solution (a flat connection $A_0$ satisfying $F_{A_0} = 0$) and organizing the resulting Feynman diagram expansion as a formal power series in $\hbar$.

\begin{construction}[Lorenz gauge-fixing]\label{constr:gauge-fix}
To define perturbation theory around a flat background connection \(A_0\), one sets
\[
A=A_0+a,
\qquad
F_{A_0}=0,
\]
where \(a\in \Omega^1(M,\fg)\) is the fluctuation field. Explicit gauge-fixing requires passing from the minimal BV space of fields to the non-minimal one by adjoining an antighost \(\bar c\), a Nakanishi--Lautrup field \(b\), and their antifields:
\[
\bar c \in \Omega^0(M,\fg), \qquad \mathrm{gh}(\bar c)=-1,
\]
\[
b \in \Omega^0(M,\fg), \qquad \mathrm{gh}(b)=0.
\]
The non-minimal action is
\[
S_{\mathrm{nm}} = S_M^{\BV} + \int_M \Tr(\bar c^+\wedge b).
\]

Fix an auxiliary Riemannian metric on \(M\). The Lorenz gauge condition
\[
d_{A_0}^* a = 0
\]
is implemented by the gauge-fixing fermion
\[
\Psi = \int_M \Tr(\bar c\wedge d_{A_0}^* a).
\]
The gauge-fixed Lagrangian is defined by the usual BV prescription
\[
\Phi^+ = \frac{\delta \Psi}{\delta \Phi}
\]
for all fields \(\Phi\). In particular, we have
\[
A^+ = \frac{\delta \Psi}{\delta a}, \qquad
\bar c^+ = \frac{\delta \Psi}{\delta \bar c}= d_{A_0}^* a,
\qquad
c^+=0,
\qquad
b^+=0.
\]
\end{construction}

Substituting these conditions into the non-minimal action yields the usual Lorenz-gauge-fixed action. Up to cubic order in the fluctuations one obtains
\begin{multline}
\label{eq:gauge-fixed-CS}
S_{\mathrm{gf}}
=
\mathrm{CS}(A_0)
+\frac12 \int_M \Tr(a \wedge d_{A_0} a)
+\int_M \Tr\bigl(b\wedge d_{A_0}^* a\bigr)\\
+\int_M \Tr\bigl(\bar c\wedge \Delta_{A_0} c\bigr)
+\frac16 \int_M \Tr\bigl(a \wedge [a,a]\bigr)
+\int_M \Tr\bigl(\bar c\wedge[a,c]\bigr),
\end{multline}
where
\[
\Delta_{A_0}=d_{A_0}d_{A_0}^*+d_{A_0}^*d_{A_0}
\]
is the covariant Laplacian. Integrating out \(b\) imposes the Lorenz gauge condition
\[
d_{A_0}^* a=0.
\]
The quadratic part determines the propagators, while the cubic terms determine the interaction vertices. The first term is the classical Chern--Simons invariant of the flat connection. The ghost propagator is governed by the Green's operator of \(\Delta_{A_0}\), whereas the gauge-field propagator arises from the first-order Chern--Simons kinetic operator together with the gauge condition. The interaction vertices are determined by the cubic terms in \eqref{eq:gauge-fixed-CS}. Higher-order terms in $\hbar$ arise from iterating this cubic vertex in Feynman diagrams.

In the perturbative BV--BFV formalism, one must in general also account for residual fields arising from zero modes. We suppress these technicalities here and state only the standard formal structure of the perturbative expansion around an acyclic flat background. Carrying out the Gaussian integral over the quadratic part and summing the perturbative corrections gives the following result.

\begin{proposition}[Structure of the perturbative expansion]\label{prop:pert-Z}
Formally, the perturbative partition function of Chern--Simons theory around a flat connection \(A_0\) has the form
\begin{equation}
\label{eq:pert-CS-structure}
Z^{\mathrm{pert}}_{A_0}(M)
=
e^{ik\,\mathrm{CS}(A_0)}
\cdot \tau_{A_0}(M)^{1/2}
\cdot e^{\frac{i\pi}{4}\eta(A_0)}
\cdot
\exp\!\left(\sum_{n\ge 1}\hbar^n I_n(A_0,M)\right),
\end{equation}
where \(\tau_{A_0}(M)\) is the Ray--Singer torsion \cite{RS1971}, \(\eta(A_0)\) is the eta-invariant of the relevant odd-signature operator, and the coefficients \(I_n(A_0,M)\) are given by sums of compactified configuration-space integrals associated to connected trivalent Feynman diagrams.
\end{proposition}

Equation \eqref{eq:pert-CS-structure} should be understood as the standard formal structure of the perturbative expansion rather than as a literal consequence of the naive finite-dimensional Gaussian integral.

\begin{remark}[Structure of the Feynman diagrams]\label{rmk:Feynman}
Since the only interaction vertex is cubic (trivalent), the Feynman diagrams contributing to $I_n$ are trivalent graphs $\Gamma$ with $2n$ vertices and $3n$ edges, at loop order $n+1$. The contribution of each graph is not an ordinary integral but a \emph{configuration space integral}: one integrates a product of propagator forms over a compactified configuration space of points in $M$ labeled by the vertices of $\Gamma$. Explicitly, we have
\begin{equation}\label{eq:Feynman-integral}
I_\Gamma(A_0, M) = \int_{C_\Gamma(M)} \bigwedge_{e \in E(\Gamma)} \omega_e(x_{s(e)},\, x_{t(e)}),
\end{equation}
where $C_\Gamma(M)$ is the Axelrod--Singer compactification of the configuration space of $|V(\Gamma)|$ distinct points in $M$, and $\omega_e$ is a propagator form (the integral kernel of the Green's operator of $\Delta_{A_0}$) associated to the edge $e$ connecting the source vertex $s(e)$ to the target vertex $t(e)$. The compactification is essential: it resolves the singularities that arise when two or more points collide, rendering the integral well-defined. This regularization procedure, introduced by Axelrod--Singer \cite{AS1994} and independently by Kontsevich \cite{Kontsevich1994}, is one of the technical cornerstones of perturbative Chern--Simons theory (see also \cite{CMR2014b}).
\end{remark}
\subsection{The modified quantum master equation}

\begin{proposition}[CMR {\cite[Thm.~2.4.1 and \S 3]{CMR2018}}]\label{prop:CMR-mQME}
The effective BV action $S_M^{\mathrm{eff}}$ satisfies the modified quantum master equation:
\begin{equation}\label{eq:mQME}
(\hbar^2 \Delta_M + \hat{\Omega}_\Sigma)\, e^{iS_M^{\mathrm{eff}}/\hbar} = 0,
\end{equation}
where $\Delta_M$ is the BV Laplacian on residual fields and $\hat{\Omega}_\Sigma$ is the quantized BFV operator.
\end{proposition}

\begin{remark}\label{rmk:mQME-meaning}
The mQME says that $\Zbv(M) = e^{iS_M^{\mathrm{eff}}/\hbar}$ is a cocycle in the BFV cohomology of the boundary, i.e., $\Zbv(M) \in H^\bullet(\hat{\Omega}_\Sigma) = \mathcal{H}_\Sigma^{\mathrm{phys}}$. The partition function is a physical state associated to the boundary.
\end{remark}

\subsection{Functoriality and gluing}\label{subsec:gluing}

A central virtue of the BV-BFV formalism is that it is designed from the outset to be compatible with cutting and gluing of manifolds. This is what gives the partition function the structure of a topological quantum field theory, and it is the property that will allow us to compare it with the RT functor.

Suppose a closed 3-manifold $M$ is decomposed as $M = M_1 \cup_\Sigma M_2$ by cutting along an embedded surface $\Sigma$. Each piece $M_i$ is a cobordism with $\Sigma$ as a boundary component, and the BV-BFV formalism assigns to each a partition function $\Zbv(M_i)$ living in the BFV state space $\mathcal{H}_\Sigma$ of the cutting surface. The partition function of the glued manifold is then recovered by pairing these two states.

\begin{theorem}[CMR gluing; {\cite[\S 3.5]{CMR2018}} and {\cite[\S 6]{CMR2020}}]\label{thm:BV-gluing}
Let $M = M_1 \cup_\Sigma M_2$ be a decomposition along a common boundary component $\Sigma$. Then the BV-BFV partition function satisfies the gluing formula
\begin{equation}\label{eq:BV-gluing}
\Zbv(M) = \langle \Zbv(M_1),\, \Zbv(M_2) \rangle_{\mathcal{H}_\Sigma},
\end{equation}
where $\langle -,\, - \rangle_{\mathcal{H}_\Sigma}$ is the pairing on the BFV state space induced by the symplectic structure on the boundary. At the level of the path integral, this pairing is realized as a fiber BV integral over the boundary fields:
\begin{equation}\label{eq:BV-gluing-integral}
Z_{\BV}(M_1 \cup_\Sigma M_2) = \int_{\mathcal{L}_\Sigma} Z_{\BV}(M_1) \cdot Z_{\BV}(M_2)\;\mu_\Sigma,
\end{equation}
where $\mathcal{L}_\Sigma$ is a Lagrangian submanifold in the BFV phase space of $\Sigma$ and $\mu_\Sigma$ is the induced measure.
\end{theorem}

The gluing formula endows the BV-BFV partition function with the algebraic structure of a TQFT. Surfaces are assigned state spaces, cobordisms are assigned vectors (or linear maps) in those state spaces, and composition of cobordisms corresponds to the pairing \eqref{eq:BV-gluing}.

\begin{corollary}\label{cor:BV-functor}
The BV-BFV partition function defines a symmetric monoidal functor
\begin{equation}
\Zbv: \Bord_3^{\mathrm{or}} \longrightarrow \Vect_\hbar,
\end{equation}
where $\Vect_\hbar = \Vect_{\CC[\![\hbar]\!]}$ denotes the category of topological vector spaces over the ring of formal power series in $\hbar$. This is an \emph{ordinary} (non-extended) TQFT functor. It assigns vector spaces to surfaces and formal power series to closed 3-manifolds, but does not assign categorical data to circles. One of the goals of our program is to promote this to a $(3\text{-}2\text{-}1)$-extended functor by extracting the $\EE_2$-category $\mathcal{B}_{\CS}^{(k)}$ at the circle level (Section \ref{sec:bridge}).
\end{corollary}

\begin{remark}[Perturbative nature]\label{rmk:perturbative}
The functor $\Zbv$ is \emph{perturbative}: its values are formal power series in $\hbar$, not complex numbers. The partition function of a closed 3-manifold is an element of $\CC[\![\hbar]\!]$, and the state space of a surface is a module over this ring. In particular, $\Zbv$ cannot directly ``see'' non-perturbative phenomena such as tunneling between distinct flat connections. The question motivating this paper is how to \emph{de-formalize} this functor, i.e., how to pass from the formal power series ring $\CC[\![\hbar]\!]$ to $\CC$ by specializing $\hbar$ to a concrete value, and whether the result agrees with the RT invariants.
\end{remark}

\subsection{The BV-BFV formalism on simplicial decompositions}

An important refinement of the CMR program, developed in \cite{CMR2020} (building on the discrete BF-theory constructions of \cite{Mnev2009}), is the construction of the BV-BFV theory on simplicial (or more generally CW) decompositions of $M$. This ``cellular'' version of the theory replaces the infinite-dimensional spaces of differential forms with finite-dimensional combinatorial data, making the partition function concretely computable while preserving the essential algebraic structure.

\begin{construction}[Cellular BV-BFV; {\cite[\S 4--5]{CMR2020}}]\label{constr:cellular}
Let $\Delta$ be a triangulation of $M$. The cellular BV-BFV construction assigns a finite-dimensional BV-BFV theory to each simplex: the BV space associated to a 3-simplex is $\fg[1] \oplus \fg \oplus \fg[-1] \oplus \fg[-2]$ (one copy of the Lie algebra in each degree, mirroring the four components of the superfield), and the BFV space associated to a 2-simplex is $\fg[1] \oplus \fg \oplus \fg[-1]$ (the boundary truncation). The partition function of $M$ is then computed by assembling the simplices: one performs an iterated BV pushforward, gluing the finite-dimensional BV-BFV data of adjacent simplices step by step according to the combinatorics of $\Delta$.
\end{construction}

The advantage of this approach is computability. Each step involves a finite-dimensional integral rather than a functional integral, and the entire computation reduces to a sequence of linear-algebraic operations on copies of $\fg$. The price one pays is an apparent dependence on the triangulation $\Delta$. However, the partition function is in fact independent of this choice. This follows from a cellular version of the modified quantum master equation, together with the observation that any two triangulations of $M$ are related by a finite sequence of Pachner moves (bistellar flips), and the BV-BFV partition function is invariant under each such move. This cellular construction will play an important role in Section \ref{sec:bridge}, where we compare it with the factorization homology approach to computing topological invariants. See also \cite{CMR2014b} for the semiclassical perspective on this construction and \cite{CG2017,CG2021,GwilliamHaugseng2018} for the factorization algebra viewpoint.

\section{The Reshetikhin--Turaev Construction}\label{sec:RT}

We now turn to the other side of the bridge, namely, the \emph{Reshetikhin--Turaev invariants}. These are constructed purely algebraically, starting from a modular tensor category, and produce a fully functorial $(2+1)$-dimensional TQFT. Our goal in this section is to describe the construction in enough detail to make the comparison with the BV-BFV formalism precise.

\begin{quote}\small
\emph{This section is expository.} The original sources are Reshetikhin--Turaev
\cite{RT1990,RT1991} and Turaev's book \cite{Turaev1994}; systematic modern treatments are
Bakalov--Kirillov \cite{BK2001} and Etingof--Gelaki--Nikshych--Ostrik \cite{EGNO2015}, and we
follow their conventions except where stated. For the quantum group input at roots of unity we
follow Lusztig \cite{Lusztig1990}, Andersen \cite{AndersenH1992} and Chari--Pressley
\cite{CP1994}; for semisimplification as a general categorical operation, Etingof--Ostrik
\cite{EO2021}. The reader who wants more examples or more detail should go to these sources; we
record only what Sections \ref{sec:bridge}--\ref{sec:evidence} use.
\end{quote}

\subsection{Ribbon categories and graphical calculus}

The algebraic input for the RT construction is a \emph{ribbon category}, i.e., a monoidal category with enough additional structure (braiding, twist, duality) to support a topological calculus for morphisms.

\begin{definition}[Ribbon category]\label{def:ribbon}
A \emph{ribbon category} is a monoidal category $\MTC$ equipped with:
\begin{enumerate}[label=(\roman*)]
\item A \emph{braiding}, i.e., a natural family of isomorphisms\\ $c_{V,W}: V \otimes W \xrightarrow{\sim} W \otimes V$ satisfying the hexagon axioms, which are the categorical incarnation of the Yang--Baxter equation.
\item A \emph{twist} (or \emph{balancing}), i.e., a natural automorphism\\ $\theta_V: V \xrightarrow{\sim} V$ compatible with the braiding in the sense that\\ $\theta_{V \otimes W} = c_{W,V}\, c_{V,W}\, (\theta_V \otimes \theta_W)$.
\item \emph{Rigidity}: every object $V$ possesses a dual $V^*$, together with evaluation $\mathrm{ev}_V: V^* \otimes V \to \mathbf{1}$ and coevaluation\\ $\mathrm{coev}_V: \mathbf{1} \to V \otimes V^*$ morphisms satisfying the usual zigzag identities, and these are compatible with the braiding and twist.
\end{enumerate}
\end{definition}

Ribbon categories have a \emph{graphical calculus}. Morphisms in $\MTC$ can be represented as ribbon tangles in the strip $\RR^2 \times [0,1]$. Objects label the strands (oriented ribbons), the braiding corresponds to one strand crossing over another, the twist corresponds to a full $2\pi$-rotation of a ribbon, and duality is represented by cups (coevaluation) and caps (evaluation). Isotopic ribbon tangles represent the same morphism, so the algebraic relations in $\MTC$ are faithfully encoded by the topology of the tangles. This calculus is what ultimately connects the algebra to 3-dimensional topology.

A ribbon category comes with natural notions of dimension and trace that generalize the ordinary ones.

\begin{definition}[Quantum dimension and trace]\label{def:qdim}
For an object $V$ in a ribbon category, the \emph{quantum trace} of an endomorphism $f: V \to V$ is defined using the ribbon structure as \[\Tr_q(f) = \mathrm{ev}_V \circ c_{V,V^*} \circ (\theta_V f \otimes \id_{V^*}) \circ \mathrm{coev}_V,\]
where $\theta_V$ is the twist. The \emph{quantum dimension} of $V$ is the quantum trace of the identity: $\dim_q(V) = \Tr_q(\id_V)$. In the graphical calculus, this is computed by closing a strand labeled $V$ into a loop with a single twist. The resulting ``framed unknot'' evaluates to $\dim_q(V)$. The \emph{total quantum dimension} of the category is $\mathcal{D}^2 = \sum_{i \in \mathrm{Irr}(\MTC)} (\dim_q V_i)^2$, where the sum runs over isomorphism classes of simple objects.
\end{definition}

The quantum dimension need not be a positive integer. It is in general a complex number, and for quantum groups at roots of unity it is an algebraic number that can be negative or even zero for certain objects.

\subsection{Modular tensor categories}

We fix notation for the representation categories used throughout the paper.

For $\fg$ a simple Lie algebra with associated simply connected complex group $\GC$ and $q$ not
a root of unity, $\Rep_q(\GC)$ denotes the braided monoidal category of finite-dimensional
type-$\mathbf{1}$ representations of $U_q(\fg)$. At generic $q$ this category is semisimple with
simple objects indexed by dominant integral weights, and it is braided but not finite. At a root
of unity $q = e^{i\pi/(k+h^\vee)}$ semisimplicity fails, and the relevant category for the RT construction is the full subcategory
\[
\mathrm{Tilt}_q(\GC) \;\subset\; \Rep_q(\GC)
\]
of \emph{tilting modules}, that is, modules admitting both a Weyl filtration and a dual Weyl
filtration \cite{Lusztig1990,AndersenH1992,CP1994}. This subcategory is closed under tensor
products, duals and direct summands, and inherits the ribbon structure; it is neither
semisimple nor abelian, but it is the natural setting for the semisimplification below, since on
tilting modules the quantum trace is compatible with the tensor structure. At generic $q$ every
finite-dimensional module is tilting, so the two categories agree.

We recall some standard terminology. A
\emph{tensor ideal} in a monoidal category $\mathcal{A}$ is a collection $\mathcal{N}(X,Y)
\subseteq \Hom_{\mathcal{A}}(X,Y)$ of subgroups, one for each pair of objects, closed under
composition with arbitrary morphisms on either side and under tensoring with arbitrary
morphisms. Given such an ideal, the \emph{quotient category} $\mathcal{A}/\mathcal{N}$ has the
same objects as $\mathcal{A}$ and morphism groups $\Hom_{\mathcal{A}}(X,Y)/\mathcal{N}(X,Y)$;
it is again monoidal, and inherits a braiding and ribbon structure if $\mathcal{A}$ had one. A
morphism $f: X \to X$ in a ribbon category is \emph{negligible} if $\Tr_q(f \circ g) = 0$ for
every $g: X \to X$, and negligible morphisms form a tensor ideal $\mathcal{N}$. A linear
category is \emph{finite} if it has finitely many isomorphism classes of simple objects,
finite-dimensional $\Hom$ spaces, enough projectives, and every object of finite length;
it is \emph{semisimple} if every object is a direct sum of simple ones, equivalently if every
short exact sequence splits.

With this vocabulary, the \emph{semisimplification} (or \emph{modularization}) is
\[
\overline{\Rep}_q(\GC) \;:=\; \mathrm{Tilt}_q(\GC)\,/\,\mathcal{N},
\]
which is a finite semisimple ribbon category, and is moreover modular (Definition
\ref{def:MTC}); see \cite{EO2021} for semisimplification as a general operation on tensor
categories and \cite{BK2001,Turaev1994} for this case. We use the bar notation
$\overline{(-)}$ consistently to denote this semisimplification, and we use the same formula
\emph{throughout the paper}: at generic $q$, where $\mathrm{Tilt}_q = \Rep_q$, it reads
$\overline{\Rep}_q(\GC) = \Rep_q(\GC)/\mathcal{N}$, and at a root of unity the tilting
subcategory must be taken first. When the distinction matters, we always specify whether we
work at generic $q$, at a root of unity with the non-semisimplified category, or with the
semisimplified MTC.

Not every ribbon category is suitable for constructing 3-manifold invariants. The additional condition needed is \emph{modularity}, i.e., a non-degeneracy requirement on the braiding that ensures the category carries enough information to detect all topological features of 3-manifolds.

To state this precisely, recall that a \emph{fusion category} is a semisimple $\CC$-linear rigid monoidal category with finitely many isomorphism classes of simple objects (including the unit object, which is required to be simple) and finite-dimensional Hom spaces. The tensor product of simple objects decomposes into simples according to the \emph{fusion rules} $V_i \otimes V_j \cong \bigoplus_k N_{ij}^k\, V_k$, where the multiplicities $N_{ij}^k \in \ZZ_{\geq 0}$ encode the multiplicative structure of the Grothendieck ring.

\begin{definition}[Modular tensor category]\label{def:MTC}
A \emph{modular tensor category} (MTC) is a ribbon fusion category $\MTC$ that is \emph{non-degenerate} in the following sense. The $S$-matrix, defined by
\begin{equation}\label{eq:S-matrix}
S_{ij} = \frac{1}{\mathcal{D}}\, \Tr_q(c_{V_j, V_i} \circ c_{V_i, V_j}),
\end{equation}
where $\{V_i\}_{i \in I}$ are the simple objects and $\mathcal{D}$ is the total quantum dimension, is invertible.
\end{definition}

In the graphical calculus, the entry $S_{ij}$ is computed by evaluating the Hopf link (two linked unknots, one colored by $V_i$ and the other by $V_j$) and normalizing by $\mathcal{D}$. The modularity condition (invertibility of $S$) says that the braiding is ``maximally non-symmetric'', i.e., no non-trivial object braids trivially with every other object. Physically, it corresponds to the absence of ``transparent'' anyons in the associated topological phase of matter. Mathematically, it is what makes the surgery formula for 3-manifold invariants well-defined. The invertibility of $S$ is needed to show invariance under Kirby moves.

The prototypical and motivating family of examples comes from quantum groups at roots of unity.

\begin{example}[Quantum group MTC]\label{ex:quantum-MTC}
Let $\fg$ be a simple Lie algebra of a compact group $G$ with complexification $\GC$, let $k$ be a positive integer (the \emph{level}), and set $q = e^{i\pi/(k+h^\vee)}$, where $h^\vee$ is the \emph{dual Coxeter number}. Applying the construction of the previous subsection verbatim --- pass to tilting modules, then quotient by negligible morphisms --- gives
\begin{equation}
\MTC \;=\; \overline{\Rep}_q(\GC) \;=\; \mathrm{Tilt}_q(\GC)\, /\, \mathcal{N}.
\end{equation}
The resulting category is a semisimple ribbon category whose simple objects are indexed by the dominant integral weights $\lambda$ lying in the \emph{Weyl alcove}
\begin{equation}
\mathcal{A}_k = \big\{\lambda \in \Lambda^+_{\mathrm{dom}} \;:\; 0 \leq \langle\lambda, \theta\rangle \leq k\big\},
\end{equation}
where $\theta$ is the highest root and $\Lambda^+_{\mathrm{dom}}$ is the set of dominant integral weights. This is a finite set, so $\overline{\Rep}_q(\GC)$ is a fusion category. It is moreover modular. Indeed, the $S$-matrix is given explicitly by the Kac--Peterson formula
\begin{equation}\label{eq:KP-formula}
S_{\lambda\mu} = \frac{i^{|\Delta_+|}}{(k+h^\vee)^{\rk(\fg)/2}} \cdot \frac{\displaystyle\sum_{w \in W} (-1)^{\ell(w)}\, e^{-2\pi i \langle w(\lambda + \rho),\, \mu + \rho\rangle / (k+h^\vee)}}{\displaystyle\prod_{\alpha \in \Delta_+} 2\sin\!\big(\pi \langle\alpha, \rho\rangle / (k+h^\vee)\big)},
\end{equation}
where $W$ is the Weyl group, $\rho$ is the Weyl vector (half-sum of positive roots), $\Delta_+$ is the set of positive roots, and $\ell(w)$ is the length function on $W$. The invertibility of this matrix is a classical result in the representation theory of affine Lie algebras.
\end{example}

For $\fg = \mathfrak{sl}_2$ at level $k$, the Weyl alcove contains $k+1$ weights (corresponding to spins $j = 0, \tfrac{1}{2}, 1, \ldots, \tfrac{k}{2}$), and the $S$-matrix reduces to $S_{jj'} = \sqrt{2/(k+2)}\;\sin\!\big((j+1)(j'+1)\pi/(k+2)\big)$, which is manifestly invertible. This is the simplest non-trivial example and the one most frequently encountered in the literature. We note that the RT construction has been extended to the non-semisimple setting in \cite{DGGPR2022}, which is relevant to our program since the non-semisimplified category $\Rep_q(\GC)$ at a root of unity is the natural output of BV-BFV quantization before the semisimplification step.

\subsection{The RT-TQFT functor}

Given a modular tensor category $\MTC$, the Reshetikhin--Turaev construction produces a symmetric monoidal functor
\begin{equation}\label{eq:RT-functor}
\Zrt: \Bord_3^{\mathrm{ext}} \longrightarrow 2\text{-}\Vect,
\end{equation}
assigning categories to circles (namely, the MTC $\MTC$ itself), vector spaces to closed surfaces, linear maps to 3-cobordisms, and complex numbers to closed 3-manifolds. We describe each of these assignments in turn.

\subsubsection{Assignment to surfaces}

To a closed oriented surface $\Sigma_g$ of genus $g$, the RT construction assigns a finite-dimensional vector space $V(\Sigma_g) = \Zrt(\Sigma_g)$, interpreted as the ``space of states'' on $\Sigma_g$. Its dimension is determined by the MTC through the \emph{Verlinde formula}
\begin{equation}\label{eq:Verlinde}
\dim V(\Sigma_g) = \sum_{i \in \mathrm{Irr}(\MTC)} \left(\frac{\mathcal{D}}{\dim_q V_i}\right)^{2-2g} = \mathcal{D}^{2g-2} \sum_{i} (S_{0i})^{2-2g},
\end{equation}
where the sum runs over the (finitely many) isomorphism classes of simple objects, $\mathcal{D}$ is the total quantum dimension, and $S_{0i} = \dim_q(V_i)/\mathcal{D}$ relates the first row of the $S$-matrix to the quantum dimensions. For genus $0$ this gives $\dim V(S^2) = 1$; for genus $1$ it gives $\dim V(T^2) = |\mathrm{Irr}(\MTC)|$, the number of simple objects.

For the quantum group category $\overline{\Rep}_q(\mathfrak{sl}_2)$ at level $k$, the simple objects are $V_0, V_1, \ldots, V_k$ (corresponding to spins $0, \frac{1}{2}, \ldots, \frac{k}{2}$), and the Verlinde formula specializes to
\begin{equation}\label{eq:Verlinde-sl2}
\dim V(\Sigma_g) = \left(\frac{k+2}{2}\right)^{g-1} \sum_{j=0}^{k} \left(\sin\left(\frac{(j+1)\pi}{k+2}\right)\right)^{2-2g}.
\end{equation}
The formula was first derived by Verlinde from conformal field theory; in the RT framework it follows from the modularity of $\MTC$.

More generally, for a surface with punctures $\Sigma_{g,n}$ where the $j$-th puncture is labeled by a simple object $V_{i_j} \in \MTC$, the state space is the space of \emph{conformal blocks}, a finite-dimensional vector space whose dimension can be computed recursively from the fusion rules $N_{ij}^k$ by decomposing the surface into pairs of pants. The dimension depends on the genus, the number of punctures, and the labels, and is given by a generalization of the Verlinde formula involving the full $S$-matrix.

\subsubsection{Surgery formula for closed 3-manifolds}

The RT invariant of a closed 3-manifold is computed via a surgery presentation. Recall that every closed oriented 3-manifold can be obtained from $S^3$ by Dehn surgery on a framed link, and two surgery presentations yield the same 3-manifold if and only if they are related by a finite sequence of Kirby moves (stabilization and handle slides). The RT construction exploits this by defining an expression associated to a surgery presentation and proving that it is invariant under Kirby moves.

Concretely, suppose $M$ is presented as surgery on a framed link $L = L_1 \cup \cdots \cup L_m$ in $S^3$, with each component $L_j$ carrying an integer framing $f_j$. The RT invariant is
\begin{equation}\label{eq:surgery}
\Zrt(M) = \kappa^{-\sigma(L)} \cdot \mathcal{D}^{-(m+1)} \sum_{V_1, \ldots, V_m \in \mathrm{Irr}(\MTC)} \left(\prod_{j=1}^m \dim_q(V_j)\right) J_L(V_1, \ldots, V_m).
\end{equation}
Here $J_L(V_1, \ldots, V_m)$ is the \emph{colored link invariant} (or \emph{quantum link invariant}), obtained by evaluating the ribbon tangle presentation of $L$ in the graphical calculus of $\MTC$ with the $j$-th component colored by the simple object $V_j$. The linking matrix $B$ has entries $B_{ij} = \mathrm{lk}(L_i, L_j)$ for $i \neq j$ and $B_{ii} = f_i$, and $\sigma(L) = \sigma_+ - \sigma_-$ is its signature. The constant $\kappa = \mathcal{D}^{-1} \sum_i (\dim_q V_i)^2\, \theta_i$, where $\theta_i$ is the twist eigenvalue of the simple object $V_i$, is a root of unity whose argument is related to the central charge of the associated conformal field theory. The overall normalization factor $\mathcal{D}^{-(m+1)}$ is chosen to ensure the correct behavior under stabilization.

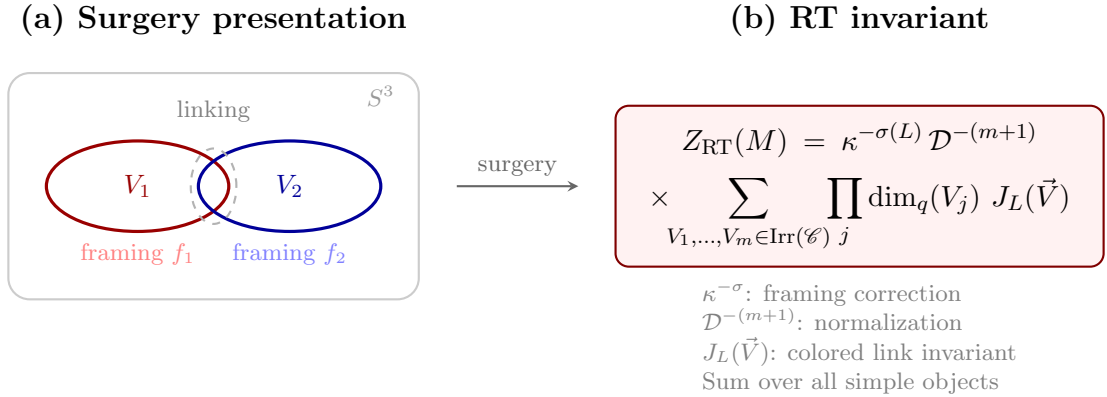
\begin{figure}[ht]
\centering
\begin{tikzpicture}[scale=1]

% === Left: framed link in S^3 ===

\node[font=\bfseries] at (-4.5,3.2) {(a) Surgery presentation};

% --- A genuine Hopf link: the two components link, they do not intersect. ---
% Component 2 (blue) drawn first, in full.
\draw[blue!60!black, very thick] (-3.5,1) ellipse (1.2 and 0.6);

% Component 1 (red) drawn with a white halo, so it passes OVER the blue
% at both apparent crossings.
\draw[white, line width=5pt] (-5.5,1) ellipse (1.2 and 0.6);
\draw[red!60!black, very thick] (-5.5,1) ellipse (1.2 and 0.6);

% Restore the blue strand OVER the red at the lower crossing, so that the
% two components are genuinely linked (one over-crossing, one under-crossing).
\draw[white, line width=5pt]
  ([shift={(203:1.2 and 0.6)}] -3.5,1) arc[start angle=203, end angle=224, x radius=1.2, y radius=0.6];
\draw[blue!60!black, very thick]
  ([shift={(203:1.2 and 0.6)}] -3.5,1) arc[start angle=203, end angle=224, x radius=1.2, y radius=0.6];

\node[red!60!black, font=\footnotesize\bfseries] at (-6.0,1) {$V_1$};
\node[red!50, font=\scriptsize] at (-6.0,0.1) {framing $f_1$};
\node[blue!60!black, font=\footnotesize\bfseries] at (-3.0,1) {$V_2$};
\node[blue!50, font=\scriptsize] at (-3.0,0.1) {framing $f_2$};
\node[black!50, font=\scriptsize] at (-4.5,2.0) {$\mathrm{lk}(L_1,L_2) = 1$};

% S^3 background
\draw[black!20, thick, rounded corners=8pt] (-7.2,-0.5) rectangle (-1.8,2.5);
\node[black!40, font=\scriptsize] at (-2.3,2.2) {$S^3$};

% Arrow
\draw[-{stealth}, thick, black!60] (-1.3,1) -- (0.3,1);
\node[above, font=\scriptsize, black!60] at (-0.5,1) {surgery};

% === Right: the invariant ===

\node[font=\bfseries] at (4,3.2) {(b) RT invariant};

\node[draw=red!50!black, thick, rounded corners=5pt, fill=red!5,
      inner sep=6pt, font=\small, align=center, text width=6cm] at (4,1)
      {$\displaystyle\Zrt(M) = \kappa^{-\sigma(L)}\,\mathcal{D}^{-(m+1)}$\\[6pt]
       $\displaystyle\times\!\!\sum_{V_1,\ldots,V_m \in \mathrm{Irr}(\MTC)}\!\!
       \prod_j \dim_q(V_j)\; J_L(\vec{V})$};

% Annotations
\node[black!50, font=\scriptsize, align=left] at (4,-1)
      {$\kappa^{-\sigma}$: framing correction\\
       $\mathcal{D}^{-(m+1)}$: normalization\\
       $J_L(\vec{V})$: colored link invariant\\
       Sum over all simple objects};

\end{tikzpicture}
\caption{The RT invariant via surgery. \textbf{(a)} A closed 3-manifold $M$ is presented as surgery on a framed link $L = L_1 \cup L_2$ in $S^3$, with each component colored by a simple object $V_j$ of the MTC. \textbf{(b)} The RT invariant is computed by summing the colored link invariant $J_L(\vec{V})$ over all labelings by simple objects, weighted by quantum dimensions, and corrected by normalization factors depending on the signature $\sigma(L)$ and the total quantum dimension $\mathcal{D}$. Invariance under Kirby moves requires modularity of $\MTC$.}
\label{fig:surgery}
\end{figure}

The main point of the RT construction is that this expression is a topological invariant:

\begin{theorem}[Reshetikhin--Turaev \cite{RT1990}]\label{thm:RT-invariance}
The expression \eqref{eq:surgery} is independent of the choice of surgery presentation of $M$. That is, it is invariant under the two Kirby moves:
\begin{enumerate}[label=(\roman*)]
\item \emph{Stabilization:} adding or removing an unknotted component with framing $\pm 1$, disjoint from the rest of the link. Invariance under this move is guaranteed by the normalization factors $\mathcal{D}^{-(m+1)}$ and $\kappa^{-\sigma(L)}$.
\item \emph{Handle slide:} replacing one component by its band-sum with another. Invariance under this move relies on the modularity of $\MTC$ (specifically, on the invertibility of the $S$-matrix) and is where the non-degeneracy condition plays its essential role.
\end{enumerate}
\end{theorem}

\subsection{Extended TQFT structure}

The RT construction as described so far assigns vector spaces to closed surfaces and numbers to closed 3-manifolds. However, the theory has a richer structure: it can be \emph{extended} one categorical level down, assigning data not only to 2- and 3-dimensional manifolds but also to 1-dimensional manifolds (circles) and their cobordisms (surfaces with boundary). This extended structure is both mathematically natural and essential for the comparison with factorization homology in Section \ref{sec:fact-homology}.

\begin{definition}[Extended TQFT]\label{def:extended-TQFT}
A \emph{$(3\text{-}2\text{-}1)$-extended TQFT} is a symmetric monoidal 2-functor
\begin{equation}
Z: \Bord_3^{\mathrm{ext}} \longrightarrow 2\text{-}\Vect,
\end{equation}
where $\Bord_3^{\mathrm{ext}}$ is the extended bordism 2-category whose objects are closed 1-manifolds, whose 1-morphisms are surfaces with boundary (cobordisms between 1-manifolds), and whose 2-morphisms are 3-manifolds with corners (cobordisms between cobordisms) and $2\text{-}\Vect$ is the 2-category of $\CC$-linear categories, linear functors, and natural transformations.
\end{definition}

For the Reshetikhin--Turaev theory built from a modular tensor category $\MTC$, the extended structure takes the following form:
\begin{enumerate}[label=(\roman*)]
\item To a circle $S^1$, the theory assigns the modular tensor category itself, i.e., $\Zrt(S^1) = \MTC$. This is the ``local datum'' from which everything else is built.

\item To a surface $\Sigma$ with incoming boundary circles $S^1_{\mathrm{in}}$ and outgoing boundary circles $S^1_{\mathrm{out}}$, the theory assigns a functor $\Zrt(\Sigma): \MTC^{\boxtimes |\mathrm{in}|} \to \MTC^{\boxtimes |\mathrm{out}|}$ between (Deligne) tensor powers of $\MTC$. Gluing surfaces corresponds to composition of functors.

\item The annulus $S^1 \times [0,1]$, viewed as a cobordism from $S^1$ to itself, is assigned the identity functor. More interestingly, the category associated to the annulus as an object in its own right (via a different decomposition) is the \emph{Drinfeld center} $\mathcal{Z}(\MTC)$ which is the category of objects of $\MTC$ equipped with a natural half-braiding. This is the categorical analogue of the fact that the Hochschild cohomology of an algebra acts on its modules. The full dualizability of $\overline{\Rep}_q(\GC)$ as an object in the Morita 4-category, which is required for the Cobordism Hypothesis to apply, is established in \cite{BJS2021} (Full dualizability pertains to the fully extended $(3\text{-}2\text{-}1\text{-}0)$-setting; for the $(3\text{-}2\text{-}1)$-extended framework used in our Main Conjecture, the relevant input is the $\EE_2$-structure on $\overline{\Rep}_q(\GC)$ and its factorization homology over surfaces.).
\end{enumerate}

For a surface with punctures $\Sigma_{g,n}$, where the $j$-th puncture is labeled by an object
$V_{i_j} \in \MTC$, the theory assigns the space of \emph{conformal blocks}
$V(\Sigma_{g,n};\, V_{i_1}, \ldots, V_{i_n})$. Its dimension is computed purely inside $\MTC$,
by decomposing the surface into pairs of pants: each pair of pants contributes a fusion tensor
$N_{ij}^k$, and the dimension is obtained by contracting these tensors according to the
combinatorics of the decomposition. The result is independent of the pants decomposition,
reflecting the associativity of the tensor product in $\MTC$; this is the classical
Verlinde-type computation and requires no input beyond the fusion rules. There is also a
description of these spaces in terms of factorization homology, but since factorization homology
is only introduced in the next section we postpone it to Remark \ref{rmk:conformal-blocks-FH},
where it belongs.

\section{Factorization Homology}\label{sec:fact-homology}

Factorization homology provides a systematic machine for producing topological invariants of manifolds from algebraic input. It is the mathematical framework that will mediate between the field-theoretic (BV-BFV) and the combinatorial (RT) sides of our bridge. The basic idea is that given an algebraic object $\mathcal{A}$ with a multiplication parametrized by configurations of disks (an $\EE_n$-algebra), one can ``integrate'' $\mathcal{A}$ over any $n$-manifold to obtain a topological invariant. The resulting theory satisfies a gluing axiom ($\otimes$-excision) that parallels both the BV-BFV gluing formula and the TQFT composition law.

\begin{quote}\small
\emph{This section is expository.} The little disks operads go back to Boardman--Vogt
\cite{BV1973} and May \cite{May1972}; the $\infty$-categorical treatment of $\EE_n$-algebras we
use is Lurie's \cite[\S 5]{Lurie_HA}, and Dunn additivity is
\cite{Dunn1988} and \cite[Thm.~5.1.2.2]{Lurie_HA}. Factorization homology in the topological
form used here is due to Ayala--Francis \cite{AF2015}, with the stratified version in
\cite{AFT2019} and the $\EE_n$-tangent-complex perspective in Francis \cite{Francis2013};
Costello--Gwilliam \cite{CG2017,CG2021} develop the closely related theory of factorization
algebras, and Beilinson--Drinfeld \cite{BD2004} the chiral-algebra antecedent. For the
categorical aspects ($\EE_2$-algebras in linear categories, the Deligne tensor product, and
the $(\infty,2)$-categorical setting) we follow Ben-Zvi--Brochier--Jordan
\cite{BBJ2018,BBJ2018b}, Brochier--Jordan--Snyder \cite{BJS2021} and \cite{EGNO2015}. We refer
to these sources for detail and give here only what is needed later.
\end{quote}

\subsection{$\EE_n$-algebras}

The algebraic input for factorization homology is the notion of an $\EE_n$-algebra which is an algebra whose multiplication operations are parametrized by the little $n$-disks operad \cite{May1972,BV1973,Lurie_HA}. We recall the definitions.

\begin{definition}[Little $n$-disks operad]\label{def:En-operad}
The \emph{little $n$-disks operad} $\EE_n$ is the topological operad whose space of $k$-ary operations is
\begin{equation}
\EE_n(k) = \mathrm{Emb}^{\mathrm{rec}}\!\left(\coprod_{i=1}^k D^n,\; D^n\right),
\end{equation}
the space of rectilinear embeddings of $k$ disjoint copies of the standard $n$-disk into a single $n$-disk. Composition in the operad is given by substitution of embeddings, i.e., plugging the output disk of one configuration into an input disk of another.
\end{definition}

\begin{definition}[$\EE_n$-algebra]\label{def:En-algebra}
An \emph{$\EE_n$-algebra} in a symmetric monoidal $(\infty,1)$-category $\mathcal{V}$ is an algebra over the operad $\EE_n$ in $\mathcal{V}$. Concretely, this is an object $A \in \mathcal{V}$ together with, for each $k \geq 0$ and each configuration $\varphi \in \EE_n(k)$ of $k$ non-overlapping disks inside a larger disk, a $k$-ary multiplication map $A^{\otimes k} \to A$, subject to coherent associativity and equivariance conditions encoded by the topology of the configuration spaces $\EE_n(k)$.
\end{definition}

The $\EE_n$-algebras form a hierarchy of increasingly commutative algebraic structures, interpolating between fully associative and fully commutative.

\begin{example}[The hierarchy of $\EE_n$-algebras]\label{ex:En-hierarchy}
\leavevmode
Throughout this example we say explicitly in which symmetric monoidal target $\mathcal{V}$ we are
working, since the answer depends on it. Items (i), (ii) and (v) are stated for $\mathcal{V} =
\mathrm{\mathbf{Ch}}_\CC$, the $\infty$-category of cochain complexes over $\CC$; items (iii)
and (iv) are stated for $\mathcal{V}$ a suitable $(2,1)$-category of linear categories.
\begin{enumerate}[label=(\roman*)]
\item \emph{$\EE_0$-algebras} in a symmetric monoidal $\infty$-category are pointed objects,
i.e.\ objects $A$ together with a map $\eta: \mathbf{1} \to A$ from the unit. In
$\mathrm{\mathbf{Ch}}_\CC$ this is a cochain complex with a distinguished cocycle. (The formal
moduli problems of derived algebraic geometry are controlled not by $\EE_0$-algebras but by dg
Lie algebras, through Koszul duality between the Lie and commutative operads
\cite{Lurie_DAGX,Pridham2010}; see Remark \ref{rmk:koszul-formal}.)

\item \emph{$\EE_1$-algebras} in $\mathrm{\mathbf{Ch}}_\CC$ are associative algebras up to
coherent homotopy, also known as $A_\infty$-algebras. The space $\EE_1(k)$ has contractible
components (the configuration space of $k$ ordered intervals in an interval is contractible),
so the only structure is a binary product together with higher associativity homotopies.

\item \emph{$\EE_2$-algebras} carry a braided multiplication. The configuration space
$\EE_2(2)$ of two disjoint disks inside a larger disk is homotopy equivalent to $S^1$; the
generating loop moves one input disk \emph{all the way around} the other and back to its
starting position; this full winding, not a half-rotation, gives the fundamental class of
$\EE_2(2) \simeq S^1$. (The braiding itself corresponds to a path in the configuration
space of \emph{unordered} pairs, i.e.\ to a half-winding, and the full loop above is the
\emph{double} braiding $c_{W,V}c_{V,W}$, the monodromy that appears in the
$S$-matrix \eqref{eq:S-matrix}.) By the Dunn additivity theorem
\cite{Dunn1988}, \cite[Thm.~5.1.2.2]{Lurie_HA}, $\EE_2 \simeq \EE_1 \otimes \EE_1$, so an
$\EE_2$-algebra has two compatible associative products, one for each coordinate direction. In
the categorical setting, $\EE_2$-algebras in a $(2,1)$-category of linear categories are
braided monoidal categories \cite{BJS2021}.

\item \emph{$\EE_n$-algebras for $n \ge 3$} in a $(2,1)$-category of linear categories are
symmetric monoidal categories: an $\EE_3$-structure on a monoidal category is a braiding together
with a trivialization of the double braiding. In particular a non-degenerate braided category
such as $\overline{\Rep}_q(\GC)$ is not an $\EE_3$-algebra in this sense, and the extension of
the RT theory down to points goes through a $4$-dimensional theory instead (Section
\ref{sec:outlook}).

\item \emph{$\EE_\infty$-algebras} are fully commutative up to all higher homotopies, since
$\EE_\infty(k) \simeq *$ for all $k$. In $\mathrm{\mathbf{Ch}}_\CC$, and more generally in
cochain complexes over a field of characteristic zero, $\EE_\infty$-algebras are the same
thing as commutative differential graded algebras up to quasi-isomorphism. In categories, the
$\EE_\infty$-algebras are the symmetric monoidal categories.
\end{enumerate}
\end{example}

The key point for our purposes is that a ribbon category (such as the modular tensor category $\MTC$ appearing in the RT construction) is naturally an $\EE_2$-algebra in the $(\infty,2)$-category of linear categories. This is the algebraic structure that factorization homology will ``integrate'' over surfaces.

\subsection{Definition and basic properties of Factorization homology}

With the notion of $\EE_n$-algebra in hand, we can define factorization homology as the operation that ``integrates'' an $\EE_n$-algebra over a manifold to produce a topological invariant.

\begin{definition}[Factorization homology]\label{def:fact-homology}
Let $\mathcal{A}$ be an $\EE_n$-algebra in a symmetric monoidal $(\infty,1)$-category $\mathcal{V}$, and let $M$ be a closed $n$-framed manifold. The \emph{factorization homology} of $M$ with coefficients in $\mathcal{A}$ is\footnote{We will denote the factorization homology over a 1-dimensional manifold by $\int$, over a 2-dimensional manifold by $\iint$ and over a 3-dimensional manifold by $\iiint$. For the general $n$-dimensional case we just use $\int$.}
\begin{equation}\label{eq:FH-def}
\int_M \mathcal{A} \;:=\; \mathop{\mathrm{colim}}\limits_{\mathbf{Disk}_{n/M}}\; \mathcal{A}^{(-)}\;\;\in\;\mathcal{V},
\end{equation}
where $\mathrm{\mathbf{Disk}}_{n/M}$ is the $(\infty,1)$-category whose objects are rectilinear embeddings \[\varphi: \coprod_{i \in I} D^n \hookrightarrow M\] of finite disjoint unions of $n$-disks into $M$, and the functor $\mathcal{A}^{(-)}$ sends such an embedding to the tensor product $\bigotimes_{i \in I} \mathcal{A}$. See Figure \ref{fig:fact-homology} for a visualization.
\end{definition}

\begin{remark}\label{rmk:FH-intuition}
The colimit in \eqref{eq:FH-def} can be understood informally as a ``continuous tensor product''. In particular, one places a copy of $\mathcal{A}$ at each point of $M$ and multiplies nearby copies using the $\EE_n$-structure, then assembles the result globally by taking the colimit over all ways of covering $M$ by disks. This is the higher-algebraic analogue of integrating a function over a manifold, with the $\EE_n$-multiplication playing the role of pointwise multiplication and the colimit playing the role of the integral. From the perspective of quantum field theory, if $\mathcal{A}$ is the algebra of local observables of a topological field theory (as constructed by Costello--Gwilliam \cite{CG2017,CG2021}), then $\int_M \mathcal{A}$ is the algebra of global observables on $M$. So the factorization homology reconstructs the global theory from its local data. See also \cite{Scheimbauer2014} for the construction of factorization homology as a fully extended topological field theory, and \cite{AFT2019} for the extension to stratified spaces.
\end{remark}

The fundamental computational tool for factorization homology is the excision property, which allows one to compute $\int_M \mathcal{A}$ by cutting $M$ into simpler pieces and assembling the results.

\begin{theorem}[$\otimes$-Excision, Ayala--Francis \cite{AF2015}]\label{thm:AF-excision}
Let $M = M_1 \cup_{N \times \RR} M_2$ be a decomposition of an $n$-manifold along a collar neighborhood of a codimension-1 submanifold $N$ (a ``collar-gluing''). Then factorization homology satisfies
\begin{equation}\label{eq:excision}
\int_M \mathcal{A} \;\simeq\; \int_{M_1} \mathcal{A} \;\otimes_{\int_{N \times \RR} \mathcal{A}}\; \int_{M_2} \mathcal{A},
\end{equation}
where $\otimes_{\int_{N \times \RR} \mathcal{A}}$ denotes the relative tensor product (bar construction) over the $\EE_1$-algebra $\int_{N \times \RR} \mathcal{A}$. See Figure \ref{fig:fact-homology} for a visualization.
\end{theorem}

\begin{figure}[ht]
\centering
\begin{tikzpicture}[scale=1]

% === Left: Surface with disks ===

% Genus-2 surface (simplified as a blob)
\fill[green!6] (0,0) ellipse (3.2 and 1.8);
\draw[green!40!black, thick] (0,0) ellipse (3.2 and 1.8);

% Holes to suggest genus (two small ellipses cut out)
\fill[white] (-1.2,0.1) ellipse (0.45 and 0.25);
\draw[green!40!black, thick] (-1.2,0.1) ellipse (0.45 and 0.25);
\fill[white] (1.2,0.1) ellipse (0.45 and 0.25);
\draw[green!40!black, thick] (1.2,0.1) ellipse (0.45 and 0.25);

% Small disks covering the surface
\foreach \x/\y in {-2.1/0.7, -0.3/1.3, 1.8/0.9, -1.9/-0.7, 1.5/-0.9, 2.2/-0.5, 0.1/-0.2, -0.5/0.5, 1/0.8, 2.5/0.3} {
    \fill[green!25, opacity=0.6] (\x,\y) circle (0.35);
    \draw[green!50!black, thick, opacity=0.7] (\x,\y) circle (0.35);
    \node[green!50!black, font=\tiny\bfseries, opacity=0.9] at (\x,\y) {$\mathcal{A}$};
}

% Surface label
\node[green!40!black, font=\footnotesize\bfseries] at (0,-2.3) {$\Sigma$};

% === Arrow ===
\draw[-{stealth}, thick, black!60] (3.8,0) -- (5.2,0);
\node[above, font=\scriptsize, black!60] at (4.5,0) {$\displaystyle\mathop{\mathrm{colim}}\limits_{\mathbf{Disk}_{2/\Sigma}}$};

% === Right: Result ===
\node[draw=green!40!black, thick, rounded corners=5pt, fill=green!5,
      inner sep=8pt, font=\normalsize] at (7,0)
      {$\displaystyle\iint_\Sigma \mathcal{A}$};

% === Below: Excision ===
\draw[thick, black!20] (-5.5,-3.2) -- (9,-3.2);

\node[font=\bfseries, black!70] at (-3.5,-3.8) {$\otimes$-Excision:};

% Cut surface
\fill[green!6] (-3.5,-5.5) ellipse (1.2 and 0.8);
\draw[green!40!black, thick] (-3.5,-5.5) ellipse (1.2 and 0.8);
\node[green!40!black, font=\footnotesize\bfseries] at (-3.5,-5.5) {$\Sigma_1$};

% Cutting curve
\draw[red!60!black, very thick] (-2,-5.5) -- (-2,-4.5);
\draw[red!60!black, very thick] (-2,-5.5) -- (-2,-6.5);
\node[red!60!black, font=\scriptsize] at (-1.5,-4.5) {$N$};

\fill[green!6] (-0.5,-5.5) ellipse (1.2 and 0.8);
\draw[green!40!black, thick] (-0.5,-5.5) ellipse (1.2 and 0.8);
\node[green!40!black, font=\footnotesize\bfseries] at (-0.5,-5.5) {$\Sigma_2$};

% Equals
\node[font=\normalsize] at (1.5,-5.5) {$\leadsto$};

% Formula
\node[draw=green!40!black, thick, rounded corners=4pt, fill=green!5,
      inner sep=6pt, font=\small, align=center] at (5.2,-5.5)
      {$\displaystyle\iint_\Sigma \mathcal{A} \;\simeq\;
       \iint_{\Sigma_1}\!\mathcal{A} \;\otimes_{\iint_{N \times \RR}\!\mathcal{A}}\;
       \iint_{\Sigma_2}\!\mathcal{A}$};

\end{tikzpicture}
\caption{\emph{Top:} Factorization homology as a ``continuous tensor product.'' The $\EE_2$-algebra $\mathcal{A}$ is placed on each disk embedded in the surface $\Sigma$; the colimit over all disk embeddings produces the global invariant $\iint_\Sigma \mathcal{A}$. \emph{Bottom:} The $\otimes$-excision property (Theorem \ref{thm:AF-excision}): cutting $\Sigma$ along a codimension-1 submanifold $N$ expresses the factorization homology as a relative tensor product over the data assigned to $N \times \RR$.}
\label{fig:fact-homology}
\end{figure}

\begin{remark}\label{rmk:excision-TQFT}
The excision formula \eqref{eq:excision} is the factorization homology incarnation of the TQFT gluing axiom. It parallels the BV-BFV gluing formula \eqref{eq:BV-gluing}. In both cases, the invariant of a manifold obtained by gluing along a hypersurface is computed as a pairing (or relative tensor product) over the data associated to the gluing surface. We would like to make this parallel precise. In particular, we want to show that the BV-BFV pairing over $\mathcal{H}_\Sigma$ and the factorization homology relative tensor product over $\iint_\Sigma \mathcal{A}$ are computing the same thing.
\end{remark}

\subsection{Low-dimensional computations}

To build intuition and establish the key identifications that our bridge relies on, we compute factorization homology in the simplest cases.

\begin{proposition}\label{prop:FH-S1}
For an $\EE_1$-algebra $A$ (i.e., an associative algebra up to homotopy), the factorization homology of the circle is the Hochschild homology of $A$. So, we get
\begin{equation}
\int_{S^1} A \;\simeq\; \HH_\bullet(A) \;:=\; A \otimes_{A \otimes A^{\op}} A.
\end{equation}
\end{proposition}

\begin{proof}[Proof sketch]
Decompose the circle as a union of two intervals glued at their endpoints, i.e., $S^1 = I \cup_{\{0,1\}} I$. Factorization homology of an interval is simply $\int_I A \simeq A$ (since the interval is contractible, the colimit over disk embeddings degenerates). The two boundary points contribute $\int_{\{0,1\}} A \simeq A \otimes A$, which acts on each copy of $A$ as a bimodule via left and right multiplication. Applying $\otimes$-excision (Theorem \ref{thm:AF-excision}) to this decomposition gives $\int_{S^1} A \simeq A \otimes_{A \otimes A^{\op}} A$, which is the bar construction computing the Hochschild homology of $A$.
\end{proof}

This computation has a natural categorical upgrade. When we pass from $\EE_1$-algebras (associative algebras) to $\EE_2$-algebras (braided monoidal categories), the Hochschild homology is replaced by its categorical analogue which is the Drinfeld center.

\begin{proposition}\label{prop:FH-cat}
Let $\MTC$ be an $\EE_2$-algebra in the $(\infty,2)$-category of $\CC$-linear categories (i.e., a braided monoidal category). Then we get that
\begin{enumerate}[label=(\roman*)]
\item $\int_{S^1} \MTC \simeq \mathcal{Z}(\MTC)$, the Drinfeld center of $\MTC$. This is the categorical Hochschild homology.
\item $\iint_{\RR^2} \MTC \simeq \MTC$. Factorization homology of Euclidean space recovers the algebra itself, reflecting the fact that $\RR^n$ is the ``unit'' for the gluing operation.
\item $\iint_{S^1 \times \RR} \MTC \simeq \mathcal{Z}(\MTC)$ as a monoidal category, where the monoidal structure comes from the non-compact $\RR$ direction. This is the category that factorization homology assigns to an annulus, consistent with the extended TQFT picture of Section \ref{sec:RT}.
\end{enumerate}
\end{proposition}

\begin{remark}[The Drinfeld center]\label{rmk:Drinfeld}
The Drinfeld center $\mathcal{Z}(\MTC)$ is the category whose objects are pairs $(V, \sigma)$, where $V$ is an object of $\MTC$ and \[\sigma = \{\sigma_W: V \otimes W \xrightarrow{\sim} W \otimes V\}_{W \in \MTC}\] is a \emph{half-braiding}, i.e., a natural family of isomorphisms compatible with the tensor product, but not required to come from the braiding of $\MTC$ itself. The Drinfeld center is always a braided monoidal category, even when $\MTC$ is only monoidal (not braided); when $\MTC$ is already braided, $\mathcal{Z}(\MTC)$ is typically larger. For the full (non-semisimplified) quantum group category, there is a canonical equivalence $\mathcal{Z}(\Rep_q(\GC)) \simeq \Rep_q(\GC) \boxtimes \overline{\Rep_q(\GC)}$, where $\overline{\Rep_q(\GC)}$ denotes the same category with the reversed braiding and $\boxtimes$ is the Deligne tensor product. This is a categorical analogue of the algebraic identity $\HH_\bullet(A) \simeq A \otimes_{A^e} A$, where $A^e = A \otimes A^{\op}$ is the enveloping algebra.
\end{remark}

\subsection{The Ben-Zvi--Brochier--Jordan theorem}

The results of the previous subsection (factorization homology of $S^1$ giving the Drinfeld center, and the excision formula allowing computations on arbitrary surfaces) set the stage for the following theorems. The BBJ and AKZ theorems stated below are established results in the literature; the conjectural content of our program begins with the identification of these structures with the BV-BFV output in Section \ref{sec:bridge}. It connects the factorization homology of a braided tensor category to quantum character varieties, which at roots of unity yield the Reshetikhin--Turaev state spaces via the Ai--Kong--Zheng theorem.

\begin{theorem}[Ben-Zvi--Brochier--Jordan {\cite[Thms.~1.1, 5.14, 6.9]{BBJ2018}}, {\cite{BBJ2018b}}]\label{thm:BBJ}
Let $\mathcal{A}$ be a braided tensor category, viewed as an $\EE_2$-algebra in the symmetric
monoidal $(\infty,2)$-category of presentable $\CC$-linear categories with the Kelly--Deligne
tensor product. Then:
\begin{enumerate}[label=(\roman*)]
\item For a punctured (i.e.\ open) oriented surface $\Sigma^\circ$, the factorization homology
$\iint_{\Sigma^\circ} \mathcal{A}$ is equivalent to the category of modules over an explicitly
presented algebra object $\mathcal{A}_{\Sigma^\circ} \in \mathcal{A}$, obtained by gluing
handles to the disk. (That $\iint_{\Sigma^\circ}\mathcal{A}$ is a category holds by
construction; the content of the statement is the explicit presentation.)
\item For a closed surface $\Sigma_g$, the category $\iint_{\Sigma_g} \mathcal{A}$ is computed
as a quantum Hamiltonian reduction of the category attached to the punctured surface.
\end{enumerate}
\end{theorem}

\begin{theorem}[BBJ; the case of quantum groups, {\cite[\S 6--7]{BBJ2018}}]\label{thm:BBJ-repq}
Take $\mathcal{A} = \Rep_q(\GC)$ for $\GC$ reductive and $q$ generic. Then the algebra objects of
Theorem \ref{thm:BBJ} are the quantum function algebras: for the once-punctured torus the
resulting category is the category of modules over the quantum differential operators
$\mathcal{D}_q(\GC)$, and for the closed torus
\begin{equation}
\iint_{T^2} \Rep_q(\GC) \;\simeq\; \mathcal{D}_q(\GC/\GC)\text{-}\mathrm{\mathbf{mod}},
\end{equation}
the category of $\mathcal{D}_q$-modules on $\GC$ equivariant for the quantum adjoint action.
More generally $\iint_{\Sigma}\Rep_q(\GC)$ quantizes $\QCoh(\Loc_{\GC}(\Sigma))$.
\end{theorem}

\begin{remark}[On the name ``quantum character variety'']\label{rmk:QCV-name}
The term \emph{quantum character variety} for $\iint_\Sigma \mathcal{A}$ is due to
\cite{BBJ2018}. The name is justified when $\mathcal{A}$ is a category of representations of a
quantum group, in which case there is a deformation statement:
$\iint_\Sigma \Rep_q(\GC)$ is a flat deformation of $\QCoh(\Loc_{\GC}(\Sigma))$, the classical
character stack appearing at $q=1$. For a general braided tensor category $\mathcal{A}$ the name is only an analogy. Even for
$\mathcal{A} = \Rep_q(\GC)$, the identification of the classical limit with sheaves on the
character stack rests on the Tannakian statement $\Rep(\GC) \simeq
\QCoh(B\GC)$ and hence on the theory of algebraic stacks, and it is discussed further in
Remark \ref{rmk:Tannakian}.
\end{remark}

\begin{theorem}[Ai--Kong--Zheng {\cite[Thm.~3.5 and \S 4]{AKZ2017}}]\label{thm:AKZ}
Let $\MTC$ be a unitary modular tensor category, viewed as an anomaly-free\footnote{See
\cite{AKZ2017} for the definition of anomaly-freeness; for a unitary MTC it is implied by
modularity via M\"uger's theorem \cite{Mueger2003}.} coefficient system on a closed oriented
surface $\Sigma_g$. Then there is an equivalence of pointed categories
\begin{equation}
\iint_{\Sigma_g} \MTC \;\simeq\; (\mathbf{H},\, u_{\Sigma_g}),
\end{equation}
where $\mathbf{H}$ is the category of finite-dimensional Hilbert spaces and $u_{\Sigma_g} \in
\mathbf{H}$ is a distinguished object.
\end{theorem}

\begin{remark}[On Theorem \ref{thm:AKZ}]\label{rmk:AKZ-nature}
We use this theorem repeatedly, so we comment on three aspects of it.

\emph{(a) Nature of the equivalence.} The equivalence $\iint_{\Sigma_g}\MTC \simeq \mathbf{H}$
is not canonical as a bare equivalence of categories: any two equivalences differ by an
autoequivalence of $\mathbf{H}$, and $\mathbf{H}$ has a large group of such (tensoring by a
line). What \emph{is} canonical is the pair $(\mathbf{H}, u_{\Sigma_g})$ up to a contractible
space of choices: the object $u_{\Sigma_g}$ is the image of the monoidal unit under the excision
decomposition of $\Sigma_g$, and the coherence of the excision data determines it. The residual
ambiguity is a scalar, the framing anomaly; this is why the RT theory is only projectively
defined and why the mapping class group representations of Section \ref{sec:evidence} are
projective.

\emph{(b) The identification with the RT state space.} This is more than an equality of
dimensions: Ai--Kong--Zheng construct, for each pants decomposition of
$\Sigma_g$, an explicit isomorphism between $u_{\Sigma_g}$ and the corresponding conformal
block space of $\MTC$, and show that these isomorphisms are compatible with change of pants
decomposition, and the RT state space is defined by the same conformal block data. Thus
$u_{\Sigma_g} \cong V_{\mathrm{RT}}(\Sigma_g)$ is an explicit isomorphism, canonical up to the
scalar in (a); this is the meaning of $u_{\Sigma_g} = V_{\mathrm{RT}}(\Sigma_g)$ throughout.

\emph{(c) The inner product.} Factorization homology of $\MTC$ produces a $\CC$-linear
category without inner product. The Hilbert space structure on $u_{\Sigma_g}$ comes from the
unitary ($\dagger$-) structure on $\MTC$, which is where unitarity is used in Theorem
\ref{thm:AKZ}: the pairing on conformal blocks is built from the $\dagger$-structure and the
evaluation and coevaluation morphisms. For $\overline{\Rep}_q(\GC)$ at
$q = e^{i\pi/(k+h^\vee)}$ the unitary structure comes from the compact real form, so this is a
point where the compact group $G$ enters (compare Section \ref{subsec:compact-vs-complex}).
\end{remark}

\begin{remark}[Conformal blocks via factorization homology]\label{rmk:conformal-blocks-FH}
We can now state the formula postponed from Section \ref{sec:RT}. For a punctured surface
$\Sigma_{g,n}$ with the $j$-th puncture labelled by $V_{i_j} \in \MTC$, factorization homology
with coefficients in $\MTC$ and with the labelled punctures treated as stratified points
\cite{AFT2019} produces an object of $\MTC$, and the conformal block space is
\begin{equation}\label{eq:conformal-blocks}
V(\Sigma_{g,n};\, V_{i_1}, \ldots, V_{i_n}) \;=\; \Hom_{\MTC}\!\left(\mathbf{1},\; \iint_{\Sigma_{g,n}} \MTC(V_{i_1}, \ldots, V_{i_n})\right).
\end{equation}
That this agrees with the pants-decomposition computation of Section \ref{sec:RT} is the
content of Theorem \ref{thm:AKZ} in the punctured case.
\end{remark}

\begin{remark}
For the quantum group category $\overline{\Rep}_q(\GC)$ with $G$ a compact simple Lie group and $q = e^{i\pi/(k+h^\vee)}$, unitarity follows from the $*$-structure on $U_q(\fg)$ induced by the compact real form, and modularity from the invertibility of the Kac--Peterson $S$-matrix and anomaly-freeness follows from modularity via M\"uger's theorem $\MTC \boxtimes \overline{\MTC} \simeq \mathcal{Z}(\MTC)$, so the hypotheses of Theorem \ref{thm:AKZ} are satisfied.
\end{remark}

\begin{remark}[Two levels of the theory]\label{rmk:two-levels}
We distinguish two levels:
\begin{enumerate}
    \item[(a)] At \emph{generic $q$}, the factorization homology $\iint_\Sigma \Rep_q(\GC)$ is a non-trivial category (the quantum character variety of BBJ). This is the natural setting for comparison with BV-BFV, which produces formal deformations over $\CC[\![\hbar]\!]$.
    \item[(b)] At a \emph{root of unity}, after semisimplification to the MTC $\overline{\Rep}_q(\GC)$, the factorization homology of a closed surface yields $(\mathbf{H},\, u_{\Sigma_g})$, where $u_{\Sigma_g}$ is the RT state space (Theorem \ref{thm:AKZ}).
\end{enumerate}
The passage from (a) to (b) is part of the conjectural Step 2c of our program. The dualizability properties required for these constructions are established in \cite{BJS2021}, the classification of invertible braided tensor categories relevant to the framing anomaly is given in \cite{BJSS2020}, and the extension of factorization homology computations to the 4-dimensional Crane--Yetter theory appears in \cite{KT2022}.
\end{remark}

\subsection{Factorization homology and the Cobordism Hypothesis}

Factorization homology is closely related to the \emph{Cobordism Hypothesis}, which proposes a
classification of fully extended TQFTs in terms of a single piece of algebraic data.

\begin{conjecture}[Cobordism Hypothesis; Baez--Dolan \cite{BaezDolan1995}, Lurie \cite{Lurie_cobordism}]\label{conj:cobordism-hyp}
A fully extended framed $n$-dimensional TQFT with values in a symmetric monoidal $(\infty,n)$-category $\mathcal{V}$ is completely determined by the object it assigns to a point, and this object must be \emph{fully dualizable} in $\mathcal{V}$. More precisely, there is an equivalence
\begin{equation}
\mathrm{\mathbf{Fun}}^{\otimes}(\Bord_n^{\mathrm{fr}},\, \mathcal{V}) \;\simeq\; (\mathcal{V}^{\mathrm{fd}})^{\sim},
\end{equation}
where $\mathcal{V}^{\mathrm{fd}}$ denotes the full subcategory of fully dualizable objects and $(-)^{\sim}$ takes the underlying $\infty$-groupoid (discarding non-invertible morphisms).
\end{conjecture}

\begin{remark}[On the status of the Cobordism Hypothesis]\label{rmk:cobordism-status}
The statement was formulated by Baez and Dolan \cite{BaezDolan1995}. Lurie
\cite{Lurie_cobordism} developed many of the ideas and steps of a proof strategy; some of the
gaps have since been filled, but part of the required machinery is still missing, and the
statement is not yet broadly accepted as proven. Ayala and Francis \cite{AyalaFrancis2017}
outline an approach to (a generalization of) the Cobordism Hypothesis via factorization
homology, several crucial steps of which remain to be completed; Grady and Pavlov
\cite{GradyPavlov2020} claim a proof of a stronger statement, which is currently under review;
the $1$-dimensional case is treated by Harpaz \cite{Harpaz2018}. There is strong confidence in
the statement, and a complete proof can be expected in the not too distant future. Nothing in
this paper depends on it: in the $(3\text{-}2\text{-}1)$-extended framework of Definition
\ref{def:extended-TQFT} the input is an $\EE_2$-algebra and the theory is produced directly by
factorization homology, so no classification statement is needed. The Cobordism Hypothesis
serves only as motivation, here and in the discussion of the fully extended theory in Section
\ref{sec:outlook}.
\end{remark}

\begin{remark}\label{rmk:cobordism-FH}
Ayala and Francis \cite{AF2015} characterize factorization homology as the unique symmetric
monoidal functor on framed $n$-manifolds and embeddings that satisfies $\otimes$-excision and
takes the value $\mathcal{A}$ on $\RR^n$. Scheimbauer \cite{Scheimbauer2014} constructs
factorization homology as a fully extended TQFT with values in a higher Morita category, so
that for suitably dualizable $\mathcal{A}$ the assignment $M \mapsto \int_M \mathcal{A}$ is the
fully extended TQFT that the Cobordism Hypothesis attaches to $\mathcal{A}$. In this sense
factorization homology is the excisive extension of local data to global invariants. This
motivates the expectation, formulated in Section \ref{sec:bridge}, that the BV-BFV partition
function, which also satisfies a gluing axiom, agrees with factorization homology once the
local data (the $\EE_2$-algebra) are matched.
\end{remark}

\section{The Derived Character Stack and Shifted Symplectic Geometry}\label{sec:derived-character}

The second ingredient is the derived algebraic geometry of the character stack, in particular, the moduli space of local systems (flat connections) on a manifold. This is the geometric object that both the BV-BFV formalism and the RT construction ``see,'' albeit from different angles. The BV-BFV formalism quantizes it perturbatively via gauge-fixing and Feynman diagrams, while the RT construction accesses it non-perturbatively through the representation theory of quantum groups. The shifted symplectic structures discovered by Pantev--To\"en--Vaqui\'e--Vezzosi (PTVV) \cite{PTVV2013} provide the geometric framework in which these two perspectives can be compared.

\begin{quote}\small
\emph{This section is expository up to and including Section \ref{subsec:BV-BFV-dictionary}.}
For derived algebraic geometry the references are To\"en--Vezzosi \cite{TV2005,TV2008},
Lurie \cite{Lurie_SAG,Lurie_HTT} and Gaitsgory--Rozenblyum \cite{GR2017}; for shifted symplectic
geometry, Pantev--To\"en--Vaqui\'e--Vezzosi \cite{PTVV2013} and Calaque \cite{Calaque2015,Calaque2017};
for shifted Poisson structures and their quantization, \cite{CPTVV2017,Safronov2017b,Safronov2018,MS2018};
and for surveys, Pantev--Vezzosi \cite{PV2018} and Calaque \cite{CalaqueSurvey}. We recall only what is needed.
\end{quote}

\subsection{Essentials of derived algebraic geometry}

We briefly recall the notions from derived algebraic geometry (DAG) that we need.

The basic idea of DAG is to replace ordinary commutative rings by \emph{connective} commutative
differential graded algebras (those with vanishing cohomology in positive degrees) or,
equivalently in characteristic zero, by simplicial commutative rings; connectivity is the
condition matching simplicial commutative rings under Dold--Kan. This allows the ``functions'' on a space to form a cochain complex rather than a plain ring, and the
extra structure records higher-order information about the geometry: in particular, the
obstructions and deformations that are invisible to the classical (underived) moduli space.

\begin{definition}[Derived stack]\label{def:derived-stack}
A \emph{derived stack} (over $\CC$) is a functor \[X: \mathrm{\mathbf{cdga}}_\CC^{\leq 0} \to \mathcal{S}\]
from the $\infty$-category of connective commutative differential graded $\CC$-algebras to the
$\infty$-category of spaces (simplicial sets), satisfying descent for the \'etale topology.
Here the \'etale topology on $\mathrm{\mathbf{cdga}}_\CC^{\leq 0}$ is the Grothendieck topology
in which a family $\{A \to B_i\}$ is a covering if each $A \to B_i$ is \'etale (that is,
$H^0(A)\to H^0(B_i)$ is \'etale and $H^\bullet(A)\otimes_{H^0(A)}H^0(B_i) \xrightarrow{\sim}
H^\bullet(B_i)$) and the
induced family $\{\mathrm{Spec}\,H^0(B_i) \to \mathrm{Spec}\,H^0(A)\}$ is a covering in the
usual \'etale topology of schemes. Descent then says that $X$ takes \v{C}ech nerves of such
families to limits. See \cite[\S 1.3]{TV2008} or \cite[\S B]{Lurie_SAG} for details.

Every derived stack $X$ has an underlying \emph{classical truncation} $X_{\mathrm{cl}} =
X|_{\mathrm{\mathbf{cAlg}}_\CC}$, obtained by restricting the functor to ordinary commutative
algebras.
\end{definition}

\begin{remark}[What the derived structure remembers]\label{rmk:derived-remembers}
The classical truncation records the $A$-points of $X$ for $A$ an ordinary ring; the derived
stack in addition records the $A$-points for cdgas $A$ with cohomology in negative degrees, and
these are what deformation theory needs.

Concretely, take $X$ an ordinary affine scheme cut out by equations, say\\ $X =
\mathrm{Spec}\,\CC[x_1,\dots,x_n]/(f_1,\dots,f_m)$, and let $x \in X$ be a point. The classical
tangent space at $x$ is $\ker(\mathrm{d}f)_x$, the space of first-order deformations of $x$
inside the ambient affine space that preserve the equations. What the classical scheme does
\emph{not} record is the \emph{obstruction} space: if one tries to extend a first-order
deformation to second order, one meets an element of $\mathrm{coker}(\mathrm{d}f)_x$, and this
element is the obstruction. The derived enhancement $X^{\mathrm{der}} =
\mathrm{Spec}$ of the Koszul complex $\CC[x_1,\dots,x_n]\otimes\Lambda^\bullet(\CC^m)$ has
$\mathbb{T}_{X^{\mathrm{der}},x}$ the two-term complex $[\CC^n \xrightarrow{(\mathrm{d}f)_x}
\CC^m]$, whose $H^0$ is the tangent space and whose $H^1$ is the obstruction space. When $X$ is
a complete intersection (the $f_i$ form a regular sequence) the derived and classical objects
agree; in general the derived structure measures the failure of this. It is detected by
cdgas with $H^{-1} \neq 0$: a map $\mathrm{Spec}\,A \to X$ for such an $A$ is a solution
of the equations up to a specified homotopy, and the homotopy is the datum the classical scheme
forgets. The same discussion applies to the character stack, where the equations
are the relations in $\pi_1$ and the obstruction space is $H^2(\Sigma;\Ad_\rho)$; see
Proposition \ref{prop:Loc-explicit}.
\end{remark}

\begin{definition}[Cotangent complex]\label{def:cotangent}
For a derived stack $X$, the \emph{cotangent complex} $\mathbb{L}_X \in \QCoh(X)$ is the
quasi-coherent sheaf that \emph{corepresents} derivations: for every $A$-point $x:
\mathrm{Spec}\,A \to X$ and every connective $A$-module $M$, one has
\[
\Map_{\QCoh(\mathrm{Spec}\,A)}(x^*\mathbb{L}_X,\, M) \;\simeq\; \mathrm{Der}(x, M),
\]
naturally in $M$, where the right-hand side is the space of extensions of $x$ over the square-zero
extension $A \oplus M$. It is the derived-geometric replacement of the sheaf of K\"ahler
differentials. The \emph{tangent complex} is its dual, $\mathbb{T}_X := \mathbb{L}_X^\vee$
(which requires $\mathbb{L}_X$ to be perfect, as it is in all cases we consider).

If $X$ is a moduli space, a point $x$ classifies an object (a representation, a bundle, a
local system), and the ``deformation theory of $x$'' refers to deformations of that object over
Artinian bases.
Thus $H^0(\mathbb{T}_{X,x})$ parametrizes first-order deformations of the classified object,
$H^1(\mathbb{T}_{X,x})$ receives the obstructions to extending them, and
$H^{-1}(\mathbb{T}_{X,x})$ is the space of infinitesimal automorphisms of the object.
\end{definition}

\begin{proposition}[Cotangent complex of a mapping stack; PTVV {\cite[\S 2.1]{PTVV2013}}]\label{prop:cotangent-map}
Let $Y$ be a compact oriented $d$-manifold, let $Y_B$ denote the associated Betti stack (the
constant derived stack on the homotopy type of $Y$), let $Z$ be a derived stack with perfect
cotangent complex, and let $X = \Map(Y_B, Z)$ be the derived mapping stack, i.e.\ the derived
stack whose $A$-points are maps $Y_B \times \mathrm{Spec}\,A \to Z$. Write
\[
\begin{tikzcd}[column sep=3em]
Y_B \times X \arrow[r, "\mathrm{ev}"] \arrow[d, "p_X"'] & Z\\
X &
\end{tikzcd}
\]
for the evaluation map and the projection. Then
\begin{equation}\label{eq:cotangent-map}
\mathbb{T}_X \;\simeq\; (p_X)_* \,\mathrm{ev}^*\,\mathbb{T}_Z ,
\qquad
\mathbb{L}_X \;\simeq\; (p_X)_*\big(\mathrm{ev}^*\,\mathbb{L}_Z \otimes \omega_{Y}\big),
\end{equation}
where $\omega_Y \simeq \CC[d]$ is the (constant) dualizing complex of the compact oriented
manifold $Y$ and $(p_X)_*$ is pushforward along $p_X$, i.e.\ ``integration over $Y$''.
\end{proposition}

\begin{remark}\label{rmk:transgression}
Here $\Map$ is the internal mapping object in derived stacks, and the manifold $Y$ enters only
through its Betti stack $Y_B$, which retains only its homotopy type. The pushforward $(p_X)_*$
in \eqref{eq:cotangent-map} is what is often called ``integration over the fibres''; in the
shifted symplectic setting of Section \ref{subsec:shifted-symp} it is the transgression map of
the AKSZ construction, and the factor $\omega_Y \simeq \CC[d]$ (Poincar\'e duality) is the
source of the shift by $-d$.
\end{remark}

\subsection{The derived character stack}

The central geometric object of our program is the derived moduli space of \emph{local
systems}, the ``character stack''. It is often described as the moduli of flat connections up to
gauge equivalence; the identification is a theorem rather than a definition (Remark
\ref{rmk:RH}).

\begin{definition}[Derived character stack]\label{def:char-stack}
Let $H$ be a reductive algebraic group over $\CC$ and let $\Sigma$ be a space with the homotopy
type of a finite CW complex. The \emph{derived character stack} (or \emph{derived moduli stack
of $H$-local systems}) is the derived mapping stack
\begin{equation}
\Loc_H(\Sigma) \;:=\; \Map(\Sigma_B,\, BH),
\end{equation}
where $\Sigma_B$ is the Betti stack (the constant derived stack associated to the homotopy type
of $\Sigma$) and $BH = [\mathrm{pt}/H]$ is the classifying stack of $H$. We apply this with
$H = \GC$, the complexification of the compact gauge group $G$ of the Chern--Simons theory.
\end{definition}

\begin{remark}\label{rmk:char-stack-caveats}
(a) Definition \ref{def:char-stack} requires a reductive \emph{algebraic} group and does not apply
to the compact group $G$ of Chern--Simons theory; this is why we use $\GC$ (Section
\ref{subsec:compact-vs-complex}). Note that $[G/G]$ and $[\GC/\GC]$ are different objects in
different geometric contexts, with different categories of quasi-coherent sheaves.

(b) Only the homotopy type of $\Sigma$ enters, through $\Sigma_B$, so $\Loc_H$ is defined on
homotopy types and is insensitive to smooth structures or metrics. The manifold hypotheses
imposed below (compactness, orientability, closedness) are needed only for the shifted
symplectic structure of Section \ref{subsec:shifted-symp}.
\end{remark}

\begin{remark}[Riemann--Hilbert]\label{rmk:RH}
The stack $\Loc_{\GC}(\Sigma)$ is not by definition a moduli space of flat connections. For
$\Sigma$ a compact manifold, the Riemann--Hilbert correspondence identifies, via holonomy, the
groupoid of $\GC$-bundles with flat connection on $\Sigma$ with the groupoid of representations
$\pi_1(\Sigma) \to \GC$ up to conjugacy. This is an equivalence of topological (or $C^\infty$)
stacks, and it justifies passing between the field-theoretic and the algebro-geometric
vocabulary; it does not give an isomorphism of algebraic stacks with a differential-geometric
moduli space. (This is the distinction between the Betti and de Rham versions of the character
stack; see Section \ref{subsec:langlands}.) When we pass between ``flat connection'' and ``point
of $\Loc_{\GC}$'' below, we use this correspondence.
\end{remark}

The mapping stack construction produces a derived stack that encodes, in a single object, the space of representations of $\pi_1(\Sigma)$ into $G$, the gauge equivalences between them, the deformation theory around each representation, and the obstructions to extending infinitesimal deformations. We unpack this in the following proposition.

\begin{proposition}[{\cite[\S 2.1]{PTVV2013}}, {\cite[\S 1--2]{BZN2018}}, {\cite[\S 3]{Calaque2015}}]\label{prop:Loc-explicit}
The derived character stack $\Loc_{\GC}(\Sigma)$ has the following concrete descriptions:
\begin{enumerate}[label=(\roman*)]
\item \emph{Points.} The groupoid of $\CC$-points of $\Loc_{\GC}(\Sigma)$ is the quotient
groupoid\\ $[\Hom(\pi_1(\Sigma), \GC)/\GC]$; in particular its set of isomorphism classes is the
set of conjugacy classes of representations $\pi_1(\Sigma) \to \GC$. We reserve the term
\emph{character variety} for the affine GIT quotient $\Hom(\pi_1(\Sigma),\GC)/\!/\GC$, which is
a different (and coarser) object; see (iv).

\item \emph{Tangent complex.} At a point $[\rho] \in \Loc_{\GC}(\Sigma)$ corresponding to a
representation $\rho: \pi_1(\Sigma) \to \GC$, the tangent complex is the shifted cohomology
complex $\mathbb{T}_{[\rho]}\, \Loc_{\GC}(\Sigma) \simeq C^\bullet(\Sigma;\, \fg_{\Ad_\rho})[1]$,
where $\fg_{\Ad_\rho}$ is the adjoint local system determined by $\rho$. In particular, $H^0$ of
the tangent complex (first-order deformations) is $H^1(\Sigma;\, \fg_{\Ad_\rho})$, $H^1$
(obstructions) is $H^2(\Sigma;\,\fg_{\Ad_\rho})$, and $H^{-1}$ (infinitesimal automorphisms) is
$H^0(\Sigma;\, \fg_{\Ad_\rho})$, the centralizer of the representation. This is
Proposition \ref{prop:cotangent-map} applied to $Z = B\GC$, using $\mathbb{T}_{B\GC} = \fg[1]$.

\item \emph{The circle.} For $\Sigma = S^1$, the fundamental group is $\ZZ$, and a
representation $\ZZ \to \GC$ is determined by the image of the generator. Conjugation acts by
the adjoint action, so $\Loc_{\GC}(S^1) = [\GC/\GC]$ is the adjoint quotient stack (which
differs from the differential-geometric $[G/G]$; see Remark \ref{rmk:char-stack-caveats}(a)).

\item \emph{Closed surfaces.} For a closed surface $\Sigma_g$ of genus $g$, the fundamental
group has the standard presentation with $2g$ generators $a_1, b_1, \ldots, a_g, b_g$ and one
relation $\prod_{i=1}^g [a_i, b_i] = 1$. Consequently the classical truncation is the quotient
\emph{stack}
\[
\Loc_{\GC}(\Sigma_g)_{\mathrm{cl}} \;=\; \big[\,\mu^{-1}(1)\,/\,\GC\,\big],
\qquad \mu: \GC^{2g} \to \GC, \quad \mu(A_i,B_i) = \textstyle\prod_i [A_i,B_i],
\]
parametrizing tuples satisfying the relation, together with their automorphisms. It is
\emph{not} the GIT quotient $\mu^{-1}(1)/\!/\GC$. The GIT quotient is the affinization of the
stack: it sees only the closed orbits (the polystable, i.e.\ semisimple, representations), it
identifies all representations with the same semisimplification, and it forgets the isotropy
groups. The distinction matters here, since the interesting perturbative structure is located
where the isotropy jumps (Section \ref{subsec:theta-strat}). The derived structure enhances the stack further by remembering the
full deformation--obstruction complex at each point.
\end{enumerate}
\end{proposition}

\subsection{Shifted symplectic structures}\label{subsec:shifted-symp}

The central result connecting derived character stacks to both symplectic geometry and the BV formalism is the theory of \emph{shifted symplectic structures}, developed by Pantev, To\"en, Vaqui\'e, and Vezzosi in \cite{PTVV2013}. The idea is that on a derived stack, symplectic forms can exist in degrees other than zero, and the degree is dictated by the dimension of the source manifold.

\begin{definition}[$n$-shifted symplectic structure; PTVV \cite{PTVV2013}]\label{def:shifted-symplectic}
An \emph{$n$-shifted symplectic structure} on a derived stack $X$ is a closed 2-form $\omega \in \mathcal{A}^{2,\mathrm{cl}}(X, n)$ of cohomological degree $n$ in the de Rham complex of $X$, such that the induced map $\omega^\sharp: \mathbb{T}_X \to \mathbb{L}_X[n]$ is an equivalence of quasi-coherent sheaves. The non-degeneracy condition generalizes the classical requirement that a symplectic form induce an isomorphism between tangent and cotangent bundles.
\end{definition}

By the following theorem, the character stack of a closed $d$-dimensional manifold carries a natural symplectic structure of degree $2-d$, and the three cases $d = 1, 2, 3$ correspond to the three types of structure (BFV, classical symplectic, BV) that appear in the BV-BFV formalism.

\begin{theorem}[PTVV {\cite[Thm.~2.5]{PTVV2013}}]\label{thm:PTVV}
Let $\GC$ be reductive and let $\Sigma$ be a \emph{closed} (i.e.\ compact without boundary)
oriented $d$-dimensional manifold. Then the derived character stack $\Loc_{\GC}(\Sigma)$ carries
a canonical $(2-d)$-shifted symplectic structure. Closedness is needed: the construction
transgresses the $2$-shifted form on $B\GC$ against the fundamental class
$[\Sigma] \in H_d(\Sigma)$, and non-degeneracy of the result is Poincar\'e duality for $\Sigma$.
If $\partial\Sigma \neq \emptyset$ then Poincar\'e--Lefschetz duality relates $H^\bullet(\Sigma)$
to $H^{d-\bullet}(\Sigma,\partial\Sigma)$ rather than to $H^{d-\bullet}(\Sigma)$, the
transgressed form is degenerate, and one obtains instead a Lagrangian structure on the
restriction map to the boundary (Theorem \ref{thm:AKSZ-PTVV}). The
three cases relevant to Chern--Simons theory are:
\begin{enumerate}[label=(\roman*)]
\item \emph{$d = 1$ (circles):} $\Loc_{\GC}(S^1) = [\GC/\GC]$ carries a $1$-shifted symplectic structure. At a point $[g]$, the tangent complex is the two-term complex $\fg \xrightarrow{\Ad_g - 1} \fg$ in degrees $0$ and $1$, and the symplectic form is given by the Killing form on $\fg$ composed with the Poincar\'e duality pairing of $S^1$. The $1$-shifted degree reflects the fact that the symplectic pairing takes values in degree $1$, corresponding to the ``odd'' nature of the BFV phase space.

\item \emph{$d = 2$ (surfaces):} $\Loc_{\GC}(\Sigma_g)$ carries a $0$-shifted symplectic structure. This is a symplectic structure in the ordinary sense (on the smooth locus of the classical truncation), and it recovers the classical Atiyah--Bott--Goldman symplectic form given by
\begin{equation}\label{eq:AB-form}
\omega_{AB}([\alpha],\, [\beta]) = \int_\Sigma \Tr_\fg\,(\alpha \wedge \beta),
\end{equation}
where $\Tr_\fg$ denotes the invariant bilinear form on $\fg$ induced by the trace in the
\emph{adjoint} representation (i.e.\ the Killing form, suitably normalized), and not the trace
in some auxiliary defining representation; we write $\Tr_\fg$ throughout when this is what is
meant.
for tangent vectors $[\alpha], [\beta] \in H^1(\Sigma;\, \fg_{\Ad_\rho}) \cong T_{[\rho]}\, \Loc_{\GC}(\Sigma)_{\mathrm{cl}}$. This is the symplectic structure on the phase space that Chern--Simons theory assigns to a spatial slice.

\item \emph{$d = 3$ (3-manifolds):} $\Loc_{\GC}(M)$ carries a $(-1)$-shifted symplectic structure. A $(-1)$-shifted symplectic structure is the derived-geometric form of a BV symplectic structure in the sense of Definition \ref{def:BV-manifold}: it pairs degree $k$ fields with degree $(-1-k)$ antifields. This is the derived-geometric origin of the BV structure on the space of solutions to the Chern--Simons equations of motion.
\end{enumerate}
\end{theorem}

\begin{figure}[ht]
\centering
\begin{tikzpicture}[scale=0.95]

% === Column 1: d=1, circle ===
\begin{scope}[shift={(-5,0)}]
  % Circle drawing
  \draw[violet!60!black, very thick] (0,2.8) circle (0.6);
  \node[violet!60!black, font=\footnotesize\bfseries] at (0,3.8) {$d = 1$: circle};

  % Box
  \node[draw=violet!60!black, thick, rounded corners=4pt, fill=violet!5,
        inner sep=5pt, font=\footnotesize, align=center,
        minimum width=3.2cm] (b1) at (0,1.2)
        {$\Loc_{\GC}(S^1) = [\GC/\GC]$\\[3pt]
         \textbf{1-shifted symplectic}};

  % BFV label
  \node[draw=blue!50, thick, rounded corners=3pt, fill=blue!5,
        inner sep=4pt, font=\scriptsize, align=center,
        minimum width=3.2cm] at (0,-0.3)
        {BFV phase space\\[1pt]
         degree $0$ symplectic\\[1pt]
         $\to$ \textbf{category} $\mathcal{B}_{\CS}^{(k)}$};
\end{scope}

% === Column 2: d=2, surface ===
\begin{scope}[shift={(0,0)}]
  % Surface drawing (genus 1)
  \draw[violet!60!black, very thick] (0,2.8) ellipse (0.8 and 0.45);
  \fill[violet!10, very thick] (0,2.8) ellipse (0.8 and 0.45);
  \draw[violet!60!black, thick] (-0.25,2.8) arc(180:360:0.25 and 0.12);
  \node[violet!60!black, font=\footnotesize\bfseries] at (0,3.8) {$d = 2$: surface};

  % Box
  \node[draw=violet!60!black, thick, rounded corners=4pt, fill=violet!5,
        inner sep=5pt, font=\footnotesize, align=center,
        minimum width=3.2cm] (b2) at (0,1.2)
        {$\Loc_{\GC}(\Sigma_g)$\\[3pt]
         \textbf{0-shifted symplectic}};

  % Classical label
  \node[draw=blue!50, thick, rounded corners=3pt, fill=blue!5,
        inner sep=4pt, font=\scriptsize, align=center,
        minimum width=3.2cm] at (0,-0.3)
        {Atiyah--Bott--Goldman\\[1pt]
         classical symplectic form\\[1pt]
         $\omega_{AB} = \int_\Sigma \mathrm{Tr}(\alpha \wedge \beta)$};
\end{scope}

% === Column 3: d=3, 3-manifold ===
\begin{scope}[shift={(5,0)}]
  % 3-manifold drawing (blob)
  \draw[violet!60!black, very thick, rounded corners=8pt]
       (-0.7,2.3) -- (-0.5,3.3) -- (0.5,3.3) -- (0.7,2.3) -- cycle;
  \fill[violet!10, rounded corners=8pt]
       (-0.7,2.3) -- (-0.5,3.3) -- (0.5,3.3) -- (0.7,2.3) -- cycle;
  \node[violet!60!black, font=\footnotesize\bfseries] at (0,3.8) {$d = 3$: 3-manifold};

  % Box
  \node[draw=violet!60!black, thick, rounded corners=4pt, fill=violet!5,
        inner sep=5pt, font=\footnotesize, align=center,
        minimum width=3.2cm] (b3) at (0,1.2)
        {$\Loc_{\GC}(M)$\\[3pt]
         \textbf{$(-1)$-shifted symplectic}};

  % BV label
  \node[draw=blue!50, thick, rounded corners=3pt, fill=blue!5,
        inner sep=4pt, font=\scriptsize, align=center,
        minimum width=3.2cm] at (0,-0.3)
        {BV symplectic structure\\[1pt]
         degree $-1$ antibracket\\[1pt]
         $\to$ \textbf{partition function}};
\end{scope}

% Connecting formula
\node[draw=black!50, thick, rounded corners=4pt, fill=yellow!8,
      inner sep=6pt, font=\small, align=center] at (0,-1.8)
      {PTVV: $\;\Loc_{\GC}(\Sigma)$ carries a canonical $(2-d)$-shifted symplectic structure};

\end{tikzpicture}
\caption{The PTVV hierarchy of shifted symplectic structures on character stacks (Theorem \ref{thm:PTVV}). A compact oriented $d$-manifold $\Sigma$ gives rise to $\Loc_{\GC}(\Sigma)$ with a $(2-d)$-shifted symplectic structure. The three cases $d = 1, 2, 3$ correspond to the BFV, classical symplectic, and BV structures in the BV-BFV formalism. This dimensional pattern, produced by a single mechanism (AKSZ transgression of the Killing form via Poincar\'e duality), is the geometric foundation of the program.}
\label{fig:PTVV-hierarchy}
\end{figure}
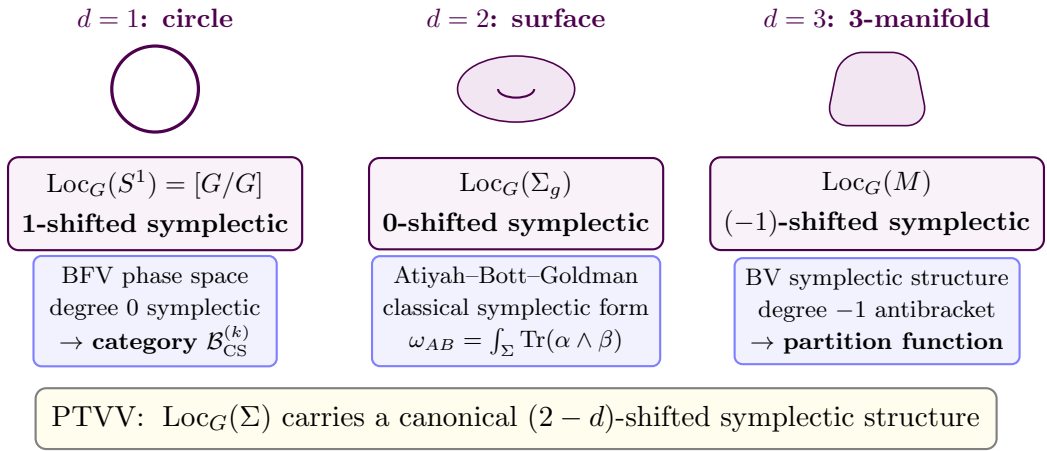

\begin{remark}[The PTVV mechanism]\label{rmk:PTVV-mechanism}
The shifted symplectic structure on\\ $\Loc_{\GC}(\Sigma) = \Map(\Sigma, BG)$ is constructed by a derived-geometric version of the \emph{AKSZ transgression}, using three ingredients:
\begin{enumerate}[label=(\alph*)]
\item The invariant pairing on $\fg$ (the Killing form), which endows the classifying stack $BG = [\mathrm{pt}/G]$ with a \emph{$2$-shifted symplectic structure}. This is a ``symplectic form of degree $2$'' in the sense that it pairs the tangent complex $\mathbb{T}_{BG} = \fg[1]$ with itself using the Killing form, producing a closed 2-form in $\mathcal{A}^{2,\mathrm{cl}}(BG, 2)$.

\item The \emph{Poincar\'e duality} of the source manifold $\Sigma$. For a compact oriented $d$-manifold, Poincar\'e duality provides a non-degenerate pairing $H^k(\Sigma) \otimes H^{d-k}(\Sigma) \to \CC$, which shifts the degree of the symplectic structure by $-d$.

\item The \emph{AKSZ transgression}: the derived mapping stack $\Map(\Sigma, BG)$ inherits a $(2-d)$-shifted symplectic structure by ``integrating'' the $2$-shifted structure on $BG$ over the fundamental class of $\Sigma$. This integration is the derived-geometric analogue of the classical AKSZ construction, which produces BV symplectic structures on mapping spaces between a source manifold with a volume form and a target with a symplectic structure.
\end{enumerate}
Thus a single mechanism (transgression of the Killing form via Poincar\'e duality) produces the
BFV structure in dimension 1, the Atiyah--Bott form in dimension 2, and the BV structure in
dimension 3. Calaque \cite{Calaque2017} shows that $T^*[n]X$ carries a canonical $n$-shifted
symplectic structure. By the Darboux theorems of \cite{BravBussiJoyce2019,BBBJ2015},
$(-1)$-shifted symplectic stacks are (smooth-)locally derived critical loci $\mathrm{Crit}(f)$, the
case $f = 0$ being $T^*[-1]U$; we use this in Remark \ref{rmk:schefers}.
\end{remark}

\subsection{Lagrangian structures and the BV-BFV correspondence}\label{subsec:BV-BFV-dictionary}

The PTVV theorem tells us that the character stacks of the bulk and the boundary each carry shifted symplectic structures, of degrees $-1$ and $0$ respectively. The final ingredient from derived geometry is the observation that the relationship between bulk and boundary, i.e., the restriction of flat connections from $M$ to $\partial M$, is captured by a \emph{Lagrangian structure} on the restriction map. This provides a precise derived-geometric reformulation of the BV-BFV axiom.

\begin{definition}[Lagrangian morphism; PTVV \cite{PTVV2013}]\label{def:Lagrangian}
Let $(X, \omega)$ be an $n$-shifted symplectic derived stack. A morphism $f: L \to X$ is \emph{Lagrangian} if it is equipped with an isotropic structure,
i.e., a homotopy $f^*\omega \sim 0$ witnessing the vanishing of the pullback of the symplectic
form, such that the induced map $\mathbb{T}_{L/X} \to \mathbb{L}_L[n-1]$ is an equivalence.
\end{definition}

The isotropic structure is a datum (a chosen nullhomotopy of $f^*\omega$), and different choices
give different Lagrangian structures. It does not by itself say that $L$ is ``half-dimensional'':
trivializing $f^*\omega$ imposes a kind of upper bound on the size of $L$, and the
non-degeneracy condition, which makes the structure maximal, replaces the classical dimension
count. (For derived stacks, dimension is a virtual notion.)

The key theorem connecting this to the BV-BFV formalism is the following.

\begin{theorem}[Calaque \cite{Calaque2015}]\label{thm:AKSZ-PTVV}
Let $M$ be a compact oriented 3-manifold with boundary $\partial M = \Sigma$. The restriction map
\begin{equation}
\mathrm{res}: \Loc_{\GC}(M) \longrightarrow \Loc_{\GC}(\Sigma)
\end{equation}
carries a canonical Lagrangian structure, where $\Loc_{\GC}(M)$ has its $(-1)$-shifted symplectic structure and $\Loc_{\GC}(\Sigma)$ has its $0$-shifted symplectic structure.
\end{theorem}

This theorem suggests a derived-geometric reinterpretation of the BV-BFV formalism for
Chern--Simons theory, summarized in the following table. \emph{Throughout, $M$ is a compact oriented $3$-manifold with (possibly empty)
boundary and $\Sigma$ is a closed oriented surface}; the closedness of $\Sigma$ is required by
Theorem \ref{thm:PTVV}, and the compactness and orientability of $M$ by Theorem
\ref{thm:AKSZ-PTVV}. The third column indicates which lines hold only after replacing the
compact group $G$ by $\GC$ (Section \ref{subsec:compact-vs-complex}).

\begin{center}
\renewcommand{\arraystretch}{1.25}
\begin{tabular}{p{4.1cm}|p{5.4cm}|p{3.6cm}}
\textbf{BV-BFV (compact $G$)} & \textbf{Derived geometry (reductive $\GC$)} & \textbf{Status of the match} \\ \hline
BV symplectic (bulk), on fields over $M$ &
$(-1)$-shifted symplectic on $\Loc_{\GC}(M)$ &
after complexification; \emph{and} after Riemann--Hilbert (Rmk.~\ref{rmk:RH}) \\
BFV symplectic (boundary), on fields over $\Sigma$ closed &
$0$-shifted symplectic on $\Loc_{\GC}(\Sigma)$ &
same caveats; the form itself matches on the smooth locus \\
BV-BFV axiom \eqref{eq:BV-BFV-axiom-2} &
Lagrangian structure on $\mathrm{res}: \Loc_{\GC}(M) \to \Loc_{\GC}(\partial M)$ &
structurally parallel; a proof of the equivalence is Step 1, Prop.~\ref{prop:step1} \\
BV action, CME &
derived critical locus of $\CS$ &
literal for $\GC$; for compact $G$ the functional is $\RR/\ZZ$-valued \\
gauge-fixing (choice of $\mathcal{L}$, metric) &
--- (no counterpart) &
\emph{not} a Lagrangian in $\Loc_{\GC}(\Sigma)$; see below \\
level $k \in \ZZ$ &
formal parameter $\hbar$, $q = e^{\hbar/2}$ &
integrality is a compact-group phenomenon (\S\ref{subsec:compact-vs-complex}(c))
\end{tabular}
\end{center}

\medskip

Two entries deserve comment. Gauge-fixing has no counterpart on the right-hand side: it is the
choice of a Lagrangian submanifold in the infinite-dimensional space of \emph{fields}
$\mathcal{F}_M$ and involves a Riemannian metric, whereas the derived character stack is the
space of \emph{solutions}; a Lagrangian in $\Loc_{\GC}(\Sigma)$ would rather enter a geometric
quantization of the boundary. And the correspondence between the BV-BFV axiom and the Lagrangian
structure is a structural parallel, which Step 1 below makes precise only at the level of
tangent complexes.

On the left of the table is the differential-geometric, perturbative framework of the BV-BFV
formalism; on the right, the algebro-geometric framework of derived symplectic geometry. After
complexification they encode the same moduli problem (flat connections on $M$ restricting to
flat connections on $\Sigma$), and the question is whether they also agree at the quantum level.
The rest of the paper develops a program for this comparison, with factorization homology as
the mechanism. The quantization of shifted coisotropic and Lagrangian structures is developed in \cite{MS2018}; see also \cite{Safronov2020} for a shifted geometric quantization perspective.

\subsection{Deformation quantization of shifted symplectic stacks}

We now turn to quantization. The shifted symplectic structures on character stacks can be
deformation-quantized, and for the classifying stack the output is the representation category
of the quantum group. The relevant results live in different settings, and the objects being
deformed belong to different categories, so we state them separately and then explain (Remark
\ref{rmk:what-quantizes}) what is quantized in each case.

\begin{theorem}[Shifted deformation quantization; CPTVV {\cite[\S\S 3.5--3.6]{CPTVV2017}}]\label{thm:CPTVV}
Let $X$ be a derived Artin stack, locally of finite presentation, equipped with an $n$-shifted
Poisson structure with $n \geq 1$. Then $X$ admits a deformation quantization: a
$\CC[\![\hbar]\!]$-linear deformation of the $\infty$-category $\Perf(X)$ of perfect complexes
as an $\EE_{n}$-monoidal $\infty$-category. For $n \le 0$ the statement requires additional
input and is not available in this generality.
\end{theorem}

\begin{theorem}[Poisson--Lie structures are $1$-shifted Poisson structures; Safronov {\cite[Thm.~1.1, \S 4]{Safronov2017b}}, {\cite[\S 5]{Safronov2018}}]\label{thm:Safronov-PL}
Let $\GC$ be a reductive algebraic group with Lie algebra $\fg$. There is an equivalence between
Poisson--Lie structures on $\GC$ (equivalently, Lie bialgebra structures on $\fg$) and
$1$-shifted Poisson structures on the classifying stack $B\GC$. In particular, the standard
quasi-triangular Lie bialgebra structure on $\fg$, whose classical $r$-matrix $r \in \fg \otimes
\fg$ has symmetric part the Casimir element of the invariant form, determines a $1$-shifted
Poisson structure on $B\GC$. Moreover, Poisson additivity \cite{Safronov2018} relates
$n$-shifted Poisson structures to $\EE_n$-structures, which is what converts the geometry into
algebra.
\end{theorem}

\begin{theorem}[Quantization of Lie bialgebras; Drinfeld \cite{Drinfeld1987}, Etingof--Kazhdan {\cite[Thm.~1.1]{EK1996}}]\label{thm:EK}
Every Lie bialgebra $(\fg,\delta)$ over a field of characteristic zero admits a quantization: a
Hopf algebra $U_\hbar(\fg)$ over $\CC[\![\hbar]\!]$, flat as a $\CC[\![\hbar]\!]$-module, with
$U_\hbar(\fg)/\hbar U_\hbar(\fg) \cong U(\fg)$ and with cobracket $\delta$ recovered to first
order in $\hbar$. The construction is functorial in $(\fg,\delta)$. For $\fg$ simple with the
standard quasi-triangular structure this recovers the Drinfeld--Jimbo quantum group
$U_q(\fg)$, $q = e^{\hbar/2}$, together with its universal $R$-matrix; consequently
$\Rep_q(\GC) := \Rep\,U_q(\fg)$ is a braided monoidal category, i.e.\ an $\EE_2$-algebra in
linear categories.
\end{theorem}

\begin{theorem}[Integration over surfaces; BBJ {\cite[Thms.~1.1, 6.9]{BBJ2018}}]\label{thm:BBJ-int}
For a closed oriented surface $\Sigma_g$, the factorization homology $\iint_{\Sigma_g}\Rep_q(\GC)$
is a deformation of\\ $\QCoh(\Loc_{\GC}(\Sigma_g))$ over $\CC[\![\hbar]\!]$ quantizing the
$0$-shifted Poisson structure of Theorem \ref{thm:PTVV}(ii).
\end{theorem}

\begin{remark}[What is quantized]\label{rmk:what-quantizes}
In all cases the object being deformed is a category of sheaves, with a monoidal structure
whose type is dictated by the shift (Theorem \ref{thm:CPTVV}).

(i) \emph{$B\GC = \Loc_{\GC}(D^2)$.} Here $\QCoh(B\GC) \simeq \Rep(\GC)$, and the $2$-shifted
structure deforms it as an $\EE_2$-monoidal, i.e.\ braided monoidal, category; the result is
$\Rep_q(\GC)$. The quasi-triangular Hopf algebra $U_q(\fg)$ is a presentation of this braided
category by means of its fibre functor: a quasi-triangular Hopf algebra gives rise to an
$\EE_2$-algebra in linear categories (its representation category, braided by the $R$-matrix),
whereas the converse requires a fibre functor and Tannakian reconstruction. Only the first
direction is used here; in particular an $\EE_2$-algebra in categories is not the same thing as
a braided Hopf algebra.

(ii) \emph{$[\GC/\GC] = \Loc_{\GC}(S^1)$.} The $1$-shifted structure deforms $\QCoh([\GC/\GC])$ as
a monoidal ($\EE_1$-) category. The result is \emph{not} $\Rep_q(\GC)$ but the category
$\iint_{S^1\times\RR}\Rep_q(\GC)$ that factorization homology assigns to the annulus, which
Ben-Zvi--Brochier--Jordan identify with modules over the quantum coordinate algebra
$\mathcal{O}_q(\GC)$ (the reflection equation algebra) in $\Rep_q(\GC)$ \cite{BBJ2018b}; there is
a central monoidal functor $\Rep_q(\GC) \to \iint_{S^1\times\RR}\Rep_q(\GC)$.

(iii) \emph{Closed surfaces.} The $0$-shifted structure on $\Loc_{\GC}(\Sigma)$ deforms
$\QCoh(\Loc_{\GC}(\Sigma))$ as a category, and the deformation is computed by factorization
homology of the answer to (i) (Theorem \ref{thm:BBJ-int}).

These statements are compatible: quantizing the character stack of a surface can be computed
either directly or by quantizing the local datum and integrating, the compatibility being
Poisson additivity together with $\otimes$-excision. This compatibility is what makes the
program of this paper possible.
\end{remark}

\begin{remark}[Local data]\label{rmk:local-data}
The local datum of the theory is $\Rep_q(\GC)$, the quantization of $B\GC = \Loc_{\GC}(D^2)$. It
is what the $(3\text{-}2\text{-}1)$-extended theory assigns to the circle, equivalently what the
fully extended theory assigns to a point, and everything else (surfaces, cobordisms, the RT
invariants) is obtained from it by factorization homology. This is why Conjecture
\ref{conj:E2-equiv}, a statement about the disk, has the consequences we draw from it. The
circle stack $[\GC/\GC]$ enters through its $1$-shifted Lagrangians, which label boundary
conditions (Remark \ref{rmk:bc-derived}).
\end{remark}

\begin{theorem}[Ben-Zvi--Brochier--Jordan \cite{BBJ2018}]\label{thm:BBJ-quant}
The deformation quantization of the character stack $\Loc_{\GC}(\Sigma)$ is computed by factorization homology of the quantized local data, so we have
\begin{equation}\label{eq:quant-as-FH}
\Quant\big(\Loc_{\GC}(\Sigma)\big) \;\simeq\; \iint_\Sigma \Rep_q(\GC).
\end{equation}
\end{theorem}

\begin{remark}
This is the connection with the RT side. The RT state spaces arise from quantizing the same geometric object (the character stack of
$\GC$) that underlies the complexified BV-BFV formalism; the question is whether the two
quantization procedures agree. At a root of unity, after semisimplification, the factorization homology of a closed
surface $\Sigma_g$ yields $(\mathbf{H},\, u_{\Sigma_g})$ with $u_{\Sigma_g} \cong
V_{\mathrm{RT}}(\Sigma_g)$ (Theorem \ref{thm:AKZ}).
\end{remark}

\begin{remark}[Why semisimplify?]\label{rmk:why-semisimplify}
We explain why roots of unity and semisimplification enter.

\emph{(a) Finiteness.} At generic $q$, $\Rep_q(\GC)$ has infinitely many simple objects and the
sums appearing in the RT surgery formula \eqref{eq:surgery} and in the Verlinde formula
\eqref{eq:Verlinde} do not converge. Specializing $q$ to a root of unity and passing to the
semisimple quotient produces a category with finitely many simple objects, labelled by the
Weyl alcove, which makes the RT sums finite.

\emph{(b) Non-degeneracy.} Invariance under the second Kirby move requires invertibility of the
$S$-matrix (Theorem \ref{thm:RT-invariance}(ii)), i.e.\ modularity. At a root of unity, before
semisimplification, the category has transparent objects and the $S$-matrix is degenerate;
quotienting by negligible morphisms removes them, which makes the surgery formula well-defined.

\emph{(c) Derived geometry.} We do not know a derived-geometric reason for semisimplification.
From the point of view of deformation quantization, the natural output is the category at
generic (or formal) $q$, and the two subsequent operations (specialization to a root of unity,
and semisimplification) are additional. The first is an algebraization statement (Remark \ref{rmk:q-parameter}); the second
is a categorical localization whose compatibility with factorization homology is one of the
open problems of this program (Section \ref{sec:outlook}(c)). There is a natural expectation
that semisimplification corresponds to imposing a support or singular-support condition on the
sheaf side, in the spirit of Arinkin--Gaitsgory \cite{AG2015}, and the non-semisimple TQFT
literature \cite{DGGPR2022,CGPM2023,Haioun2025} shows that one can do interesting topology
without semisimplifying at all.
\end{remark}

\begin{remark}[Quantum character varieties]\label{rmk:QCV}
The term \emph{quantum character variety}, introduced by Ben-Zvi--Brochier--Jordan
\cite{BBJ2018}, refers to the factorization homology $\iint_\Sigma \Rep_q(\GC)$; see Remark
\ref{rmk:QCV-name} for the caveat about this name. The content of Theorem \ref{thm:BBJ-quant} is
that this coincides with the deformation quantization $\Quant(\Loc_{\GC}(\Sigma))$ of the derived
character stack. In the diagram of Conjecture \ref{conj:main}, the arrow labeled ``BBJ'' is this
identification; the substantive open problem is the arrow labeled ``Step 2,'' which asserts that
the BV-BFV quantization agrees with the algebro-geometric deformation quantization
$\Quant(\Loc_{\GC}(\Sigma))$.
\end{remark}

\begin{remark}[The classical limit is a Tannakian statement]\label{rmk:Tannakian}
The classical limit of the previous theorems is the assertion that factorization homology of
$\Rep(\GC)$ over $\Sigma$ agrees with quasi-coherent sheaves on the character stack. This is not
an obvious equivalence; the mechanism is Tannakian. The starting point is the equivalence
\[
\Rep(\GC) \;\simeq\; \QCoh(B\GC),
\]
which presupposes the theory of algebraic stacks (one needs $B\GC$ as a stack and $\QCoh$ on
it). Given this, $\QCoh(\Loc_{\GC}(\Sigma)) = \QCoh(\Map(\Sigma_B,B\GC))$ is
computed by descent along the disk cover of $\Sigma$, which is the colimit defining
factorization homology; this is the classical limit of Theorem \ref{thm:BBJ-int}. Recovering the
stack from the tensor category is Tannakian reconstruction in the derived setting (Lurie
\cite{Lurie_SAG}, \cite{BZFN2010,Wallbridge2012}). Statements identifying a factorization homology
category with a category of sheaves on a moduli space are therefore Tannakian in nature.
\end{remark}

\begin{remark}[The quantization parameter]\label{rmk:q-parameter}
The deformation quantization of Theorems \ref{thm:CPTVV}--\ref{thm:EK} produces formal
deformations over the power series ring $\CC[\![\hbar]\!]$, with $q = e^{\hbar/2}$. The modular
tensor category, however, lives at a specific non-formal value, $\hbar = 2\pi i/(k+h^\vee)$,
i.e.\ $q = e^{i\pi/(k+h^\vee)}$.

Deformation quantization produces formal families only and does not by itself allow one to
evaluate the formal parameter. The passage is provided instead by a classical fact due to
Drinfeld \cite{Drinfeld1987} (see also \cite{CP1994,Lusztig1990}): for $\fg$ simple there is an
\emph{algebraic} one-parameter family, the Drinfeld--Jimbo quantum group $U_q(\fg)$ over
$\CC[q,q^{-1}]$, which is defined for all $q$ and which can therefore simply be evaluated at a
root of unity; and the formal expansion of this family around $q = 1$ recovers the
deformation-quantization answer $U_\hbar(\fg)$.

For our program the issue is therefore not whether an algebraization exists, but whether the
formal family that the \emph{BV-BFV} side produces is the Taylor expansion of that
same algebraic family, rather than of some other formal deformation with the same first-order
term. Formal deformations of $\Rep(\GC)$ as a braided monoidal category form a moduli problem
with non-trivial higher structure, and the Etingof--Kazhdan quantization is only unique up to
the action of Drinfeld associators; matching the BV-BFV answer to the Drinfeld--Jimbo family
therefore requires identifying the associator produced by the Chern--Simons propagator with the
Knizhnik--Zamolodchikov associator. This is Step 2b below.
\end{remark}

\begin{remark}[Comparison with geometric quantization]\label{rmk:geom-quant}
There is an alternative, more classical approach to quantizing the character variety. It can be done by \emph{geometric quantization}. One chooses a complex structure on $\Sigma$, which identifies $\Loc_{\GC}(\Sigma)$ with a moduli space of holomorphic bundles, and then defines the quantum Hilbert space as the space of holomorphic sections of a certain line bundle (the determinant bundle at level $k$). Different choices of complex structure give a priori different Hilbert spaces, but the \emph{Hitchin connection} provides a projectively flat connection on the resulting bundle of Hilbert spaces over Teichm\"uller space, ensuring that the quantization is independent of the complex structure up to projective equivalence. This produces the same projective representations of $\MCG(\Sigma_g)$ as the RT theory.
The deformation quantization approach of Theorems \ref{thm:CPTVV}--\ref{thm:EK} has the
advantage of being algebraic (no choice of complex structure is needed) and of being naturally compatible with factorization homology. On the other hand, geometric
quantization gives more direct access to the Hilbert space structure (inner products,
unitarity). The geometric quantization approach was initiated by Hitchin \cite{Hitchin1990}; for
further developments connecting it to the RT representations, see \cite{Andersen2006,Andersen2012,AU2015}.

The polarizations involved in the two pictures are of different types. Geometric quantization of the moduli space of
flat connections, in the Hitchin picture, uses a \emph{holomorphic} (K\"ahler) polarization: one
chooses a complex structure on $\Sigma$, obtains a complex structure on $\Loc(\Sigma)$, and
takes holomorphic sections of the determinant line bundle. By contrast, the polarizations that
arise in the BV-BFV/Lagrangian-correspondence picture (for instance the Lagrangian
$\Loc(M) \to \Loc(\partial M)$ attached to a bulk $3$-manifold $M$, or the Bohr--Sommerfeld
fibres used in Proposition \ref{prop:genus-one}(iv)) are \emph{real} Lagrangian fibrations. Real
and holomorphic polarizations of the same symplectic manifold need not give the same
quantization, and comparing them requires work; the standard mechanism is a degeneration of the complex structure to the real
polarization together with an index-theoretic count of Bohr--Sommerfeld fibres, and even for the
torus this requires care (compare \cite{JeffreyWeitsman1992,Andersen2006}). We do not claim that
the two match in general. Establishing a rigorous equivalence between the
two quantization procedures, beyond the level of projective mapping class group
representations, remains an important open problem and a consistency check for our program.
\end{remark}

\section{Main Conjectures and Proof Strategy}\label{sec:bridge}

We now assemble the ingredients of the previous sections into precise conjectural statements and
develop the strategy for establishing them. The logic is as follows: from the BV-BFV
quantization of Chern--Simons theory, we extract an $\EE_2$-category; we conjecture that this
category is equivalent to the modular tensor category used in the RT construction; and we
explain how this equivalence, together with factorization homology and an additional input on
cobordisms, would lead to the identification of the two TQFT functors.

\begin{quote}\small
\emph{Status of this section.} With the exception of Proposition \ref{prop:step1}, nothing in
this section is a theorem. Constructions \ref{constr:E2-from-BVBFV} and \ref{constr:BCS} are
\emph{heuristic}: they describe operations that we expect to exist in an extended BV-BFV
formalism which has not been constructed, since the standard CMR framework
\cite{CMR2014,CMR2018} assigns data to codimension-$1$ boundaries but not to codimension-$2$
corners, and the manifolds with corners that appear below (the complement of sub-disks in $D^2$,
times an interval) go beyond it (see also Remark \ref{rmk:extended-BVBFV}). The missing
ingredients (a derived BV-BFV formalism, a universal property for the BV pushforward, the
root-of-unity specialization, and the interaction of semisimplification with factorization
homology) are listed as open problems in Section \ref{sec:outlook}.
\end{quote}

\subsection{The $\EE_2$-category from BV-BFV}

Throughout this section and the Main Conjecture, we work in the $(3\text{-}2\text{-}1)$-extended TQFT framework of Definition \ref{def:extended-TQFT}, in which the circle $S^1$ is the lowest-dimensional stratum and is assigned an $\EE_2$-category. The fully extended $(3\text{-}2\text{-}1\text{-}0)$-framework (assigning data to points) is discussed separately in Section \ref{sec:outlook}.

The first step is to extract an $\EE_2$-monoidal category from the BV-BFV quantization of Chern--Simons theory. The key geometric input is the topology of the disk. The little 2-disks operad acts on configurations of sub-disks inside a larger disk, and this action translates into algebraic structure on the BFV state space.

\begin{construction}\label{constr:E2-from-BVBFV}
Consider Chern--Simons theory on the cylinder $D^2 \times [0,1]$, where $D^2$ is a disk. We take the following points into account.

\begin{enumerate}[label=(\roman*)]
\item To the spatial boundary circle $S^1 = \partial D^2$, the extended TQFT assigns a \emph{category} $\mathcal{B}_{\CS}^{(k)}$ of boundary conditions (not a vector space; the BFV state space is assigned to codimension-1 surfaces, while $S^1$ is codimension-2).

\item \textbf{Vertical composition.} Stacking two cylinders $D^2 \times [0,1]$ end to end
(gluing along a common $D^2$ face) gives a longer cylinder $D^2 \times [0,2]$. At the level of
partition functions, this gluing is computed by the BV-BFV pairing over the BFV state space
$\mathcal{H}_{D^2}$ of the intermediate disk. The existence and associativity of this pairing
for gluing along a \emph{closed} codimension-$1$ boundary component is proved in
\cite[\S 3.5]{CMR2018} and \cite[\S 6]{CMR2020} (see Theorem \ref{thm:BV-gluing}). Here $D^2$ has
a boundary, so the cylinder is a manifold with corners, and the required gluing statement is not
available in the literature. For fixed boundary conditions on the lateral circle
$S^1 = \partial D^2$, this would define composition of morphisms in $\mathcal{B}_{\CS}^{(k)}$.

\item \textbf{Horizontal juxtaposition.} Embedding two disjoint smaller disks $D^2_1 \sqcup
D^2_2$ into a larger disk $D^2$ and considering the partition function on $(D^2_1 \sqcup D^2_2)
\times [0,1]$ versus $D^2 \times [0,1]$ gives a multiplication, parametrized by the space of such
embeddings.
\end{enumerate}
One would like to show that the operations in (iii), together with their compatibility with the
composition in (ii), satisfy the relations of the little $2$-disks operad (the two coordinate
directions of $D^2$ providing two compatible $\EE_1$-structures, as in Dunn additivity), so that
by the recognition principle for $\EE_2$-algebras \cite[\S 5.2]{Lurie_HA} they endow
$\mathcal{B}_{\CS}^{(k)}$ with a braided monoidal structure. We do not know a proof of this, and
we are not aware of one in the literature; the closest available statements are in the factorization-algebra formulation of
perturbative field theory, where the $\EE_n$-structure of observables on $\RR^n$ is a theorem of
Costello--Gwilliam \cite[\S 3]{CG2017}, and in Rabinovich's construction of factorization
algebras for field theories on manifolds with boundary \cite{Rabinovich2021,Rabinovich2020}.
Deducing the compatibility from those would require knowing that the BV-BFV output on a manifold
with corners agrees with the factorization-algebra output, which is not known. Construction
\ref{constr:E2-from-BVBFV} is therefore a heuristic, and the assertion that it produces an
$\EE_2$-structure is part of the content of Conjecture \ref{conj:E2-equiv}.
\end{construction}

The central question is: which braided monoidal category does this construction produce?

\begin{construction}[The category $\mathcal{B}_{\CS}^{(k)}$]\label{constr:BCS}
We define the category $\mathcal{B}_{\CS}^{(k)}$ as the category of
boundary conditions for Chern--Simons theory on the disk $D^2$.
Consider CS theory on the 3-manifold $D^2 \times [0,1]$, where $D^2$
is the spatial slice and $[0,1]$ is the time direction. The full
boundary decomposes as
\[
\partial(D^2 \times [0,1]) = \underbrace{D^2 \times \{0,1\}}_{\text{temporal: incoming/outgoing}} \;\cup\; \underbrace{S^1 \times [0,1]}_{\text{spatial: lateral boundary}}.
\]
The BV-BFV formalism assigns state spaces to the temporal
boundary components $D^2 \times \{0,1\}$, but the partition function
depends on the choice of \emph{boundary condition} imposed on the
lateral boundary $S^1 \times [0,1]$. The \emph{objects} of
$\mathcal{B}_{\CS}^{(k)}$ are these boundary conditions. In the
extended TQFT language, they correspond to the data assigned to the
codimension-2 stratum $S^1 = \partial D^2$; physically, they are
the labels for Wilson lines ending on the boundary circle.
The \emph{morphisms} $\Hom(V, W)$ are defined by the BV-BFV
partition functions on the annulus $S^1 \times [0,1]$ interpolating
between boundary conditions $V$ and $W$, and composition is given by
the gluing formula (Theorem \ref{thm:BV-gluing}). The \emph{monoidal
structure} is defined by the pair-of-pants cobordism. Given boundary
conditions $V, W$ on two smaller disks $D^2_1, D^2_2 \subset D^2$,
their tensor product $V \otimes W$ is the boundary condition on
$\partial D^2$ obtained by evaluating the BV-BFV partition function
on the trinion $D^2 \setminus (D^2_1 \sqcup D^2_2)$. The
\emph{braiding} arises from exchanging the two inner disks by a
half-rotation in $D^2$, a path in $\EE_2(2)$ from a configuration to
its swap (the full loop gives the double braiding). The resulting $\EE_2$-monoidal structure encodes the
braided monoidal category of boundary conditions for CS theory on
the disk. At the perturbative level (over $\CC[\![\hbar]\!]$), the objects of
$\mathcal{B}_{\CS}^{(k)}$ can be identified with representations of
the quantum group $U_\hbar(\fg)$ (these are the Wilson line labels) with the braiding given by the universal $R$-matrix, but the
definition above is independent of this identification. See Figure \ref{fig:cylinder} and \ref{fig:E2-structure} for a visualization.
\end{construction}

\begin{figure}[ht]
\centering
\begin{tikzpicture}[scale=1.3]

% --- Left: the solid cylinder ---

% Back half of top ellipse (drawn first, behind the body)
\draw[blue!50, thick, dashed] (0,3) ellipse (1.5 and 0.5);

% Lateral surface
\fill[red!8] (-1.5,0) -- (-1.5,3) arc(180:360:1.5 and 0.5) -- (1.5,0) arc(360:180:1.5 and 0.5) -- cycle;
\draw[red!60!black, thick] (-1.5,0) -- (-1.5,3);
\draw[red!60!black, thick] (1.5,0) -- (1.5,3);

% Bottom disk (filled)
\fill[blue!12] (0,0) ellipse (1.5 and 0.5);
\draw[blue!60!black, thick, dashed] (0,0) ellipse (1.5 and 0.5);

% Top disk (filled, front half only to show depth)
\fill[blue!12, opacity=0.7] (0,3) ellipse (1.5 and 0.5);
\draw[blue!60!black, thick] (0,3) ellipse (1.5 and 0.5);

% Front half of lateral surface border
\draw[red!60!black, thick] (-1.5,0) arc(180:360:1.5 and 0.5);
\draw[red!60!black, thick] (-1.5,3) arc(180:360:1.5 and 0.5);

% Labels on the cylinder
\node[blue!60!black, font=\footnotesize\bfseries] at (0,-0.05) {$D^2 \times \{0\}$};
\node[blue!60!black, font=\footnotesize\bfseries] at (0,3) {$D^2 \times \{1\}$};
\node[red!60!black, font=\footnotesize\bfseries] at (2.2,1.5) {$S^1 \!\times\! [0,1]$};

% Time arrow
\draw[-{stealth}, thick, black!50] (-2.3,0) -- (-2.3,3) node[midway, left, font=\footnotesize] {time};

% --- Right: the schematic decomposition ---

%\node[font=\normalsize] at (5.3,1.5) {$\partial(D^2 \times [0,1])\;=$};

% Temporal part
\node[draw=blue!60!black, thick, rounded corners=4pt, fill=blue!8,
      minimum width=2.8cm, minimum height=1cm, align=center,
      font=\footnotesize] (temp) at (4.5,2.7)
      {\textbf{Temporal boundary}\\[1pt]
       $D^2 \times \{0\} \;\cup\; D^2 \times \{1\}$\\[1pt]
       \color{blue!40!black}\scriptsize codim-1 surfaces: BFV state spaces};

% Spatial part
\node[draw=red!60!black, thick, rounded corners=4pt, fill=red!8,
      minimum width=2.8cm, minimum height=1cm, align=center,
      font=\footnotesize] (spat) at (4.5,0.3)
      {\textbf{Lateral boundary}\\[1pt]
       $S^1 \times [0,1]$\\[1pt]
       \color{red!40!black}\scriptsize codim-2 corner $S^1 = \partial D^2$:\\
       \color{red!40!black}\scriptsize category of boundary conditions};

% Cup symbol
%\node[font=\normalsize] at (9.5,1.5) {$\cup$};

\end{tikzpicture}
\caption{The boundary decomposition of the cylinder $D^2 \times [0,1]$. The temporal boundary components $D^2 \times \{0,1\}$ (blue) are codimension-1 surfaces to which the BFV formalism assigns state spaces. The lateral boundary $S^1 \times [0,1]$ (red) contains the codimension-2 corner $S^1 = \partial D^2$, to which the extended theory assigns the \emph{category} $\mathcal{B}_{\CS}^{(k)}$ of boundary conditions. The objects of $\mathcal{B}_{\CS}^{(k)}$ are the choices of boundary condition on this lateral circle.}
\label{fig:cylinder}
\end{figure}
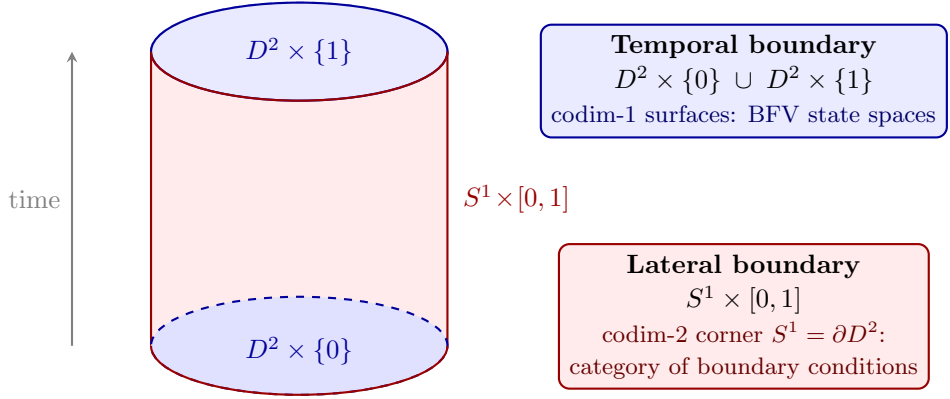

\begin{figure}[ht]
\centering
\begin{tikzpicture}[scale=1.1]

% === Panel (a): Vertical composition ===

\node[font=\bfseries] at (-5.5,3.5) {(a) Vertical composition ($\EE_1$)};

% Bottom cylinder
\fill[blue!8]
  (-7,0) -- (-7,1.22)
  arc(180:360:1.5 and 0.35)
  -- (-4,0)
  arc(360:180:1.5 and 0.35) -- cycle;

\draw[blue!50, thick] (-7,0) -- (-7,1.3);
\draw[blue!50, thick] (-4,0) -- (-4,1.3);

\fill[blue!15] (-5.5,0) ellipse (1.5 and 0.35);
\draw[blue!60!black, thick, dashed] (-5.5,0) ellipse (1.5 and 0.35);
\draw[blue!60!black, thick]
  (-7,0) arc[start angle=180,end angle=360,x radius=1.5,y radius=0.35];

% Top ellipse of lower cylinder / gluing interface
\draw[red!60!black, thick]
  (-7,1.3) arc[start angle=180,end angle=360,x radius=1.5,y radius=0.35];
\draw[red!60!black, thick, dashed]
  (-4,1.3) arc[start angle=0,end angle=180,x radius=1.5,y radius=0.35];

\node[blue!60!black, font=\footnotesize\bfseries] at (-5.5,0.65) {$M_1$};

% Top cylinder
\fill[blue!8]
  (-7,1.5) -- (-7,2.8)
  arc(180:360:1.5 and 0.35)
  -- (-4,1.5)
  arc(360:180:1.5 and 0.35) -- cycle;

\draw[blue!50, thick] (-7,1.5) -- (-7,2.8);
\draw[blue!50, thick] (-4,1.5) -- (-4,2.8);

\fill[blue!15, opacity=0.7] (-5.5,2.8) ellipse (1.5 and 0.35);

% redraw lower ellipse of upper cylinder on top of the fill
\draw[blue!60!black, thick]
  (-7,1.5) arc[start angle=180,end angle=360,x radius=1.5,y radius=0.35];
\draw[blue!60!black, thick, dashed]
  (-4,1.5) arc[start angle=0,end angle=180,x radius=1.5,y radius=0.35];

% top ellipse of upper cylinder
\draw[blue!60!black, thick]
  (-7,2.8) arc[start angle=180,end angle=360,x radius=1.5,y radius=0.35];
\draw[blue!60!black, thick]
  (-4,2.8) arc[start angle=0,end angle=180,x radius=1.5,y radius=0.35];

\node[blue!60!black, font=\footnotesize\bfseries] at (-5.5,2.15) {$M_2$};

% Gluing surface label
\draw[red!60!black, thick, ->] (-3.3,1.4) -- (-2.8,1.4);
\node[red!60!black, font=\footnotesize, right] at (-2.8,1.4) {glue along $D^2$};

% === Panel (b): Braiding ===

\node[font=\bfseries] at (2.5,3.5) {(b) Braiding ($\EE_2$)};

% Big disk
\fill[blue!5] (2.5,0.8) circle (2);
\draw[blue!60!black, thick] (2.5,0.8) circle (2);

% Small disk 1
\fill[green!15] (1.5,1.3) circle (0.55);
\draw[green!50!black, thick] (1.5,1.3) circle (0.55);
\node[green!50!black, font=\footnotesize\bfseries] at (1.5,1.3) {$V$};

% Small disk 2
\fill[orange!15] (3.5,0.7) circle (0.55);
\draw[orange!60!black, thick] (3.5,0.7) circle (0.55);
\node[orange!60!black, font=\footnotesize\bfseries] at (3.5,0.7) {$W$};

% Exchange arrows
\draw[-{stealth}, thick, red!60!black,
      decorate, decoration={snake, amplitude=1pt, segment length=4pt}]
      (2.0,1.7) to[bend left=30] (3.0,1.3);
\draw[-{stealth}, thick, red!60!black,
      decorate, decoration={snake, amplitude=1pt, segment length=4pt}]
      (3.0,0.3) to[bend left=30] (2.0,0.9);

\node[red!60!black, font=\scriptsize] at (2.5,3) {half-rotation (path in $\EE_2(2)$)};

\end{tikzpicture}
\caption{The two operations generating the $\EE_2$-structure on $\mathcal{B}_{\CS}^{(k)}$.
\textbf{(a)} Vertical composition: stacking two cylinders and gluing along the intermediate disk $D^2$ via the BFV pairing gives composition of morphisms.
\textbf{(b)} Braiding: two smaller disks carrying boundary conditions $V, W$ are embedded in a larger disk $D^2$; the trinion (complement) defines $V \otimes W$, and exchanging the disks by a half-rotation produces $c_{V,W}$.}
\label{fig:E2-structure}
\end{figure}
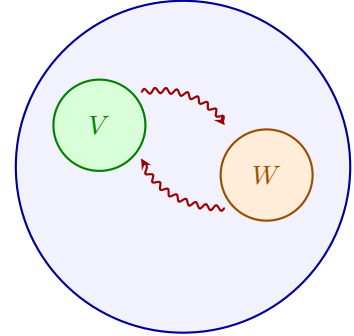

\begin{remark}[Lateral boundary conditions in derived-geometric terms]\label{rmk:bc-derived}
We explain what a boundary condition on the lateral boundary $S^1 \times [0,1]$ should be. In the
classical BV-BFV formalism, a
boundary condition on a codimension-$1$ boundary is a choice of Lagrangian (or, more generally,
coisotropic) submanifold of the boundary phase space, to which the bulk fields are required to
restrict. The derived-geometric version is the following: for a $0$-shifted symplectic stack $X$, a boundary condition is a Lagrangian
morphism $L \to X$ in the sense of Definition \ref{def:Lagrangian}, and for an $n$-shifted
symplectic stack it is an $n$-shifted Lagrangian. Applied here, a boundary condition on the
codimension-$2$ stratum $S^1$ should be a Lagrangian structure on a map
\[
L \;\longrightarrow\; \Loc_{\GC}(S^1) \;=\; [\GC/\GC],
\]
where $[\GC/\GC]$ carries its $1$-shifted symplectic structure (Theorem \ref{thm:PTVV}(i)).
Lagrangians in $[\GC/\GC]$ correspond to quasi-Hamiltonian $\GC$-spaces \cite{SafronovQH,Calaque2015};
the basic examples are the conjugacy classes $\mathcal{C} \subset \GC$, for which the inclusion
$[\mathcal{C}/\GC] \hookrightarrow [\GC/\GC]$ carries a canonical Lagrangian structure. In the
framework of shifted geometric quantization \cite{Safronov2020}, a $1$-shifted symplectic stack
is expected to quantize to a category and a $1$-shifted Lagrangian to an object of it. For
$[\GC/\GC]$ at level $k$ one expects the simple objects of $\overline{\Rep}_q(\GC)$ to correspond
to the conjugacy classes satisfying a Bohr--Sommerfeld-type integrality condition, namely those
of $\exp\big(2\pi i(\lambda+\rho)/(k+h^\vee)\big)$ with $\lambda$ in the Weyl alcove. Thus the
derived-geometric description of the boundary conditions on the circle should be: quantized
$1$-shifted Lagrangians in $[\GC/\GC]$. Making this precise, and showing that the BV-BFV lateral
boundary conditions are of this form, is part of the derived BV-BFV formalism called for in
Section \ref{sec:outlook}(a).
\end{remark}

\begin{remark}[Wilson lines: what we assume]\label{rmk:bc-Wilson}
There is also a concrete physical reading of these boundary conditions: they are \emph{Wilson
line labels}. In Chern--Simons theory a Wilson line is an observable associated to a
representation $V$ of $G$, supported on an embedded curve; when the curve is the boundary circle
$\partial D^2$, the Wilson line imposes a condition on the fields, weighting the holonomy of the
connection around $\partial D^2$ by the character of $V$. In the RT framework the simple objects
of $\overline{\Rep}_q(\GC)$ (the irreducible representations in the Weyl alcove $\mathcal{A}_k$)
correspond to the distinct Wilson line types, and the morphism spaces to local operators at
junctions between Wilson lines. This is the physical origin of the braided monoidal structure
described in Construction \ref{constr:E2-from-BVBFV}, and it is well established in the physics
literature \cite{Witten1989}.

Wilson lines give boundary conditions, but we do not know that they give all of them, and there
are reasons to expect that they do not: in the non-semisimple setting one finds module categories over $\Rep_q(\GC)$ that are not sums of
the ones coming from alcove labels \cite{DGGPR2022,BJS2021}, and the classification of module
categories over a braided tensor category is in general open. We therefore adopt the following
convention: \emph{$\mathcal{B}_{\CS}^{(k)}$ denotes the full subcategory of boundary conditions
generated by Wilson lines.} All conjectures below concern this category; we make no claim that it
exhausts the boundary conditions of Chern--Simons theory on the disk.
\end{remark}

\begin{remark}\label{rmk:extended-BVBFV}
Constructions \ref{constr:E2-from-BVBFV} and \ref{constr:BCS} implicitly use an \emph{extended} version of the BV-BFV formalism, in which data is assigned not only to codimension-1 boundaries (surfaces, carrying BFV phase spaces and vector spaces of states) but also to codimension-2 corners (circles, carrying categories of boundary conditions). In the standard CMR formalism \cite{CMR2014,CMR2018}, Chern--Simons theory on a 3-manifold $M$ with boundary $\Sigma = \partial M$ assigns a BFV phase space to the surface $\Sigma$ and a state in its quantization to the bulk $M$. The category $\mathcal{B}_{\CS}^{(k)}$ arises one level lower in the extended TQFT hierarchy. The circle $S^1$, viewed as a codimension-2 stratum (the boundary of the spatial slice $D^2$), is assigned a category whose objects are boundary conditions on $S^1$ and whose morphisms are partition functions on the annulus $S^1 \times [0,1]$. The $\EE_2$-structure comes from the 2-dimensional geometry of the spatial slice $D^2$. The extension of the BV-BFV formalism to this codimension-2 setting is initiated in \cite{CMR2014} (which treats manifolds with corners) and is closely related to the factorization algebra approach of \cite{CG2017,CG2021}; a complete derived-geometric formulation remains part of the program outlined in Section \ref{sec:outlook}. On the algebraic side, the $\EE_2$-structure on $\overline{\Rep}_q(\GC)$ is well-established (it comes from the quasi-triangular Hopf algebra structure on $U_q(\fg)$). The conjecture that the BV-BFV construction
reproduces this same $\EE_2$-structure is part of the content of Conjecture \ref{conj:E2-equiv}.
\end{remark}

We state the conjecture in two parts, since the perturbative and the non-perturbative parts
have very different status.

\begin{conjecture}[$\EE_2$-equivalence]\label{conj:E2-equiv}
\leavevmode
\begin{enumerate}[label=(\alph*)]
\item \emph{(Perturbative half.)} The $\EE_2$-monoidal category over $\CC[\![\hbar]\!]$ obtained
from BV-BFV quantization of Chern--Simons theory on $D^2\times[0,1]$ with gauge group $\GC$,
as in Construction \ref{constr:BCS}, is equivalent as an $\EE_2$-algebra in
$\CC[\![\hbar]\!]$-linear categories to $\Rep_{q}(\GC)$ with $q = e^{\hbar/2}$, where the
right-hand side is the Etingof--Kazhdan quantization of Theorem \ref{thm:EK} associated to the
standard $r$-matrix. Moreover the equivalence may be chosen so that the induced Drinfeld
associator is the Knizhnik--Zamolodchikov associator.
\item \emph{(Non-perturbative half.)} The formal family in (a) is the completion at $q=1$ of the
Drinfeld--Jimbo family over $\CC[q,q^{-1}]$; consequently it may be specialized at
$q = e^{i\pi/(k+h^\vee)}$, and after passing to tilting modules and semisimplifying one obtains
\begin{equation}
\mathcal{B}_{\CS}^{(k)} \;\simeq_{\EE_2}\; \overline{\Rep}_q(\GC),
\qquad q = e^{i\pi/(k+h^\vee)}.
\end{equation}
\end{enumerate}
\end{conjecture}

\begin{remark}[The r\^ole of the level]\label{rmk:role-of-level}
BV-BFV quantization is perturbative: it sees $\hbar$, and the level $k$ enters only through
$\hbar = 2\pi i/(k+h^\vee)$ as a formal parameter. The BV-BFV formalism by itself therefore
cannot produce a category at a finite integer level. Conjecture \ref{conj:E2-equiv}(a) is a
statement at the formal level, while the passage to a modular tensor category is part (b),
which concerns the algebraic family from which the formal family is expanded (Remark
\ref{rmk:q-parameter}) rather than BV-BFV itself.

Part (a) is closer to known results than it may appear. Suppose BV-BFV on the disk is shown to
produce \emph{some} $\EE_2$-deformation of $\Rep(\GC)$ over $\CC[\![\hbar]\!]$ with the correct
first-order term, i.e.\ with associated Lie bialgebra the standard one. Then by Drinfeld's
rigidity results on deformations of quasi-triangular Hopf algebras
\cite{Drinfeld1987,Drinfeld1990}, together with the classification of quantizations of a given
Lie bialgebra up to the action of associators \cite{EK1996}, that deformation is \emph{already}
equivalent to the Etingof--Kazhdan quantization, up to the choice of an associator. Part (a)
would thus follow once an $\EE_2$-structure with the right classical limit has been exhibited by
BV-BFV. The remaining questions are (i) to construct the $\EE_2$-structure, which requires a
BV-BFV formalism for manifolds with corners, and (ii) to determine the associator (Step 2b).
\end{remark}

\begin{remark}\label{rmk:E2-evidence}
The classical $r$-matrix underlying the Lie bialgebra structure on $\fg$ is the datum encoded in
the $1$-shifted Poisson structure on $B\GC$ (Theorem \ref{thm:Safronov-PL}), so
Conjecture \ref{conj:E2-equiv}(a) says that BV-BFV on the disk and Etingof--Kazhdan quantize the
same classical datum to the same answer. Related perspectives connecting perturbative gauge
theory to quantum group structures via factorization algebras appear in Costello's work on
$4$-dimensional Chern--Simons theory and the Yangian \cite{CostelloYangian,CWY2018}, and in the
Chern--Simons/Knizhnik--Zamolodchikov circle of ideas going back to \cite{Kohno1987,Drinfeld1990}.
\end{remark}

\begin{remark}\label{rmk:conventions-extended}
We work with the $(3\text{-}2\text{-}1)$-extended TQFT convention
of Definition \ref{def:extended-TQFT}, in which the circle $S^1$
is the lowest-dimensional object and is assigned the $\EE_2$-category
$\mathcal{B}_{\CS}^{(k)}$. Factorization homology of this
$\EE_2$-algebra over a surface $\Sigma$ produces the category
$\iint_\Sigma \mathcal{B}_{\CS}^{(k)}$ assigned to $\Sigma$.
This should be distinguished from the fully extended setting of
Conjecture \ref{conj:fully-extended}, in which the braided category is
the value at a point of a $4$-dimensional theory and the RT theory is a
boundary theory for it. The cylinder $D^2 \times [0,1]$ with its
boundary conditions on $S^1 = \partial D^2$ produces the
$\EE_2$-algebra itself, which is the value at the circle in the
$(3\text{-}2\text{-}1)$-convention.
\end{remark}

\subsection{The main identification}

Assuming Conjecture \ref{conj:E2-equiv}, the factorization homology machinery of Section \ref{sec:fact-homology} would propagate the local identification (the $\EE_2$-category assigned to a circle) to the surface-level data (the categories and distinguished objects assigned to closed surfaces). Extending this further to a full $(3\text{-}2\text{-}1)$-extended TQFT, including cobordism maps, requires the additional conjectural input of quantized Lagrangian correspondences (Conjecture \ref{conj:Lagrangian-quant} below). With this caveat, we state the central conjecture of the paper.

\begin{conjecture}[Main Conjecture]\label{conj:main}
There exists a natural equivalence of $(3\text{-}2\text{-}1)$-extended topological quantum field theories
\begin{equation}
Z_{\mathrm{BV,np}}^{(k)} \;\simeq\; Z_{\mathrm{RT}}^{(k)} \;:\; \Bord_3^{\mathrm{ext}} \longrightarrow 2\text{-}\Vect,
\end{equation}
where $Z_{\mathrm{BV,np}}^{(k)}$ denotes the non-perturbative completion of the BV-BFV partition function at level $k$ and $2\text{-}\Vect$ is the $2$-category of $\CC$-linear categories, functors, and natural transformations. At each categorical level, the equivalence identifies:
\begin{enumerate}[label=(\roman*)]
\item \emph{Circles:} the $\EE_2$-category $\mathcal{B}_{\CS}^{(k)} \simeq \overline{\Rep}_q(\GC)$ (Conjecture \ref{conj:E2-equiv});
\item \emph{Closed surfaces:} the factorization homology categories $\iint_{\Sigma_g} \overline{\Rep}_q(\GC) \simeq (\mathbf{H},\, u_{\Sigma_g})$, with $u_{\Sigma_g} = V_{\mathrm{RT}}(\Sigma_g)$ (Theorem \ref{thm:AKZ});
\item \emph{3-cobordisms:} the linear maps between RT state spaces $\Zrt(M): V_{\mathrm{RT}}(\Sigma_1) \to V_{\mathrm{RT}}(\Sigma_2)$.
\end{enumerate} More precisely, for every closed oriented surface $\Sigma_g$, the BV-BFV state space is identified with the RT state space via the chain
\begin{equation*}
\begin{tikzcd}[column sep=2.5em, row sep=2.5em]
\mathcal{H}_{\Sigma_g}^{\BV} \arrow[d, dashed, leftrightarrow, "\sim"'] & \Quant\big(\Loc_{\GC}(\Sigma_g)\big) \arrow[l, "\textnormal{Step 2}"'] \arrow[r, "\sim", "\textnormal{BBJ}"'] & \displaystyle\iint_{\Sigma_g} \Rep_q(\GC) \arrow[d, "\textnormal{semisimplify}"] \\
V_{\mathrm{RT}}(\Sigma_g) & (\Vect,\, u_{\Sigma_g}) \arrow[l, "u_{\Sigma_g} \,\mapsto"'] \arrow[r, leftarrow, "\sim"', "\textnormal{AKZ}"] & \displaystyle\iint_{\Sigma_g} \overline{\Rep}_q(\GC)
\end{tikzcd}
\end{equation*}
where the distinguished object $u_{\Sigma_g}$ is the RT state space $V_{\mathrm{RT}}(\Sigma_g)$. For every oriented 3-dimensional cobordism $M: \Sigma_1 \to \Sigma_2$, the cobordism induces a functor $\iiint_M$ between the factorization homology categories $\iint_{\Sigma_i} \mathcal{B}_{\CS}^{(k)}$. After the identification with $\overline{\Rep}_q(\GC)$ via Conjecture \ref{conj:E2-equiv}, and the collapse $\iint_{\Sigma_i} \overline{\Rep}_q(\GC) \simeq (\Vect,\, u_{\Sigma_i})$ via AKZ, this functor induces a linear map between the distinguished objects. The conjecture asserts that this linear map agrees with the RT partition function:
\begin{equation}\label{eq:main-diagram}
\begin{tikzcd}[column sep=4em]
\iint_{\Sigma_1} \overline{\Rep}_q(\GC) \arrow[r, "\iiint_M"] \arrow[d, "\textnormal{AKZ}"', "\sim"] & \iint_{\Sigma_2} \overline{\Rep}_q(\GC) \arrow[d, "\sim"', "\textnormal{AKZ}"] \\
(\Vect,\, u_{\Sigma_1} = V_{\mathrm{RT}}(\Sigma_1)) \arrow[r, "\Zrt(M)"'] & (\Vect,\, u_{\Sigma_2} = V_{\mathrm{RT}}(\Sigma_2))
\end{tikzcd}
\end{equation}
Here the vertical equivalences are given by the Ai--Kong--Zheng theorem (Theorem \ref{thm:AKZ}), applied after identifying $\mathcal{B}_{\CS}^{(k)} \simeq \overline{\Rep}_q(\GC)$ (Conjecture \ref{conj:E2-equiv}). The top horizontal arrow is the functor on factorization homology categories induced by the cobordism $M$; the bottom horizontal arrow is the linear map between distinguished objects, which is the RT partition function $\Zrt(M): V_{\mathrm{RT}}(\Sigma_1) \to V_{\mathrm{RT}}(\Sigma_2)$.
\end{conjecture}

\begin{remark}[Non-perturbative partition function]\label{rmk:np-meaning}
The notation $Z_{\mathrm{BV,np}}^{(k)}$ requires clarification. The BV-BFV partition function as constructed by Cattaneo--Mnev--Reshetikhin is inherently perturbative. It takes values in $\CC[\![\hbar]\!]$ and is defined by expansion around individual flat connections. By ``non-perturbative completion'' we mean the functor obtained by
constructing the $\EE_2$-category $\mathcal{B}_{\CS}^{(k)}$ from the BV-BFV data on the disk (Construction \ref{constr:BCS}), which is defined without reference to any particular flat connection, applying factorization homology to extend this local datum to all surfaces and cobordisms, and specializing to level $k$.
The resulting functor is non-perturbative in the sense that it does not depend on a choice of classical background and takes values in $\Vect$ rather than $\Vect_\hbar$. The content of the Main Conjecture is that this procedure recovers the RT functor.
\end{remark}

Thus the Main Conjecture asserts that the BV-BFV and RT constructions define the \emph{same $(3\text{-}2\text{-}1)$-extended TQFT}: the same $\EE_2$-category assigned to the circle, the same categories (with distinguished objects) assigned to closed surfaces, and the same linear maps assigned to 3-cobordisms. The ordinary (non-extended) CMR functor $\Zbv: \Bord_3^{\mathrm{or}} \to \Vect_\hbar$ of Corollary \ref{cor:BV-functor} captures only the top two levels (surfaces and 3-manifolds); the conjectural non-perturbative completion promotes this to a $(3\text{-}2\text{-}1)$-extended functor by adjoining the $\EE_2$-category $\mathcal{B}_{\CS}^{(k)}$ at the circle level. The surface-level identifications then follow from factorization homology, while the cobordism-level identifications require the additional input of quantized Lagrangian correspondences (Conjecture \ref{conj:Lagrangian-quant}).

\subsection{Detailed proof strategy}

We outline a strategy for establishing the Main Conjecture in four steps, from the classical level to the functorial quantum identification, and indicate the technical problems at each stage.

\subsubsection{Step 1: Classical phase space identification}

One might expect the BFV phase space of Chern--Simons theory to be identified with the derived
character stack. This cannot hold literally: the BFV phase space is a differential-geometric
object for the compact group $G$, whereas $\Loc_{\GC}(\Sigma)$ is an algebraic stack for the
complexified group. What can be said is the following comparison.

\begin{proposition}[Comparison at the classical level]\label{prop:step1}
Let $\Sigma$ be a closed oriented surface, $G$ a compact simple Lie group with complexification
$\GC$, and let $\mathcal{F}_\Sigma^{\partial}$ be the BFV phase space \eqref{eq:CS-BFV-fields} of
Chern--Simons theory on $\Sigma$. Then:
\begin{enumerate}[label=(\roman*)]
\item The BFV fields $\Omega^\bullet(\Sigma,\fg)[1]$ with the differential $Q^\partial_\Sigma$,
linearized at a flat connection $A$, compute $C^\bullet(\Sigma;\Ad_A)[1]$, which is by
Proposition \ref{prop:Loc-explicit}(ii) the tangent complex of $\Loc_{\GC}(\Sigma)$ at the
corresponding point, once $A$ is regarded as a $\GC$-local system via Remark \ref{rmk:RH}.
\item Under this identification of tangent complexes, the Atiyah--Bott form
\eqref{eq:AB-symplectic} agrees with the pairing underlying the PTVV $0$-shifted symplectic
structure of Theorem \ref{thm:PTVV}(ii), both being the transgression of the invariant form on
$\fg$ against the fundamental class of $\Sigma$.
\item The BFV charge \eqref{eq:BFV-charge} is the moment map for the gauge group action together
with the generator of gauge transformations; its zero locus modulo the gauge groupoid is the
groupoid of flat connections, which under Riemann--Hilbert corresponds to
$\Loc_{\GC}(\Sigma)_{\mathrm{cl}}$ after complexification.
\end{enumerate}
\end{proposition}

\begin{proof}
(i) is the standard computation identifying the linearization of the BFV complex with twisted de
Rham cohomology, combined with Proposition \ref{prop:Loc-explicit}(ii). (ii) is a comparison of
two transgressions of the same form and follows from Remark \ref{rmk:PTVV-mechanism} together
with the classical computation of the Atiyah--Bott form \cite{AtiyahBott1983,Goldman1984}
and with the coisotropic/Poisson-reduction picture of \cite{Safronov2017}; on the smooth
locus of the classical truncation both reduce to $\int_\Sigma \Tr_\fg(\alpha\wedge\beta)$. (iii)
is the classical statement that the BFV charge generates the constraint algebra
\cite{CMR2014}, together with Remark \ref{rmk:RH}.
\end{proof}

\begin{remark}\label{rmk:step1-caveats}
Proposition \ref{prop:step1} does not say that $\mathcal{F}_\Sigma^{\partial}$ and
$\Loc_{\GC}(\Sigma)$ are isomorphic; they are objects of different categories: the first is a graded Fr\'echet manifold
with a compact structure group, the second an algebraic stack over $\CC$. What the proposition
gives is an identification of their \emph{tangent complexes} at corresponding points, together
with a matching of the symplectic pairings, after complexification and after Riemann--Hilbert.
The following would be needed for an identification, and each is open:
\begin{enumerate}[label=(\alph*)]
\item a formulation of the BFV phase space of the \emph{compact} theory as a derived
$C^\infty$-stack, in the sense of Remark \ref{rmk:DDG}, so that the two sides live in comparable
categories;
\item a proof that its complexification is $\Loc_{\GC}(\Sigma)$ as a derived stack, and not
merely pointwise on tangent complexes; and
\item a proof that the shifted symplectic structures correspond globally, not only on the smooth
locus of the classical truncation.
\end{enumerate}
Problems (a)--(c) constitute Step 1; what is known, and used below, is the tangent-complex
comparison of Proposition \ref{prop:step1}.
\end{remark}

Modulo these caveats, this step says that the BV-BFV and derived algebraic geometry frameworks
are quantizing classical objects with the same infinitesimal structure. The divergence between
them, if any, must therefore arise either at the quantum level or in the compact/complex passage.

\subsubsection{Step 2: Quantization comparison}

This is the most demanding step of the program: one must show that two a priori different quantization procedures, the BV-BFV perturbative quantization and the shifted deformation quantization (in the framework of CPTVV \cite{CPTVV2017}, Safronov \cite{Safronov2017b}, and Etingof--Kazhdan \cite{EK1996}), produce the same output when applied to the same classical input.

\begin{conjecture}[Quantization agreement]\label{conj:quant-agree}
The BV-BFV quantization of Chern--Simons theory on $D^2 \times [0,1]$ (Construction
\ref{constr:BCS}), regarded as a $\CC[\![\hbar]\!]$-linear monoidal category, agrees with the
shifted deformation quantization of $B\GC = \Loc_{\GC}(D^2)$ (Theorems
\ref{thm:CPTVV}--\ref{thm:EK} and Remark \ref{rmk:what-quantizes}(i)), that is, with
$\Rep_q(\GC)$, $q = e^{\hbar/2}$, regarded as a monoidal category.
\end{conjecture}

\begin{remark}\label{rmk:agreement}
Both procedures start from the same classical data (the Lie algebra $\fg$ with its invariant
form and standard Lie bialgebra structure, equivalently the shifted Poisson structures on
$B\GC$ of Theorem \ref{thm:Safronov-PL}) and both produce deformations of $\Rep(\GC)$ over
$\CC[\![\hbar]\!]$. The conjecture asserts that these agree as monoidal categories. The upgrade
to an $\EE_2$-equivalence (Conjecture \ref{conj:E2-equiv}(a)) requires in addition matching the
braiding, which on the BV-BFV side arises from the disk-embedding geometry of Construction
\ref{constr:E2-from-BVBFV}(iii), and on the algebraic side from the $R$-matrix of $U_q(\fg)$.
\end{remark}

We break the argument into three sub-steps, each addressing a different aspect of the comparison.

\medskip
\noindent\textbf{Step 2a: Reduction to the Lie algebra.} The first task is to show that the BV-BFV quantization on the cylinder $D^2 \times [0,1]$ can be expressed in terms of the $L_\infty$-algebra $\fg$ (the tangent complex of $B\GC$ at its point, shifted). The key simplification is that $D^2$ is contractible, i.e., $\pi_1(D^2) = 0$, so there is a unique flat connection (the trivial one), and the construction of $\mathcal{B}_{\CS}^{(k)}$ reduces to computing the BV pushforward of the Chern--Simons action from the fields on the interior of the spatial disk $D^2$ to the fields on its boundary circle $S^1 = \partial D^2$. Since there is only one critical point, one expects no non-perturbative corrections, so that the gauge-fixed perturbative expansion captures the full formal deformation, and the effective boundary action is determined by the Lie bracket of $\fg$. This is analogous to Kontsevich's deformation quantization of Poisson manifolds, where a Feynman-diagram expansion produces the full formal deformation.

\medskip
\noindent\textbf{Step 2b: Matching with Etingof--Kazhdan.} The monoidal deformation of $\Rep(\GC)$ is governed by the Etingof--Kazhdan quantization of the Lie bialgebra $(\fg, \delta_r)$ determined by the classical $r$-matrix $r \in \fg \otimes \fg$, which encodes the $1$-shifted Poisson structure on $B\GC$. On the BV-BFV side, the same $r$-matrix is expected to arise from the propagator on the disk, and the restriction to the boundary circle of the Green's function on $D^2$ is expected to produce a Drinfeld associator. One would then expect the Feynman diagram expansion on the disk to be related to the Kontsevich integral and hence, via the Etingof--Kazhdan construction \cite{EK1996}, to the quantum group $U_q(\fg)$.\\

The central point of this step would be that the BV gauge-fixing on the disk, using the standard round metric, produces a propagator whose boundary behavior encodes the Knizhnik--Zamolodchikov (KZ) connection, which is the flat connection on configuration spaces of points on the boundary circle that governs the braiding and associativity in the quantum group. The holonomy of the KZ connection gives the $R$-matrix of $U_q(\fg)$, and the monodromy around higher-order configurations gives the higher associators. These considerations suggest that the Feynman diagram expansion of Chern--Simons theory on the disk produces the quantum group structure, but assembling these ingredients into a rigorous proof of Conjecture \ref{conj:quant-agree} remains a substantial open problem. In particular, the passage from the analytic data (propagators, KZ connection) to the algebraic data ($\EE_2$-category structure) requires a derived-geometric reformulation of the BV-BFV formalism that is not yet available.

\medskip
\noindent\textbf{Step 2c: Root of unity specialization.} The output of Steps 2a and 2b is a formal deformation, a braided monoidal category over $\CC[\![\hbar]\!]$ with $q = e^{\hbar/2}$. To make contact with the RT theory, we must specialize to the non-formal value $q = e^{i\pi/(k+h^\vee)}$, which is a root of unity. This passage requires three sub-steps:

\begin{enumerate}[label=(\roman*)]
\item \emph{Algebraization.} One must show that the formal deformation of $\Rep(G)$ over $\CC[\![\hbar]\!]$ extends to an algebraic deformation over $\CC[\hbar]$ (or a suitable Rees algebra) that can be evaluated at specific values of $\hbar$. For quantum groups, this extension is known to exist. It is provided by the explicit Drinfeld--Jimbo presentation of $U_q(\fg)$ as an algebra over $\CC[q, q^{-1}]$. However, it must be verified that this algebraic family matches the formal family produced by the BV-BFV expansion.

\item \emph{Specialization.} Setting $\hbar = 2\pi i/(k+h^\vee)$ yields a braided monoidal category $\Rep_q(\GC)$ at a root of unity. This category is no longer semisimple: the quantum dimension of certain objects vanishes, and the category contains non-trivial extensions and negligible morphisms.

\item \emph{Semisimplification.} To obtain a modular tensor category, one passes to the quotient
\begin{equation}
\overline{\Rep}_q(\GC) \;=\; \Rep_q(\GC)\, /\, \mathcal{N}
\end{equation}
by the tensor ideal of negligible morphisms. This categorical localization is what produces the
finite, semisimple, modular tensor category used in the RT construction.

For the program to be consistent, semisimplification should be compatible with
factorization homology. One might try to express this as ``$\iint_\Sigma(\mathcal{A}/\mathcal{N})
\simeq (\iint_\Sigma\mathcal{A})/\mathcal{N}_\Sigma$'', but the two occurrences of $\mathcal{N}$ are
of different kinds. On the left, $\mathcal{N} \subset \mathcal{A}$ is a tensor ideal of
morphisms, and $\mathcal{A}/\mathcal{N}$ is a quotient by morphisms, not by objects. On the
right, $\iint_\Sigma \mathcal{A}$ for a closed surface carries no monoidal structure, so it has no
tensor ideals, and $\mathcal{N}_\Sigma$ could at best be a class of morphisms, or a Serre
subcategory, defined by other means; we do not know that such a Serre subcategory exists.

The appropriate formulation is as follows. Semisimplification is a
symmetric monoidal functor $\pi: \mathcal{A} \to \mathcal{A}/\mathcal{N}$ of braided tensor
categories \cite{EO2021}, hence a map of $\EE_2$-algebras. Factorization homology is a colimit
and therefore functorial for maps of $\EE_2$-algebras, so $\pi$ induces
\begin{equation}\label{eq:semisimplify-FH}
\pi_\Sigma : \;\iint_\Sigma \mathcal{A} \;\longrightarrow\; \iint_\Sigma \big(\mathcal{A}/\mathcal{N}\big).
\end{equation}
This functor exists for formal reasons, since factorization homology is a colimit. The question
is what \eqref{eq:semisimplify-FH} is: whether it is a localization, what its kernel is, and
whether the kernel is generated by the image of $\mathcal{N}$. Since $\pi$ is not exact (it kills
objects of quantum dimension zero), $\pi_\Sigma$ need not be a quotient by a Serre subcategory.
This is the open problem (c) of Section \ref{sec:outlook}. The interaction between non-semisimple categories and TQFT constructions, which is
directly relevant here, is developed in \cite{DGGPR2022,CGPM2023,Haioun2025,BJS2021}.
\end{enumerate}

\subsubsection{Step 3: Factorization homology transport (surfaces)}

Once Conjecture \ref{conj:E2-equiv} is established, factorization homology would propagate the local identification to surface-level data: for every closed surface $\Sigma_g$, the factorization homology $\iint_{\Sigma_g} \mathcal{B}_{\CS}^{(k)} \simeq \iint_{\Sigma_g} \overline{\Rep}_q(\GC)$ would be identified, and the AKZ theorem would then give $u_{\Sigma_g} = V_{\mathrm{RT}}(\Sigma_g)$. This step concerns only the \emph{surface-level} invariants (the assignment of categories and distinguished objects to closed surfaces). It does \emph{not} by itself establish the identification of cobordism maps; that requires the additional input of Step 4 below.

The state space that BV-BFV assigns to a surface $\Sigma$ is identified with the RT state space through the following chain of equivalences:
\begin{equation}\label{eq:step3}
\mathcal{H}_\Sigma^{\BV} \xleftarrow{\;\text{Step 2}\;} \Quant\big(\Loc_{\GC}(\Sigma)\big) \xrightarrow[\text{BBJ}]{\;\sim\;} \iint_\Sigma \Rep_q(\GC).
\end{equation}
At generic $q$, the right-hand side is a non-trivial quantum character variety. At a root of unity, after semisimplification, Theorem \ref{thm:AKZ} gives $\iint_{\Sigma_g} \overline{\Rep}_q(\GC) \simeq (\mathbf{H},\, u_{\Sigma_g})$ with $u_{\Sigma_g} \cong V_{\mathrm{RT}}(\Sigma_g)$.

\begin{remark}\label{rmk:step2-disk}
The arrow labeled ``Step 2'' in the diagram of Conjecture \ref{conj:main} should be understood as follows. The comparison between the BV-BFV quantization and the deformation quantization $\Quant(\Loc_{\GC}(\Sigma))$ is \emph{not} performed surface by surface. Rather, the key identification occurs on the disk $D^2$. One shows that the $\EE_2$-category $\mathcal{B}_{\CS}^{(k)}$ produced by BV-BFV (Construction \ref{constr:BCS}) is equivalent to $\Rep_q(\GC)$, the deformation quantization of $B\GC = \Loc_{\GC}(D^2)$. Factorization homology then transports this local equivalence to all surfaces, giving
\[
\mathcal{H}_\Sigma^{\BV} \;\simeq\; \iint_\Sigma \mathcal{B}_{\CS}^{(k)} \;\simeq\; \iint_\Sigma \Rep_q(\GC) \;\simeq\; \Quant\big(\Loc_{\GC}(\Sigma)\big),
\]
where the last equivalence is the BBJ theorem (Theorem \ref{thm:BBJ-quant}). In this way the infinite-dimensional analytical problem of comparing quantizations on a general surface $\Sigma$ is reduced to an algebraic problem on the disk.
\end{remark}

The identification $\mathcal{H}_\Sigma^{\BV} \cong V_{\mathrm{RT}}(\Sigma)$ would not require a separate argument for each genus: it would follow from a single equivalence of $\EE_2$-categories together with factorization homology, for every surface, including surfaces with punctures and boundary components.

\subsubsection{Step 4: Cobordism maps}

Steps 1--3 address the surface-level data: the identification of the $\EE_2$-categories assigned to circles and, via factorization homology and AKZ, the state spaces assigned to closed surfaces. To complete the proof of the Main Conjecture at the level of a full $(3\text{-}2\text{-}1)$-extended TQFT, we must also match the \emph{linear maps} assigned to 3-dimensional cobordisms. This is a logically independent step that requires additional input beyond factorization homology. For every cobordism $M: \Sigma_1 \to \Sigma_2$, the BV-BFV partition function $\Zbv(M): \mathcal{H}_{\Sigma_1} \to \mathcal{H}_{\Sigma_2}$ must agree with the RT map $\Zrt(M): V_{\mathrm{RT}}(\Sigma_1) \to V_{\mathrm{RT}}(\Sigma_2)$ under the identification of Step 3.

The key input is the derived-geometric interpretation of the BV-BFV partition function as a \emph{quantized Lagrangian correspondence}.

The quantization of a Lagrangian correspondence to a bimodule is developed in substantial
generality by Melani--Safronov \cite{MS2018} and Safronov \cite{Safronov2020}, and the BV-BFV
partition function of a cobordism is a map between boundary state spaces (Theorem
\ref{thm:BV-gluing}). The conjectural content is that the resulting map is the RT map.

\begin{conjecture}[Quantized Lagrangian correspondence]\label{conj:Lagrangian-quant}
Let $M: \Sigma_1 \to \Sigma_2$ be a compact oriented $3$-dimensional cobordism between closed
oriented surfaces. Then the Reshetikhin--Turaev map
\[
\Zrt(M): \; V_{\mathrm{RT}}(\Sigma_1) \longrightarrow V_{\mathrm{RT}}(\Sigma_2)
\]
equals the map obtained by quantizing, in the sense of \cite{MS2018,Safronov2020}, the
Lagrangian correspondence
\begin{equation}
\Loc_{\GC}(M) \longrightarrow \Loc_{\GC}(\Sigma_1) \times \Loc_{\GC}(\Sigma_2)^-
\end{equation}
of Theorem \ref{thm:AKSZ-PTVV} (where the superscript $(-)$ denotes the opposite symplectic
structure), and then passing to distinguished objects via the identifications of Step 3; this
quantization produces a bimodule over the quantized algebras attached to $\Sigma_1$ and
$\Sigma_2$, which we conjecture is moreover the BV-BFV partition function $\Zbv(M)$.
\end{conjecture}

\begin{figure}[ht]
\centering
\begin{tikzpicture}[scale=1]

% === Top: the Lagrangian correspondence (roof diagram) ===

% Loc_G(M) at the top
\node[draw=violet!60!black, thick, rounded corners=5pt, fill=violet!8,
      inner sep=7pt, font=\small, align=center] (locM) at (0,3)
      {$\Loc_{\GC}(M)$\\[-2pt]\scriptsize $(-1)$-shifted symplectic};

% Loc_G(Sigma_1) bottom left
\node[draw=blue!60!black, thick, rounded corners=5pt, fill=blue!8,
      inner sep=7pt, font=\small, align=center] (loc1) at (-4,0)
      {$\Loc_{\GC}(\Sigma_1)$\\[-2pt]\scriptsize $0$-shifted symplectic};

% Loc_G(Sigma_2)^- bottom right
\node[draw=red!60!black, thick, rounded corners=5pt, fill=red!8,
      inner sep=7pt, font=\small, align=center] (loc2) at (4,0)
      {$\Loc_{\GC}(\Sigma_2)^-$\\[-2pt]\scriptsize $0$-shifted symplectic\\[-2pt]\scriptsize (opposite orientation)};

% Restriction arrows
\draw[-{stealth}, thick, violet!60!black] (locM) --
      node[above left, font=\scriptsize, align=right] {restrict to\\$\Sigma_1 = \partial_{\mathrm{in}} M$} (loc1);
\draw[-{stealth}, thick, violet!60!black] (locM) --
      node[above right, font=\scriptsize, align=left] {restrict to\\$\Sigma_2 = \partial_{\mathrm{out}} M$} (loc2);

% Lagrangian label
\node[draw=orange!70!black, thick, rounded corners=3pt, fill=orange!8,
      inner sep=4pt, font=\scriptsize\bfseries] at (0,1)
      {Lagrangian (Calaque)};

% === Bottom: the quantized version ===

\draw[thick, black!20] (-5.5,-1.2) -- (5.5,-1.2);
\node[black!40, font=\scriptsize\bfseries] at (0,-1.5) {quantize (conjectural)};

% Quantized version
\node[draw=blue!60!black, thick, dashed, rounded corners=5pt, fill=blue!5,
      inner sep=6pt, font=\small, align=center] (q1) at (-4,-3.5)
      {$\iint_{\Sigma_1}\!\overline{\Rep}_q(\GC)$\\[-2pt]$\simeq (\mathbf{H},\, u_{\Sigma_1})$};

\node[draw=red!60!black, thick, dashed, rounded corners=5pt, fill=red!5,
      inner sep=6pt, font=\small, align=center] (q2) at (4,-3.5)
      {$\iint_{\Sigma_2}\!\overline{\Rep}_q(\GC)$\\[-2pt]$\simeq (\mathbf{H},\, u_{\Sigma_2})$};

% Linear map
\draw[-{stealth}, thick, dashed, orange!70!black, line width=1.2pt] (q1) --
      node[above, font=\scriptsize\bfseries, orange!70!black] {$\Zrt(M)$}
      node[below, font=\scriptsize, black!50] {quantized Lagr. correspondence} (q2);

% Vertical dashed arrows (quantization)
\draw[-{stealth}, thick, dashed, black!40] (loc1) -- (q1);
\draw[-{stealth}, thick, dashed, black!40] (loc2) -- (q2);

% Cobordism picture on the right
\begin{scope}[shift={(8,0.5)}]
  \fill[violet!8, rounded corners=5pt] (-1,-2) rectangle (1,2);
  \draw[violet!40, thick, rounded corners=5pt] (-1,-2) rectangle (1,2);
  \node[violet!60!black, font=\bfseries] at (0,0) {$M$};

  \draw[blue!60!black, very thick] (-1,-2) -- (-1,2);
  \node[blue!60!black, font=\scriptsize, left] at (-1,0) {$\Sigma_1$};

  \draw[red!60!black, very thick] (1,-2) -- (1,2);
  \node[red!60!black, font=\scriptsize, right] at (1,0) {$\Sigma_2$};
\end{scope}

\end{tikzpicture}
\caption{The Lagrangian correspondence associated to a cobordism (Conjecture \ref{conj:Lagrangian-quant}). \emph{Top:} At the classical level, a 3-cobordism $M: \Sigma_1 \to \Sigma_2$ gives a ``roof'' diagram. The restriction maps from $\Loc_{\GC}(M)$ to $\Loc_{\GC}(\Sigma_1) \times \Loc_{\GC}(\Sigma_2)^-$ carry a Lagrangian structure (Theorem \ref{thm:AKSZ-PTVV}, due to Calaque). \emph{Bottom:} The conjecture asserts that quantizing this Lagrangian correspondence produces the linear map $\Zrt(M): V_{\mathrm{RT}}(\Sigma_1) \to V_{\mathrm{RT}}(\Sigma_2)$ between the distinguished objects of the factorization homology categories. This is the content of Step 4 of the proof strategy, and is logically independent of the surface-level identifications of Step 3. \emph{Right:} The cobordism $M$ with incoming boundary $\Sigma_1$ and outgoing boundary $\Sigma_2$.}
\label{fig:Lagrangian-correspondence}
\end{figure}
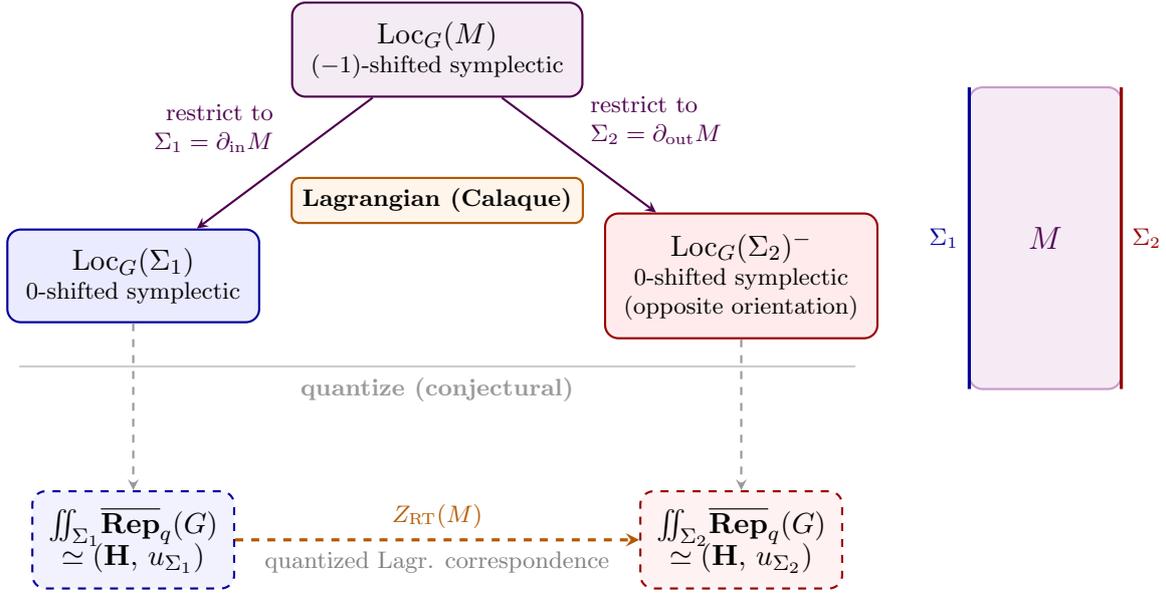

\begin{remark}[Functoriality]\label{rmk:Lagrangian-maps}
In ordinary symplectic geometry a Lagrangian submanifold $L \subset X_1^- \times X_2$ defines a
canonical relation between the two phase spaces, and upon quantization such a correspondence is
expected to become a linear map between the quantum Hilbert spaces. The shifted analogue works
one categorical level up, with bimodules in place of linear maps \cite{MS2018,Safronov2020}.

Quantizing a single Lagrangian correspondence is well understood. The difficult part, which a
TQFT requires, is to make the assignment \emph{functorial}: given composable correspondences $L_{12}$ and $L_{23}$, one wants
the quantization of the composite to be the composite of the quantizations,
\[
Q(L_{23} \circ L_{12}) \;\simeq\; Q(L_{23}) \circ Q(L_{12}),
\]
and this fails in general, for reasons that are visible already in finite dimensions. The
composite of two Lagrangian correspondences need not be a Lagrangian (transversality can fail,
and the fibre product can be badly behaved); even when it is, the composite of the quantizations
involves a fibre integral whose stationary phase produces extra factors (Maslov indices,
metaplectic corrections, and in the categorified setting corrections coming from the failure of
base change). In symplectic topology this is why Wehrheim--Woodward introduced
\emph{generalized} Lagrangian correspondences and proved functoriality theorems for quilted Floer
theory \cite{WehrheimWoodward2010}. The derived setting improves matters (fibre products are
automatically derived, so transversality is not an issue, and \cite{Calaque2015,MS2018} give
composition results for shifted Lagrangian correspondences), but a statement strong enough to
yield functoriality of the quantization, with the coherence data required for the
$(3\text{-}2\text{-}1)$-extended structure, is not available.

Conjecture \ref{conj:Lagrangian-quant} therefore concerns a single cobordism. Upgrading it to a
functor on $\Bord_3^{\mathrm{ext}}$ requires a compositionality statement for quantized shifted
Lagrangian correspondences, which is an independent problem.
\end{remark}

\subsection{Compatibility with cellular decomposition}

As a consistency check, we record a conjecture relating the cellular approach of Construction \ref{constr:cellular} to the factorization homology approach.

\begin{conjecture}[Cellular-factorization compatibility]\label{conj:cellular}
The cellular BV-BFV partition function on a triangulated 3-manifold $(M, \Delta)$, computed by iteratively gluing finite-dimensional BV-BFV data along simplices (Construction \ref{constr:cellular}), is equivalent to the value $Z_{\mathcal{B}_{\CS}^{(k)}}(M)$ assigned to $M$ by the $(3\text{-}2\text{-}1)$-extended TQFT determined by the $\EE_2$-category $\mathcal{B}_{\CS}^{(k)}$. More precisely, the iterated BV pushforward along simplices of $\Delta$ should reproduce the same linear map between state spaces that the extended TQFT assigns to the cobordism $M$, with the factorization homology of surfaces providing the intermediate state-space identifications.
\end{conjecture}

This conjecture is natural because both constructions compute the partition function by assembling local data according to a decomposition of $M$, and both satisfy the same gluing axioms. The cellular approach uses a triangulation and finite-dimensional BV-BFV algebras on simplices; the factorization homology approach uses the Ran space of disk embeddings and the $\EE_2$-algebra $\mathcal{B}_{\CS}^{(k)}$. A proof would require showing that the finite-dimensional BV-BFV data assigned to a simplex by the CMR construction \cite{CMR2020} can be identified with the restriction of the factorization algebra to a small disk, and that the iterated BV pushforward (gluing simplices) agrees with the colimit (gluing disks). If established, this conjecture would give computational access to the factorization homology side, reducing it to the finite-dimensional integrals of the cellular theory.

\section{Koszul Duality and the Perturbative/Non-perturbative Bridge}\label{sec:koszul}

Section \ref{sec:bridge} compared the BV-BFV and RT constructions through factorization homology.
This section concerns a different relation: between the perturbative expansion of Chern--Simons
theory around a single flat connection and the geometry of $\Loc_{\GC}(M)$. Near a point $[A]$ the
relation is a form of Koszul duality (Sections \ref{subsec:koszul-general} and
\ref{subsec:koszul-CS}). Koszul duality is local, and the global structure has to come from the
geometry of the stack. We discuss two sources of it: the perverse sheaf of vanishing cycles
(Remark \ref{rmk:schefers}) and a stratification by stabilizer type (Section
\ref{subsec:theta-strat}).

\subsection{$\EE_n$-Koszul duality}\label{subsec:koszul-general}

Koszul duality exchanges two descriptions of the same algebraic object, for instance the
symmetric algebra $\mathrm{Sym}(V)$ and the exterior algebra $\Lambda(V^\vee)$, or the enveloping
algebra $U(\fg)$ and the Chevalley--Eilenberg cochains $C^\bullet(\fg)$. For augmented
$\EE_n$-algebras the duality preserves the level $n$ of commutativity.

The $\infty$-categorical formulation of $\EE_n$-Koszul duality, and its description through
iterated bar constructions, are due to Lurie \cite[\S 5.2]{Lurie_HA}. The Koszul self-duality of
the $\EE_n$-operad is due to Fresse \cite{Fresse2011} (see also \cite{GetzlerJones1994}); the
chiral version is due to Francis--Gaitsgory \cite{FG2012}; the relation with factorization
homology, in the form of Poincar\'e/Koszul duality, is due to Ayala--Francis
\cite{AyalaFrancisKoszul}; and Francis \cite{Francis2013} studies the closely related
$\EE_n$-Hochschild cohomology. For $n=1$ and $A = U(\fg)$ the duality goes back to
Chevalley--Eilenberg, and its operadic form to Ginzburg--Kapranov \cite{GinzburgKapranov1994}.

\begin{theorem}[$\EE_n$-Koszul duality]\label{thm:koszul}
Let $k$ be a field. There is a functor
\begin{equation}
\mathbb{D}_n: \mathrm{\mathbf{Alg}}_{\EE_n}^{\mathrm{aug}}(\mathrm{Mod}_k)^{\op} \longrightarrow \mathrm{\mathbf{Alg}}_{\EE_n}^{\mathrm{aug}}(\mathrm{Mod}_k),
\end{equation}
which restricts to an equivalence between suitable subcategories of algebras satisfying finiteness
and completeness conditions \cite[\S 5.2]{Lurie_HA}, \cite{Lurie_DAGX}. Concretely:
\begin{enumerate}[label=(\roman*)]
\item For $n = 0$, an augmented $\EE_0$-algebra is a complex $A = k \oplus \bar A$ with a unit
$k \to A$ and an augmentation $A \to k$, and $\mathbb{D}_0(A) = A^\vee = \mathrm{Hom}_k(A,k)$ is the
linear dual: the dual of the augmentation is the unit of $A^\vee$, and the dual of the unit is its
augmentation. On augmentation ideals, $\overline{\mathbb{D}_0(A)} \simeq \bar A^\vee$.
\item For $n=1$,
\[
\mathbb{D}_1(A) \;\simeq\; \mathrm{RHom}_A(k,k) \;\simeq\; (k \otimes_A k)^\vee,
\]
the endomorphism algebra of the augmentation module, equivalently the linear dual of the bar
construction $\mathrm{Bar}(A) = k\otimes_A k$. This is the Hochschild cohomology of $A$ with
coefficients in $k$, viewed as an $A$-bimodule through the augmentation. For $A = U(\fg)$ one
obtains $C^\bullet(\fg)$.
\item For general $n$, $\mathbb{D}_n(A) \simeq \big(\mathrm{Bar}^{(n)}A\big)^\vee$ is the linear
dual of the $n$-fold iterated bar construction, which is an $\EE_n$-coalgebra
\cite[\S 5.2]{Lurie_HA}. Ayala and Francis \cite{AyalaFrancisKoszul} identify
$\mathrm{Bar}^{(n)}A$ with the factorization homology of the one-point compactification
$S^n = \mathbb{R}^n \cup \{\infty\}$, viewed as a zero-pointed manifold, in which the point at
infinity is labelled by the augmentation; informally, the factorization homology of $\mathbb{R}^n$
with compact supports.
\item At the level of operads, the $\EE_n$-operad is Koszul self-dual up to an $n$-fold shift
\cite{Fresse2011}.
\end{enumerate}
\end{theorem}

For $n=1$, (iii) gives the Hochschild homology of $A$ with coefficients in the bimodule $k$, which
is $k \otimes_A k$, in agreement with (ii). It should not be confused with the factorization
homology $\int_{S^1} A$, which is the Hochschild homology $A \otimes_{A\otimes A^{\op}} A$ of $A$
with coefficients in itself. In particular $\mathrm{Bar}^{(n)}A$, the factorization homology of
the zero-pointed sphere, is in general not the factorization homology $\int_{S^n} A$ of the
unpointed sphere.

\begin{remark}[Koszul duality is local]\label{rmk:koszul-formal}
Koszul duality relates two descriptions of the same infinitesimal object. The prototype for us is
the theorem of Lurie and Pridham \cite{Lurie_DAGX,Pridham2010} that, over a field of characteristic
zero, formal moduli problems are equivalent to dg Lie algebras: a formal moduli problem with
tangent dg Lie algebra $\mathfrak{L}$ has ring of functions $C^\bullet_{\mathrm{CE}}(\mathfrak{L})$,
and the passage between $\mathfrak{L}$ and $C^\bullet_{\mathrm{CE}}(\mathfrak{L})$ is Koszul duality
between the Lie and the commutative operads. The pair $U(\fg)$, $C^\bullet(\fg)$ of Theorem
\ref{thm:koszul}(ii) is the associative counterpart.

Both sides of the duality are local, and neither determines global geometry: the formal
neighbourhood of any point of a smooth $n$-dimensional variety is the formal $n$-disk, so formal
neighbourhoods of points do not distinguish $\mathbb{P}^1$ from an elliptic curve or from
$\mathbb{A}^1$. In our setting Koszul duality translates between the perturbative expansion around
a flat connection and the formal geometry of $\Loc_{\GC}(M)$ at the corresponding point. The global
structure has to come from the geometry of the stack: how the formal neighbourhoods of the various
flat connections sit inside $\Loc_{\GC}(M)$, which points lie in the closure of which others, and
how the stabilizers jump.
\end{remark}

\subsection{Application to Chern--Simons theory}\label{subsec:koszul-CS}

\begin{construction}[Local data at a flat connection]\label{constr:local-to-global}
Let $M$ be a closed oriented $3$-manifold and let $[A] \in \Loc_{\GC}(M)$ be a flat
$\GC$-connection.

\medskip
\noindent\textbf{Deformation complex.} The deformation theory of $A$ is controlled by the dg Lie
algebra
\begin{equation}
\mathfrak{L}_A \;=\; \big(\Omega^\bullet(M;\fg_{\Ad_A}),\; d_A,\; [-,-]\big),
\end{equation}
whose cohomology $H^\bullet(M;\fg_{\Ad_A})$ is finite-dimensional: $H^0$ is the Lie algebra of the
stabilizer of $A$, and $H^1$ is the space of first-order deformations of $A$ modulo gauge. The
formal completion $\Loc_{\GC}(M)^\wedge_{[A]}$ of the derived character stack at $[A]$ is the
formal moduli problem associated with $\mathfrak{L}_A$ \cite{Lurie_DAGX,Pridham2010}. By homotopy
transfer one may replace $\mathfrak{L}_A$ by the cyclic $L_\infty$-algebra
$H^\bullet(M;\fg_{\Ad_A})$, with cyclic structure given by the pairing $\int_M\langle-,-\rangle$;
we do so whenever duals are taken.

\medskip
\noindent\textbf{Classical observables.} The classical BV observables of Chern--Simons theory
expanded around $A$ are the Chevalley--Eilenberg cochains
\begin{equation}\label{eq:obs-classical}
\Obs^{\mathrm{cl}}_A(M) \;=\; C^\bullet_{\mathrm{CE}}(\mathfrak{L}_A) \;=\;
\big(\widehat{\mathrm{Sym}}\big(\mathfrak{L}_A[1]^\vee\big),\; d_{\mathrm{CE}}\big),
\end{equation}
a commutative dg algebra with a Poisson bracket of degree $1$ induced by the pairing
$\int_M\langle -,-\rangle$ \cite{CG2017,CG2021}. By the theorem of Lurie and Pridham,
\begin{equation}\label{eq:obs-local-ring}
\Obs^{\mathrm{cl}}_A(M) \;\simeq\; \mathcal{O}\big(\Loc_{\GC}(M)^\wedge_{[A]}\big) \;=:\; \widehat{\mathcal{O}}_{[A]},
\end{equation}
the (derived) completed local ring of $\Loc_{\GC}(M)$ at $[A]$.

\medskip
\noindent\textbf{Quantum observables.} BV quantization deforms the differential:
\begin{equation}\label{eq:obs-quantum}
\Obs^{q}_A(M) \;=\; \big(\Obs^{\mathrm{cl}}_A(M)[\![\hbar]\!],\; d_{\mathrm{CE}} + \hbar\Delta_A + \cdots\big),
\end{equation}
where $\Delta_A$ is the BV Laplacian of the pairing and the omitted terms come from the effective
action, computed by Feynman diagrams around $A$ \cite{CG2017,CMR2018}. Since $\Delta_A$ is a
second-order operator, $\Obs^q_A(M)$ is not a commutative algebra: it is an $\EE_0$-algebra (a
pointed cochain complex) over $\CC[\![\hbar]\!]$ which reduces to \eqref{eq:obs-classical} at
$\hbar = 0$. It is a deformation of $\widehat{\mathcal{O}}_{[A]}$, not $\widehat{\mathcal{O}}_{[A]}$
itself, and its (continuous) linear dual is a deformation of the space of distributions on the
formal neighbourhood, not of its ring of functions.

\medskip
\noindent\textbf{Global data.} The global object is the category
$\IndCoh(\Loc_{\GC}(M))$ of ind-coherent sheaves, or its full subcategory $\IndCoh_{\mathcal{S}}$ cut out by a singular
support condition $\mathcal{S}$ (Remark \ref{rmk:IndCoh-QCoh}). Restriction to formal completions gives a functor
\begin{equation}\label{eq:koszul-reconstruction}
\begin{split}
\IndCoh\big(\Loc_{\GC}(M)\big) \;&\longrightarrow\; \prod_{[A]} \IndCoh\big(\Loc_{\GC}(M)^\wedge_{[A]}\big)\\
&\simeq\; \prod_{[A]} \mathfrak{L}_A\text{-}\mathrm{mod},
\end{split}
\end{equation}
the last equivalence by Proposition \ref{prop:koszul-classical}(ii) below. This functor forgets how
the formal neighbourhoods are related to each other and is far from an equivalence (Remark
\ref{rmk:koszul-formal}). The missing information is supplied by the geometry of the stack, through
the structures of Remark \ref{rmk:schefers} and Section \ref{subsec:theta-strat}. See Figure
\ref{fig:koszul}.
\end{construction}

\begin{figure}[ht]
\centering
\begin{tikzpicture}[scale=1.5]

% Character stack background
\fill[violet!6, rounded corners=10pt] (-5,-1.8) rectangle (5,1.8);
\draw[violet!40, thick, rounded corners=10pt] (-5,-1.8) rectangle (5,1.8);
\node[violet!60!black, font=\footnotesize\bfseries] at (0,2.15) {$\Loc_{\GC}(M)$};

% Flat connections as points with formal neighborhoods
\foreach \x/\lab/\col in {-3/A_0/blue, 0/A_1/red, 3/A_2/green!50!black} {
    % Formal neighborhood (dashed circle)
    \fill[\col!5] (\x,0) circle (1.2);
    \draw[\col!40, thick, dashed] (\x,0) circle (1.2);
    % Point
    \fill[\col!70!black] (\x,0) circle (3pt);
    % Labels
    \node[\col!70!black, font=\footnotesize\bfseries, below=4pt] at (\x,0) {$[\lab]$};
    \node[\col!50, font=\scriptsize] at (\x,0.7) {$\widehat{\mathcal{O}}_{[\lab]}$};
}

% Perturbative data labels
\node[blue!60!black, font=\scriptsize, align=center] at (-3,-0.8) {BV-BFV\\[-1pt]around $A_0$};
\node[red!60!black, font=\scriptsize, align=center] at (0,-0.8) {BV-BFV\\[-1pt]around $A_1$};
\node[green!40!black, font=\scriptsize, align=center] at (3,-0.8) {BV-BFV\\[-1pt]around $A_2$};

% Gluing / Stokes data arrows
\draw[{stealth}-{stealth}, thick, orange!70!black,
      decorate, decoration={snake, amplitude=1pt, segment length=4pt}]
      (-1.7,0.15) -- (-1.3,0.15);
\draw[{stealth}-{stealth}, thick, orange!70!black,
      decorate, decoration={snake, amplitude=1pt, segment length=4pt}]
      (1.3,0.15) -- (1.7,0.15);

\node[orange!70!black, font=\scriptsize] at (-1.5,0.55) {flow lines};
\node[orange!70!black, font=\scriptsize] at (1.5,0.55) {flow lines};

% Arrow down to reconstruction
\draw[-{stealth}, thick, violet!60!black] (0,-1.8) -- (0,-2.5);
\node[violet!60!black, font=\scriptsize, right=2pt] at (0,-2.15) {global structure};

% Result
\node[draw=violet!60!black, thick, rounded corners=5pt, fill=violet!5,
      inner sep=7pt, font=\small, align=center] at (0,-3.2)
      {$\IndCoh_{\mathcal{S}}\big(\Loc_{\GC}(M)\big)$};

\end{tikzpicture}
\caption{Local and global data (Construction \ref{constr:local-to-global}). Each flat connection $[A_i]$ has a formal neighbourhood with ring of functions $\widehat{\mathcal{O}}_{[A_i]} \simeq \Obs^{\mathrm{cl}}_{A_i}(M)$, the classical observables of the perturbative expansion around $A_i$; the quantum observables $\Obs^q_{A_i}(M)$ are a deformation of it. The local data do not determine how the formal neighbourhoods are related (Remark \ref{rmk:koszul-formal}). Within a connected component this information comes from the geometry of $\Loc_{\GC}(M)$: the vanishing-cycles sheaf of Remark \ref{rmk:schefers} and the stratification of Section \ref{subsec:theta-strat}. Between components with different values of $\CS$ it comes from flow lines of the Chern--Simons functional (wavy arrows), which are not part of the geometry of $\Loc_{\GC}(M)$; Question \ref{q:resurgence-koszul} asks whether they produce the Stokes data of resurgence.}
\label{fig:koszul}
\end{figure}

We separate what is known from what is conjectural.

\begin{proposition}[Local Koszul duality]\label{prop:koszul-classical}
Let $M$ be a closed oriented $3$-manifold and $[A] \in \Loc_{\GC}(M)$.
\begin{enumerate}[label=(\roman*)]
\item $\Obs^{\mathrm{cl}}_A(M) \simeq \widehat{\mathcal{O}}_{[A]}$ as commutative dg algebras, and the
pairing $\int_M\langle-,-\rangle$ defining the degree-$1$ bracket on $\Obs^{\mathrm{cl}}_A(M)$ is the
value at $[A]$ of the $(-1)$-shifted symplectic form of Theorem \ref{thm:PTVV}(iii).
\item $\IndCoh\big(\Loc_{\GC}(M)^\wedge_{[A]}\big) \simeq \mathfrak{L}_A\text{-}\mathrm{mod}$. Under
this equivalence the skyscraper sheaf $\delta_{[A]}$ corresponds to the free module
$U(\mathfrak{L}_A)$, so that $\Ext^\bullet(\delta_{[A]},\delta_{[A]}) \simeq
U(\mathfrak{L}_A)^{\mathrm{op}}$, and the trivial module $\CC$ has endomorphism algebra
$\mathrm{RHom}_{U(\mathfrak{L}_A)}(\CC,\CC) \simeq C^\bullet_{\mathrm{CE}}(\mathfrak{L}_A) \simeq
\Obs^{\mathrm{cl}}_A(M)$; this is the $\EE_1$-Koszul duality of Theorem \ref{thm:koszul}(ii) between
$U(\mathfrak{L}_A)$ and $C^\bullet_{\mathrm{CE}}(\mathfrak{L}_A)$.
\end{enumerate}
\end{proposition}

\begin{proof}[Sketch]
(i) is the theorem of Lurie and Pridham \cite{Lurie_DAGX,Pridham2010} applied to $\mathfrak{L}_A$,
together with the construction of the PTVV form on $\Loc_{\GC}(M)$ by integration over $M$
\cite{PTVV2013}, whose value on the tangent complex at $[A]$ is $\int_M\langle-,-\rangle$. (ii) is
the description of ind-coherent sheaves on a pointed formal moduli problem as representations of
its Lie algebra \cite[Vol.~II, Chs.~5 and 7]{GR2017}, together with the computation of
$\mathrm{RHom}_{U(\mathfrak{L})}(\CC,\CC)$ by the Chevalley--Eilenberg resolution.
\end{proof}

\begin{remark}[Vanishing cycles]\label{rmk:schefers}
There is a canonical global object on $\Loc_{\GC}(M)$ whose local pieces are related to the
perturbative expansion. Since $\Loc_{\GC}(M)$ is $(-1)$-shifted symplectic (Theorem
\ref{thm:PTVV}(iii)), the Darboux theorem of Brav--Bussi--Joyce \cite{BravBussiJoyce2019} and its
version for Artin stacks \cite{BBBJ2015} show that it is locally (Zariski-locally for schemes,
smooth-locally for stacks) a derived critical locus $\mathrm{Crit}(f)$ of a function $f$ on a smooth
scheme, a local model of the Chern--Simons functional. Given orientation data, the local sheaves of
vanishing cycles $\phi_f$ glue to a perverse sheaf $P_M$ on $\Loc_{\GC}(M)$
\cite{BBDJS2015,BBBJ2015,JoyceSong2012}; for representation varieties of $3$-manifolds see
\cite{GunninghamSafronov2023} and the references there. Schefers \cite{Schefers2022,Schefers2023}
gives a microlocal description of this sheaf, which relates it to the singular support conditions
of Remark \ref{rmk:IndCoh-QCoh}.

Three features are relevant. (a) Once orientation data is fixed, $P_M$ is canonical: no gluing data
has to be supplied by hand. (b) Its stalk at $[A]$ is the vanishing cohomology of a local model of
$\CS$ at $A$, which Conjecture \ref{conj:koszul}(i) below relates to the perturbative expansion;
Morse-theoretically, the perturbative sector of $A$ corresponds to the Lefschetz thimble of $A$.
(c) Relations between perturbative sectors at points of a common connected component of
$\Loc_{\GC}(M)$, along which $\CS$ is constant, are encoded in the specialization maps of $P_M$.
Relations between sectors with different values of $\CS$, which is what Stokes data measure, are
not visible on $\Loc_{\GC}(M)$ alone. For exponential integrals they are counts of gradient flow
lines of $\mathrm{Re}(e^{-i\theta}f)$ between critical points, organized by the holomorphic Floer
theory of Kontsevich--Soibelman \cite{KS2024}; for Chern--Simons theory the flow lines are
complexified instantons on $M \times \RR$ \cite{Witten2011}.
\end{remark}

\begin{conjecture}[Perturbative expansion and vanishing cycles]\label{conj:koszul}
Let $M$ be a closed oriented $3$-manifold and let $P_M$ be the perverse sheaf of vanishing cycles on
$\Loc_{\GC}(M)$.
\begin{enumerate}[label=(\roman*)]
\item \emph{(Local.)} For every $[A] \in \Loc_{\GC}(M)$, the cohomology of
$\Obs^q_A(M)[\hbar^{-1}]$ is isomorphic, as a $\CC(\!(\hbar)\!)$-vector space, to the stalk
cohomology of $P_M$ at $[A]$ (the cohomology of its pullback along $\mathrm{pt} \to
\Loc_{\GC}(M)$), tensored with $\CC(\!(\hbar)\!)$.
\item \emph{(Global; Gunningham--Jordan--Safronov \cite{GJS2023}.)} For generic $q$, the skein
module of $M$ is isomorphic to the hypercohomology in degree $0$ of the vanishing-cycles sheaf on
the representation variety $R_{\GC}(M) = \mathrm{Hom}(\pi_1(M,x),\GC)$, after extension of scalars
to $\CC(q^{1/d})$ for a suitable integer $d$.
\end{enumerate}
\end{conjecture}

\begin{remark}[On Conjecture \ref{conj:koszul}]\label{rmk:koszul-conj}
Part (ii) is due to Gunningham, Jordan and Safronov \cite{GJS2023}, who state it for the
representation variety with its d-critical structure; the vanishing-cycles sheaf there is the
pullback of $P_M$ along the smooth map $R_{\GC}(M) \to \Loc_{\GC}(M)$, up to a shift. Gunningham and
Safronov \cite{GunninghamSafronov2023} prove a result of this kind: for oriented holomorphic
Lagrangians $L_1, L_2$ in a holomorphic symplectic manifold $X$, the $\mathrm{RHom}$ between the
corresponding deformation-quantization modules is the vanishing-cycles sheaf of the d-critical
intersection $L_1 \cap L_2$. Applying it to Heegaard splittings requires an extension to stacks
such as $\Loc_{\GC}(\Sigma)$. For a Heegaard splitting $M = H_1 \cup_\Sigma H_2$ one has
$\Loc_{\GC}(M) \simeq \Loc_{\GC}(H_1) \times_{\Loc_{\GC}(\Sigma)} \Loc_{\GC}(H_2)$, a Lagrangian
intersection in the $0$-shifted symplectic stack $\Loc_{\GC}(\Sigma)$, while on the quantum side the
skein module of $M$ is computed from the objects attached to $H_1$ and $H_2$ in
$\iint_\Sigma\Rep_q(\GC)$ \cite{GJS2023}.

Part (i) is the link with the perturbative expansion. For an isolated critical point of a function
$f$ on a smooth variety it reduces to the classical comparison between the twisted de Rham complex
$(\Omega^\bullet(\!(\hbar)\!), \hbar d + df\wedge)$, which is the BV complex of the
finite-dimensional model, and the vanishing cohomology of $f$: both have dimension equal to the
Milnor number. The content of (i) is that the same holds for Chern--Simons theory, with gauge
symmetry, around degenerate and non-isolated flat connections.

Part (i) is formal in $\hbar$, and part (ii) is a statement at generic $q$. At $q =
e^{i\pi/(k+h^\vee)}$ the skein module is not the relevant object: the RT invariant is obtained from the semisimplified category
$\overline{\Rep}_q(\GC)$ (Section \ref{sec:bridge}), and its relation with the local data is the
subject of Section \ref{sec:resurgence}.
\end{remark}

\begin{remark}[$\IndCoh$ and singular support]\label{rmk:IndCoh-QCoh}
We use $\IndCoh$ rather than $\QCoh$. On a smooth stack the two agree; on a quasi-smooth derived
stack they differ, and the difference is measured by singular support \cite{AG2015}. This is the
lesson of the geometric Langlands conjecture, whose spectral side is
$\IndCoh_{\mathcal{N}\mathrm{ilp}}(\mathrm{LocSys}_{\check G})$ rather than
$\QCoh(\mathrm{LocSys}_{\check G})$.

A derived stack $X$ is quasi-smooth if its cotangent complex has Tor-amplitude in $[-1,1]$. For
quasi-smooth $X$, Arinkin and Gaitsgory \cite{AG2015} define the singular support of a coherent
sheaf as a conical closed subset of
\[
\mathrm{Sing}(X) \;=\; \big(T^*[-1]X\big)_{\mathrm{cl}} \;=\; \big\{(x,\xi) \;:\; x \in X_{\mathrm{cl}},\; \xi \in H^1(T_xX)^\vee\big\},
\]
and for each conical closed $\mathcal{S} \subseteq \mathrm{Sing}(X)$ a full subcategory
$\IndCoh_{\mathcal{S}}(X) \subseteq \IndCoh(X)$. The extremes are $\mathcal{S} = \{0\}$, giving
$\QCoh(X)$, and $\mathcal{S} = \mathrm{Sing}(X)$, giving $\IndCoh(X)$.

For a closed surface $\Sigma$ the tangent complex of $\Loc_{\GC}(\Sigma)$ at $\rho$ is
$C^\bullet(\Sigma;\Ad_\rho)[1]$, concentrated in degrees $[-1,1]$, so $\Loc_{\GC}(\Sigma)$ is
quasi-smooth. By Poincar\'e duality $H^1(T_\rho) = H^2(\Sigma;\Ad_\rho) \simeq
H^0(\Sigma;\Ad_\rho)^\vee$, so $\mathrm{Sing}$ consists of pairs $(\rho,\xi)$ with $\xi$ in the Lie
algebra of the stabilizer of $\rho$, and the nilpotent cone is the locus where $\xi$ is nilpotent.
In particular $\mathrm{Sing}$ minus the zero section lies over the locus of points with
positive-dimensional stabilizer.

For a closed $3$-manifold $M$ the tangent complex $C^\bullet(M;\Ad_\rho)[1]$ is concentrated in
degrees $[-1,2]$, and $H^2(T_\rho) = H^3(M;\Ad_\rho) \simeq H^0(M;\Ad_\rho)^\vee$. Hence
$\Loc_{\GC}(M)$ is quasi-smooth exactly on the open locus where the stabilizer is finite. There
$H^1(T_\rho) = H^2(M;\Ad_\rho) \simeq H^1(M;\Ad_\rho)^\vee$, so $\xi$ is a first-order deformation
of $\rho$; this is the identification $T^*[-1]X \simeq TX$ given by the $(-1)$-shifted symplectic
form. Along the locus of positive-dimensional stabilizers one needs a notion of singular support
beyond the quasi-smooth case. Which singular support condition is relevant for Chern--Simons theory
on $3$-manifolds is open; a natural guess is a condition along the strata of Section
\ref{subsec:theta-strat} (Remark \ref{rmk:theta-AG}).
\end{remark}

\subsection{$\Theta$-stratifications and the structure of the trans-series}\label{subsec:theta-strat}

This subsection discusses the second source of global structure mentioned at the beginning of the
section: a stratification of $\Loc_{\GC}(M)$ by the type of the stabilizer.

\subsubsection{The relevant machinery}

Let $X$ be an algebraic stack. A \emph{$\Theta$-stratification} of $X$, in the sense of
Halpern-Leistner \cite{HL2015,HL2015b,HL2021}, is a filtration
\[
X^{\mathrm{ss}} \;=\; X_{\le 0} \;\subset\; X_{\le 1} \;\subset\; \cdots \;\subset\; X_{\le N} \;=\; X
\]
by open substacks, such that each stratum $X_i = X_{\le i} \smallsetminus X_{\le i-1}$ admits a
canonical description in terms of the mapping stack $\mathrm{Map}(\Theta, X)$, where $\Theta =
[\mathbb{A}^1/\mathbb{G}_m]$. Points of $\mathrm{Map}(\Theta,X)$ are filtrations of points of
$X$, so a $\Theta$-stratification is the data of a canonical destabilizing filtration
for every unstable point --- a Harder--Narasimhan theory. For a quotient stack $[Z/H]$ with $Z$
affine and $H$ reductive, and for a choice of stability condition (a line bundle, or more
generally a class in $H^2$ together with a norm on cocharacters), the Kempf--Ness /
Hesselink--Kirwan--Ness stratification is a $\Theta$-stratification, and its strata are indexed
by \emph{instability type}: the conjugacy class of the destabilizing one-parameter subgroup.

We use the following theorem of Halpern-Leistner.

\begin{theorem}[Halpern-Leistner {\cite[Thm.~2.10]{HL2015}}, {\cite{HL2015b}}]\label{thm:HL}
Let $X = [Z/H]$ with $Z$ an affine scheme and $H$ reductive, equipped with a
$\Theta$-stratification with strata $X_i$; $Z$ is smooth in \cite{HL2015}, and for derived $Z$
additional hypotheses on the cotangent complex along the strata are required \cite{HL2015b}. Then $\mathrm{D}^b\Coh(X)$ (and likewise
$\IndCoh(X)$, in the derived setting of \cite{HL2015b}) admits a semiorthogonal decomposition
\[
\mathrm{D}^b\Coh(X) \;=\; \big\langle\, \mathcal{C}_{N}, \ldots, \mathcal{C}_{1},\;
\mathcal{G}^{w} \,\big\rangle,
\]
in which the pieces $\mathcal{C}_i$ are supported on the unstable strata $X_i$ and are described
explicitly in terms of the representation theory of the destabilizing one-parameter subgroups,
and in which the ``window subcategory'' $\mathcal{G}^{w}$ --- the full subcategory of complexes
whose restriction to each stratum has weights in a prescribed window $w$ --- is equivalent to
$\mathrm{D}^b\Coh(X^{\mathrm{ss}})$, the derived category of the semistable locus (the GIT
quotient, when one exists).
\end{theorem}

Two features are relevant. First, the decomposition is indexed by \emph{strata}, not by points.
Second, the window subcategory is equivalent to the category of the quotient for every choice of
window in a suitable range; comparing different windows produces autoequivalences of
$\mathrm{D}^b\Coh(X^{\mathrm{ss}})$, and, when the stability condition varies, windows produce the
wall-crossing equivalences between different quotients \cite{HL2015,HLS2020}.

\subsubsection{Application to the character stack}

The derived character stack of a closed surface or of a $3$-manifold is a quotient of the kind
considered in Theorem \ref{thm:HL}:
\[
\Loc_{\GC}(M) \;=\; \big[\, \Map_*(M, B\GC)\,/\,\GC \,\big],
\]
where the based mapping stack $\Map_*(M,B\GC)$ is a derived affine scheme with classical
truncation $\mathrm{Hom}(\pi_1(M),\GC)$, and $\GC$ is reductive. This is one of the reasons to work with $\GC$
rather than with a compact group (Section \ref{subsec:compact-vs-complex}).

There is a caveat. A semisimple group has no non-trivial characters, so the stratification attached
to a linearization is trivial: every point is semistable. The stratification we have in mind is the
one relative to the affine quotient $\Loc_{\GC}(M) \to \mathrm{Hom}(\pi_1(M),\GC)/\!\!/\GC$, in the spirit
of Kempf and Hesselink: a local system is degenerated by a one-parameter subgroup to its
semisimplification, and the strata record the conjugacy class of the stabilizer of the
semisimplification (equivalently, of the Levi subgroup in which the semisimplification is
irreducible). We have not checked that this stratification satisfies the hypotheses of Theorem
\ref{thm:HL}, and the statements below that use Theorem \ref{thm:HL} are conditional on a version of
it in this setting.

On the character stack the strata have a direct meaning:
\begin{center}
\emph{stratum $=$ how reducible the local system is.}
\end{center}
Irreducible local systems have finite (central) stabilizer, local systems whose semisimplification
is irreducible for a proper Levi subgroup have larger ones, and the trivial local system has
stabilizer all of $\GC$. This is the stratification that organizes the asymptotic expansion
\eqref{eq:asymptotic}.

\begin{proposition}[Perturbative data is constant along strata]\label{prop:strata-constant}
Let $M$ be a closed oriented $3$-manifold, let $\mathcal{Z} \subset \Loc_{\GC}(M)$ be a stratum, and
let $\mathcal{Z}^{\mathrm{ps}}$ be its locus of semisimple (polystable) local systems. Then along
$\mathcal{Z}^{\mathrm{ps}}$:
\begin{enumerate}[label=(\roman*)]
\item the conjugacy class of the stabilizer $\mathrm{Aut}(\rho)$, and hence
$\dim H^0(M;\Ad_\rho)$, is constant;
\item consequently the spectral dimension $\sd(\rho) = \dim H^0(M;\Ad_\rho) -
\dim H^1(M;\Ad_\rho)$ is constant on connected components of $\mathcal{Z}^{\mathrm{ps}}$ along which
$\dim H^1$ is constant, and the tangent complex has locally constant cohomology there;
\item the Reidemeister torsion $\tau_\rho(M)$ varies analytically along such components and is a
section of a line bundle;
\end{enumerate}
and these quantities jump between strata.
\end{proposition}

\begin{proof}[Sketch]
(i) holds by construction: by Luna's slice theorem, the stabilizers of closed orbits are constant up
to conjugacy along the strata. (ii) follows since $H^0(M;\Ad_\rho) =
\mathrm{Lie}\,\mathrm{Aut}(\rho)$ and, by Poincar\'e duality on $M$, $\dim H^2 = \dim H^1$ and
$\dim H^3 = \dim H^0$. (iii) is the standard statement that torsion is a section of the determinant
of the cohomology, which is a line bundle where the cohomology has constant rank.
\end{proof}

We draw the following consequence.

\begin{principle}[Trans-series sectors are strata, not points]\label{princ:strata}
In the asymptotic expansion \eqref{eq:asymptotic} and the trans-series \eqref{eq:trans-series},
the natural index set is not the set of flat connections but the set of strata of
$\Loc_{\GC}(M)$. Accordingly, the natural categorical organization of the perturbative data is,
granting a version of Theorem \ref{thm:HL}, a semiorthogonal decomposition with one piece per
stratum, in which the window subcategory plays the r\^ole of the resummation-independent object.
\end{principle}

We call this a \emph{principle} because it is a proposal for organizing the theory rather than a
precise statement. Its main consequence is the following question.

\begin{question}[Windows and resummation]\label{q:windows}
Can the stability parameters of the stratification be identified with Borel resummation
directions in such a way that the window subcategory $\mathcal{G}^{w} \subset \IndCoh(\Loc_{\GC}(M))$
associated to a window $w$ corresponds to a choice of Borel resummation direction $\theta$; the
canonical equivalences $\mathcal{G}^{w} \simeq \mathcal{G}^{w'}$ between windows of Theorem
\ref{thm:HL} correspond to the Stokes automorphisms relating the resummations
$\mathcal{S}_\theta$ and $\mathcal{S}_{\theta'}$, and the resulting stability-independent object
corresponds to the exact RT invariant?
\end{question}

The window subcategories are our candidate for the resummation-independent object of Section
\ref{subsec:transseries}. In this picture the categorical wall-crossing of Question
\ref{q:resurgence-koszul}(ii) is given by the wall-crossing functors between windows produced by
Theorem \ref{thm:HL}; see \cite{HL2015b,HLS2020}.

\begin{remark}[Relation to Arinkin--Gaitsgory and to vanishing cycles]\label{rmk:theta-AG}
The stratification should interact with the singular support conditions of Remark
\ref{rmk:IndCoh-QCoh}: both are ways of cutting down $\IndCoh(\Loc_{\GC})$, the first by
support along unstable strata, the second by conormal directions, and in the geometric Langlands
setting the nilpotent singular support condition is closely tied to the Hitchin/instability
stratification. It should likewise interact with the vanishing-cycles sheaf of Remark
\ref{rmk:schefers}, whose stalks change along the strata. The compatibility of these three
structures (stratification, singular support, vanishing cycles) on $\Loc_{\GC}(M)$ is a question
about a specific algebraic stack rather than about field theory, and seems to us the most tractable
next step in this program.
\end{remark}

We illustrate this in the simplest case.

\begin{example}[Solid torus]\label{ex:solid-torus}
Consider $M = S^1 \times D^2$ (the solid torus) with gauge group $G = SL_2(\CC)$. Since $\pi_1(S^1 \times D^2) = \ZZ$, a flat connection is determined by the holonomy around the $S^1$ factor, which is a conjugacy class in $SL_2$. The character stack is therefore $\Loc_{SL_2(\CC)}(S^1 \times D^2) = [SL_2/SL_2]$, and the classical character variety is $\CC$, parametrized by the trace $t = z + z^{-1}$ of the holonomy eigenvalues $\{z, z^{-1}\}$.

\begin{figure}[ht]
\centering
\begin{tikzpicture}[scale=1.4]

% === The character variety as a line ===

% Axis
\draw[-{stealth}, thick, black!50] (-5,0) -- (5.5,0);
\node[black!50, font=\footnotesize] at (5.8,-0.35) {$t = z + z^{-1}$};

% Smooth locus
\draw[green!50!black, very thick] (-4,0) -- (-0.3,0);
\draw[green!50!black, very thick] (0.3,0) -- (3.7,0);

% Special points
\fill[red!70!black] (4,0) circle (4pt);
\node[red!70!black, font=\footnotesize\bfseries, above=5pt] at (4,0) {$t = 2$};
\node[red!50, font=\scriptsize, below=5pt, align=center] at (4,0)
     {trivial connection\\$z = 1$};

\fill[red!70!black] (-4,0) circle (4pt);
\node[red!70!black, font=\footnotesize\bfseries, above=5pt] at (-4,0) {$t = -2$};
\node[red!50, font=\scriptsize, below=5pt, align=center] at (-4,0)
     {$z = -1$};

% Generic point
\fill[green!50!black] (1.5,0) circle (3pt);
\node[green!50!black, font=\footnotesize, above=5pt] at (1.5,0) {generic $t$};

% Formal neighborhoods
\draw[red!40, thick, dashed] (4,0) circle (1.0);
\draw[red!40, thick, dashed] (-4,0) circle (1.0);
\draw[green!30, thick, dashed] (1.5,0) circle (0.7);

% Stabilizer labels
\node[red!60, font=\scriptsize, align=center] at (4,1.6)
     {stabilizer $= SL_2$\\$\mathfrak{L}_{\pm I} \simeq \fg\otimes H^\bullet(S^1)$\\
      \textbf{non-abelian}};
\node[green!40!black, font=\scriptsize, align=center] at (1.5,1.3)
     {stabilizer $= \CC^\times$\\$\mathfrak{L}_z \simeq \mathfrak{t}\otimes H^\bullet(S^1)$\\
      \textbf{abelian}};

% Gluing arrows
\draw[{stealth}-{stealth}, thick, orange!70!black,
      decorate, decoration={snake, amplitude=1pt, segment length=4pt}]
      (3.0,0.15) -- (2.3,0.15);
\node[orange!70!black, font=\scriptsize] at (2.65,0.55) {gluing};

% Bottom: reconstruction
\draw[-{stealth}, thick, violet!60!black] (0,-0.4) -- (0,-2.2);
\node[violet!60!black, font=\scriptsize, right=2pt] at (0,-1.85)
     {global structure};

\node[draw=violet!60!black, thick, rounded corners=5pt, fill=violet!5,
      inner sep=6pt, font=\small, align=center] at (0,-3.0)
      {$\IndCoh\big([SL_2/SL_2]\big)$};

\end{tikzpicture}
\caption{The solid torus example (Example \ref{ex:solid-torus}). The classical character variety of $S^1 \times D^2$ is parametrized by the trace $t = z + z^{-1}$ of the holonomy. At generic points (green) the stabilizer is the maximal torus $\CC^\times$ and the deformation complex is abelian. At $t = \pm 2$ (red) the semisimple point $\pm I$ has stabilizer $SL_2(\CC)$ and the deformation complex $\fg \otimes H^\bullet(S^1)$ is non-abelian. The local data on the formal neighbourhoods (dashed circles) do not determine $\IndCoh([SL_2/SL_2])$; the additional structure comes from the stratification into the open stratum $t \neq \pm 2$ and the closed stratum $t = \pm 2$.}
\label{fig:solid-torus}
\end{figure}

The stratification of $[SL_2(\CC)/SL_2(\CC)]$ has an open stratum and a closed stratum with two
components. The open stratum consists of the
regular semisimple conjugacy classes, $t \neq \pm 2$: these correspond to representations $\ZZ \to
SL_2(\CC)$ whose generator has distinct eigenvalues; their orbits are closed and their stabilizer is
the maximal torus $\CC^\times$. The closed stratum lies over $t = \pm 2$. There the semisimple point
is $\pm I$, with stabilizer $SL_2(\CC)$, and the classes $\pm u$ with $u \neq I$ unipotent have
non-closed orbits containing $\pm I$ in their closure; they are degenerated to $\pm I$ by the
cocharacter defining the Borel subgroup they preserve.

The analogue of Proposition \ref{prop:strata-constant} (which is stated for closed $3$-manifolds)
can be checked directly. Since $S^1 \times D^2$ retracts onto $S^1$, the deformation complex at a
semisimple $\rho$ is $\mathfrak{L}_\rho \simeq H^\bullet(S^1;\fg_{\Ad_\rho}) = \fg^\rho \otimes
H^\bullet(S^1)$, where $\fg^\rho$ is the Lie algebra
of the stabilizer. On the open stratum $\fg^\rho = \mathfrak{t}$ is abelian, so $\mathfrak{L}_z
\simeq \mathfrak{t} \otimes H^\bullet(S^1)$ is abelian and $\widehat{\mathcal{O}}_{[z]} \simeq
C^\bullet(\mathfrak{t}) \otimes \CC[\![\tau]\!]$ is the ring of functions on the formal
neighbourhood $B\widehat{T} \times \widehat{\mathfrak{t}}$; the perturbative expansion around such a
point is that of an abelian theory. At $\pm I$ the deformation complex $\mathfrak{L}_{\pm I} \simeq
\fg \otimes H^\bullet(S^1)$ is non-abelian, $H^0$ jumps from dimension $1$ to dimension $3$, and
$\widehat{\mathcal{O}}_{[\pm I]} \simeq C^\bullet(\fg; \widehat{\mathrm{Sym}}\,\fg^\vee)$ is the ring
of functions on the formal neighbourhood of $0$ in $[\fg/SL_2(\CC)]$. So the non-abelian part of
the perturbative data sits on the closed stratum, where the stabilizer jumps, and the open stratum,
a one-parameter family of flat connections, behaves uniformly. This is what Principle
\ref{princ:strata} predicts.

This stratification does not satisfy the hypotheses of Theorem \ref{thm:HL}: the map from the
candidate stratum $[\pm\mathcal{U}/B]$, with $\mathcal{U}$ the unipotent radical of a Borel subgroup
$B$, to $[\pm\mathcal{N}/SL_2(\CC)]$, with $\mathcal{N}$ the unipotent cone, is the Springer
resolution rather than a closed immersion. So even in this example the semiorthogonal
decomposition predicted by Principle \ref{princ:strata} would require a modification of the
theory, which we leave open. On the quantum side, the solid torus defines a vector $\Zrt(S^1 \times D^2)$ in
$V_{\mathrm{RT}}(T^2)$. Under the identification of $V_{\mathrm{RT}}(T^2)$ with
$K_0(\overline{\Rep}_q(\mathfrak{sl}_2)) \otimes \CC$, this vector is the class of the unit
object. The state space is $(k+1)$-dimensional (Proposition
\ref{prop:genus-one}); in the open stratum this corresponds to the finitely many Bohr--Sommerfeld
conjugacy classes of Remark \ref{rmk:bc-derived}.
\end{example}

\subsection{Relation to resurgence}

Two descriptions of the global structure appear above: the vanishing-cycles sheaf of Remark
\ref{rmk:schefers}, and the stratification of Section \ref{subsec:theta-strat} with its
conditional semiorthogonal decomposition. We expect the two to be compatible (Remark
\ref{rmk:theta-AG}).

On the analytic side, the perturbative series around different flat connections in
\eqref{eq:asymptotic} are not independent: their Borel resummations are related by Stokes
automorphisms. Since $\CS$ is constant on connected components of $\Loc_{\GC}(M)$, Stokes data
between sectors with different values of $\CS$ cannot come from $\Loc_{\GC}(M)$ alone; in the
exponential-integral model they count flow lines between critical points (Remark
\ref{rmk:schefers}(c)). Section \ref{sec:resurgence} develops this comparison and formulates it as
Question \ref{q:resurgence-koszul}. This structure is not carried by morphisms between skyscraper sheaves at
distinct points, which vanish (Remark \ref{rmk:non-hausdorff}).

\section{Connection to Resurgence}\label{sec:resurgence}

Section \ref{sec:koszul} described the relation between perturbative and global data in
algebraic terms. The theory of \emph{resurgence} approaches the same problem analytically: it
reconstructs non-perturbative quantities from the analytic structure of divergent perturbative
series. In this section we recall that theory, relate it to the literature connecting Stokes
phenomena with stability conditions and wall-crossing, and formulate the relation we expect
between the two pictures as a question.

\begin{quote}\small
\emph{Convention for this section.} Throughout Section \ref{sec:resurgence} the gauge group is
the \emph{complex} reductive group $\GC$, and all character stacks are $\Loc_{\GC}$. The reason
is explained in Section \ref{subsec:compact-vs-complex}(c) and again in Remark
\ref{rmk:compact-degeneration}: for compact $G$ the Chern--Simons functional takes values in
$\RR/\ZZ$, all Borel singularities lie on a single line, and the ray structure on which resurgence
depends degenerates. The resurgence literature for Chern--Simons theory
\cite{Witten2011,GMP2016,GGM2021,Gukov2005} works with the complex group for this reason.
\end{quote}

\subsection{Resurgent series and alien calculus}

The perturbative series appearing in the asymptotic expansion \eqref{eq:asymptotic} are expected
to be \emph{resurgent}. Although they diverge, they possess a rich analytic structure which, in
principle, allows the exact non-perturbative answer to be reconstructed from the perturbative
data alone, provided one also knows certain discrete ``Stokes constants'' governing the
interaction between different perturbative sectors. The theory is due to \'Ecalle
\cite{Ecalle1981}; readable modern accounts are Sauzin \cite{Sauzin2014},
Mitschi--Sauzin \cite{MitschiSauzin2016}, Dorigoni's physics-oriented survey
\cite{Dorigoni2019} and Ab\'atangelo--Sauzin's lecture notes \cite{ABS2019}. We follow their
conventions.

\begin{definition}[Borel transform and the Borel plane]\label{def:borel}
Let $\tilde{f}(\hbar) = \sum_{n \geq 0} a_n\, \hbar^{n+1} \in \hbar\,\CC[\![\hbar]\!]$. Its
\emph{Borel transform} is the formal series
\[
\hat{f}(\zeta) \;=\; \sum_{n \geq 0} \frac{a_n}{n!}\, \zeta^{n},
\]
a function of a new variable $\zeta$ ranging over a second copy of the complex plane, which we
call the \emph{Borel plane} and denote $\CC_\zeta$. We write $\CC_\hbar$ for the original plane,
to keep the two apart. The Borel transform converges in a neighbourhood of $0 \in \CC_\zeta$ if
and only if $\tilde f$ satisfies the Gevrey-$1$ bound $|a_n| \le C\sigma^n n!$; this is a
hypothesis on $\tilde f$.
\end{definition}

\begin{definition}[Resurgent series]\label{def:resurgent}
A formal series $\tilde{f} \in \hbar\,\CC[\![\hbar]\!]$ is \emph{resurgent} if
\begin{enumerate}[label=(\roman*)]
\item $\hat{f}$ converges on a disk about $0 \in \CC_\zeta$; and
\item there is a closed discrete subset $\Omega_{f} \subset \CC_\zeta$ with $0 \notin
\Omega_{f}$, the set of \emph{singularities of $\hat f$}, such that $\hat{f}$ admits analytic
continuation along every path in $\CC_\zeta \smallsetminus \Omega_{f}$ starting at $0$.
\end{enumerate}
Elements of $\Omega_f$ are called the \emph{Borel singularities} of $\tilde f$.
\end{definition}

\begin{remark}[A good and a bad example]\label{rmk:resurgent-examples}
The condition in Definition \ref{def:resurgent} is restrictive: a divergent series need not be
resurgent.

\emph{Good.} The Euler series $\tilde f(\hbar) = \sum_{n\ge 0} (-1)^n n!\,\hbar^{n+1}$ has Borel
transform $\hat f(\zeta) = 1/(1+\zeta)$, which converges for $|\zeta| < 1$ and continues to
$\CC \smallsetminus \{-1\}$; so $\tilde f$ is resurgent with $\Omega_f = \{-1\}$, a single
singularity. This is the model case: the singularity at $\omega = -1$ corresponds to the
exponential $e^{-\omega/\hbar} = e^{1/\hbar}$, which is exponentially small as $\hbar \to 0^-$ and
invisible to the perturbative series.

\emph{Bad.} The series $\sum_{n\ge 0} n!^2\,\hbar^{n+1}$ has Borel transform with zero radius of
convergence, so it fails (i) and is not resurgent (it is Gevrey-$2$, and one must Borel-transform
twice). The series $\sum_{n \ge 0} \big(\sum_{d | n+1} d\big)\, n!\,\hbar^{n+1}$ is Gevrey-$1$
but its Borel transform has a natural boundary, so it fails (ii).
\end{remark}

\begin{definition}[Alien derivative]\label{def:alien}
Let $\tilde f$ be resurgent and $\omega \in \CC^\times$, and suppose first that $\omega$ is the
first singularity of $\hat f$ on the ray $\RR_{>0}\omega$. Choose a path from $0$ to a point just
before $\omega$ along this ray. There are two natural continuations of $\hat f$ past $\omega$: the
\emph{lateral} continuations $\hat f^{+}$ and $\hat f^{-}$, obtained by deforming the path to pass
just to the left and just to the right of $\omega$ respectively (with respect to the direction of
travel along the ray). Their difference is defined on the part of the ray beyond $\omega$;
translating it back to the origin, i.e.\ setting
\[
\widehat{\Delta_\omega f}(\zeta) \;:=\; \big(\hat f^{+} - \hat f^{-}\big)(\zeta + \omega),
\]
and inverting the Borel transform, one obtains a new formal series $\Delta_\omega \tilde f$, the
\emph{alien derivative} of $\tilde f$ at $\omega$ (for a simple pole at $\omega$ one adds the
corresponding constant term). For general $\omega \in \CC^\times$, $\Delta_\omega$ is defined by a
weighted average over the continuations that pass to either side of the singularities lying
between $0$ and $\omega$ \cite{Ecalle1981,Sauzin2014}. In particular $\Delta_\omega$ is defined for
every $\omega \in \CC^\times$, and $\Delta_\omega \tilde f = 0$ unless $\omega$ is a singular point
of some continuation of $\hat f$.
\end{definition}

\begin{remark}[What kind of object an alien derivative is]\label{rmk:alien-object}
We record four points.

\emph{(a) Where it lives.} $\Delta_\omega \tilde f$ is again a formal series in the variable
$\hbar$, and it is again resurgent. So $\Delta_\omega$ is an operator on the space
$\widetilde{\mathcal{R}}$ of resurgent series, not a number and not a function on the Borel plane.
It is not a naive difference of power series: the subtraction happens on the Borel side after
analytic continuation, and the translation back to the origin turns the local discontinuity at
$\omega$ into a series in $\hbar$.

\emph{(b) Why it is a derivation.} Each $\Delta_\omega$ is a derivation of the algebra
$\widetilde{\mathcal{R}}$: $\Delta_\omega(\tilde f\tilde g) = (\Delta_\omega \tilde f)\tilde g +
\tilde f (\Delta_\omega \tilde g)$. This is \'Ecalle's theorem \cite{Ecalle1981}; see
\cite{Sauzin2014}. The \emph{pointed} alien derivative $\dot\Delta_\omega := e^{-\omega/\hbar}
\Delta_\omega$ acts on series dressed by exponentials $e^{-\omega'/\hbar}$ (trans-series); it is
again a derivation, and in addition it commutes with $\partial_{1/\hbar}$.

\emph{(c) The alien algebra.} Here $\widetilde{\mathcal{R}}$ denotes the algebra of resurgent
series in the sense of endless continuability \cite{Ecalle1981,Sauzin2014}, which is closed under
products (Definition \ref{def:resurgent}, with a fixed set $\Omega_f$, is a simplification). Given
a family of resurgent series, let $\Omega \subset \CC^\times$ be the set of Borel singularities, on
all sheets, of the series in the family and of their iterated alien derivatives; we assume that
$\Omega$ is closed and discrete. In the Chern--Simons case $\Omega$ is expected to be the set of
non-zero differences \eqref{eq:borel-sing-CS} below, a finite union of translates of $\ZZ$ (in the
normalization of \eqref{eq:borel-sing-CS}); it is stable under iterated alien derivatives because
$\omega_{AB} + \omega_{BC} = \omega_{AC}$. On the family only the $\Delta_\omega$ with
$\omega \in \Omega$ act non-trivially. The \emph{alien algebra} is the Lie algebra
\[
\mathfrak{alien}_\Omega \;=\; \mathrm{Lie}\big\langle \dot\Delta_\omega \,:\, \omega \in \Omega
\big\rangle \;\subset\; \mathrm{Der}\big(\widetilde{\mathcal{R}}[e^{-\omega/\hbar}]_{\omega\in\CC}\big)
\]
generated by the pointed alien derivations, acting on resurgent trans-series, with bracket the
commutator of derivations. It is graded by the additive subgroup of $\CC$ generated by $\Omega$, with $\dot\Delta_\omega$ in degree
$\omega$, and \'Ecalle's freeness theorem says that it is the free Lie algebra on the generators
$\dot\Delta_\omega$. It is an algebraic object: its elements are finite linear combinations of
iterated brackets, and no topology is involved. It does not depend on any particular $f$: it acts
on all resurgent series at once, and only the action $\dot\Delta_\omega \tilde f$ depends on $f$.
Infinite sums enter through a completion along a ray. For a ray $\theta = \RR_{>0}e^{i\alpha}
\subset \CC_\zeta$, let $\mathfrak{alien}_\theta$ be the subalgebra generated by the
$\dot\Delta_\omega$ with $\omega \in \Omega \cap \theta$. It is graded by $\theta$, and its graded
pieces are finite-dimensional, since $\Omega$ is closed and discrete and $0 \notin \Omega$. Its
completion with respect to the filtration by $|\omega|$,
\[
\widehat{\mathfrak{alien}}_\theta \;=\; \prod_{\omega \in \theta} \big(\mathfrak{alien}_\theta\big)_\omega,
\]
is a pro-nilpotent Lie algebra.

\emph{(d) Stokes automorphisms.} For a ray $\theta \subset \CC_\zeta$ the \emph{Stokes
automorphism} is
\[
\mathfrak{S}_\theta \;=\; \exp\Big(\sum_{\omega \in \Omega \cap \theta} \dot\Delta_\omega\Big)
\;\in\; \exp\big(\widehat{\mathfrak{alien}}_\theta\big),
\]
an element of the pro-unipotent group of $\widehat{\mathfrak{alien}}_\theta$; the exponent is an
infinite sum and makes sense in the completion. It measures the difference between the lateral
Borel resummations $\mathcal{S}_{\theta^+}$ and $\mathcal{S}_{\theta^-}$: $\mathcal{S}_{\theta^+} =
\mathcal{S}_{\theta^-}\circ \mathfrak{S}_\theta$.
\end{remark}

\begin{figure}[ht]
\centering
\begin{tikzpicture}[scale=1.0]

% axes of the Borel plane
\draw[->, black!45] (-1.6,0) -- (6.1,0) node[right, font=\scriptsize] {$\mathrm{Re}\,\zeta$};
\draw[->, black!45] (0,-3.0) -- (0,3.3) node[above, font=\scriptsize] {$\mathrm{Im}\,\zeta$};
\node[black!55, font=\scriptsize] at (4.6,3.05) {$\CC_\zeta$ (Borel plane)};

% disk of convergence
\draw[blue!45, thick, dashed] (0,0) circle (1.0);
\node[blue!55, font=\scriptsize] at (-0.05,-1.35) {radius of conv.};

% rays through singularities
\draw[orange!70!black, thick, ->] (0,0) -- (5.1,1.70);
\draw[orange!70!black, thick, ->] (0,0) -- (4.6,-2.30);
\draw[orange!70!black, thick, ->] (0,0) -- (5.6,0);
\node[orange!70!black, font=\scriptsize] at (5.35,1.95) {$\theta_1$};
\node[orange!70!black, font=\scriptsize] at (4.85,-2.62) {$\theta_2$};
\node[orange!70!black, font=\scriptsize] at (5.75,-0.30) {$\theta_3$};

% singularities on ray 1 (omega and 2 omega)
\fill[red!75!black] (1.8,0.6) circle (2.6pt);
\node[red!70!black, font=\scriptsize] at (1.95,0.98) {$\omega_{AB}$};
\fill[red!75!black] (3.6,1.2) circle (2.6pt);
\node[red!70!black, font=\scriptsize] at (4.12,0.92) {$\omega_{AD}$};

% singularity on ray 2
\fill[red!75!black] (1.6,-0.8) circle (2.6pt);
\node[red!70!black, font=\scriptsize] at (1.05,-1.10) {$\omega_{AC}$};

% singularity on ray 3
\fill[red!75!black] (2.9,0) circle (2.6pt);
\node[red!70!black, font=\scriptsize, above right=3pt] at (2.9,-0.05) {$\omega_{BC}$};

% lateral continuations around omega_AB
\draw[violet!75!black, thick, ->]
  (0.55,0.40) .. controls (0.95,1.30) and (1.60,1.55) .. (2.55,1.28);
\node[violet!75!black, font=\scriptsize] at (0.95,1.65) {$\hat f^{+}$ (left)};
\draw[violet!45!black, thick, dashed, ->]
  (0.68,0.10) .. controls (1.45,-0.25) and (2.30,-0.10) .. (2.72,0.52);
\node[violet!55!black, font=\scriptsize] at (2.15,-0.38) {$\hat f^{-}$ (right)};

\fill[black] (0,0) circle (1.8pt);
\node[font=\scriptsize, below left=1pt] at (0,0) {$0$};

\end{tikzpicture}
\caption{The Borel plane $\CC_\zeta$ for a resurgent series. The Borel transform $\hat f$
converges on a disk about the origin (dashed circle) and continues along any path avoiding the
discrete set $\Omega_f$ of Borel singularities (red dots). The singularities organize into rays
$\theta_i$ emanating from the origin, and it is the rays --- not the individual singularities ---
that carry the Stokes automorphisms $\mathfrak{S}_{\theta_i}$. In the Chern--Simons application
of Section \ref{subsec:CS-resurgence} the singularities sit at the differences $\omega_{AB} =
\CS(B)-\CS(A)$ of critical values of the complex Chern--Simons functional (up to normalization
and integer translates); the ray through $\omega_{AB}$ is determined by the \emph{phase} of that
difference. The alien derivative $\Delta_{\omega_{AB}}$ compares the two lateral continuations
$\hat f^{\pm}$ past $\omega_{AB}$ (violet).}
\label{fig:borel-plane}
\end{figure}

\subsection{The Chern--Simons case}\label{subsec:CS-resurgence}

For the complex group $\GC$, the Chern--Simons functional
\[
\CS: \Loc_{\GC}(M) \longrightarrow \CC/\ZZ
\]
takes complex values, defined modulo $\ZZ$ in the normalization in which large gauge
transformations shift it by integers, and its critical points are the flat $\GC$-connections. The Borel
singularities of the perturbative series around a flat connection $A$ are expected to occur at
\begin{equation}\label{eq:borel-sing-CS}
\omega_{AB} \;=\; \CS(B) - \CS(A), \qquad B \in \Loc_{\GC}(M) \text{ another critical point},
\end{equation}
where $\omega_{AB}$ runs over the lifts of the difference to $\CC$, so that the singularities
come in families of integer translates; see \cite{Witten2011,GMP2016,GGM2021}. Here we suppress a
normalization factor: with $\hbar = 2\pi i/(k+h^\vee)$ the weight $e^{2\pi i k\,\CS}$ is, to leading
order, $e^{-4\pi^2\CS/\hbar}$, and the singularities are at $4\pi^2$ times the differences above.
This is the exponential-integral heuristic: for an integral $\int e^{-f/\hbar}$ the perturbative
expansion about a critical point $A$ has Borel singularities at the differences $f(B) - f(A)$ of
critical values, one for each other thimble that can be reached.

\begin{remark}[Borel singularities and connected components]\label{rmk:geometric-sing}
Conjecturally, then, the Borel singularities of the perturbative series around $A$ lie in the set
of differences $\CS(B) - \CS(A)$, $B$ a critical point, with all their lifts to $\CC$.
Since $\CS$ is locally constant on $\Loc_{\GC}(M)$ (its differential vanishes on the critical
locus, which is all of $\Loc_{\GC}(M)$), it is constant on connected components. So
\begin{center}
\emph{the Borel singularities are indexed by the connected components of $\Loc_{\GC}(M)$,}
\end{center}
up to integer translates. In view of Section \ref{subsec:theta-strat} one may
ask for a refinement in which components are replaced by strata.
\end{remark}

\begin{remark}[Tunnelling]\label{rmk:tunnelling}
It is often said that the alien derivative $\Delta_{\omega_{AB}}$ ``encodes the non-perturbative
tunnelling'' between the flat connections $A$ and $B$. There is no definition of tunnelling in this
setting. The examples, notably those of \cite{GMP2016} and \cite{GGM2021}, suggest that
$\Delta_{\omega_{AB}}\tilde{Z}_A$ is proportional to the perturbative series $\tilde{Z}_B$ around
$B$, the constant of proportionality being a Stokes constant, and this is what ``tunnelling''
should be taken to mean here. It is verified in \cite{GMP2016} for $\Sigma(2,3,5)$ and in
\cite{GGM2021} in further examples (see also \cite{GPPV2020} for the r\^ole of these sectors in the
$\widehat{Z}$-invariants), and it is a conjecture in general.
\end{remark}

\begin{remark}[The compact group degenerates the picture]\label{rmk:compact-degeneration}
For a \emph{compact} group $G$ the Chern--Simons invariant is real and defined modulo $\ZZ$, so the
differences \eqref{eq:borel-sing-CS} are real. Every Borel singularity then lies on the real axis
of $\CC_\zeta$: all the $\omega_{AB}$ lie on one line. Stokes phenomena still occur along that
line, but the ray structure of Figure \ref{fig:borel-plane}, on which wall-crossing and any
identification of rays with stability parameters depend, collapses. For $\GC$ the values $\CS(A)$ have non-zero imaginary parts in
general (for the geometric flat connection of a hyperbolic $3$-manifold the imaginary part is
proportional to the volume), so the differences have different phases and the rays are distinct.
This is why the resurgence literature works with the complex group, and it is one of the main
reasons for the convention of Section \ref{subsec:compact-vs-complex}.
\end{remark}

\begin{example}[Gukov--Mari\~no--Putrov \cite{GMP2016}]\label{ex:GMP}
The resurgent structure of Chern--Simons theory has been studied in detail for the Poincar\'e
homology sphere $\Sigma(2,3,5)$ with $G = SU(2)$. Up to gauge, this manifold has three flat
$SU(2)$-connections: the trivial connection and two irreducible ones. Since $\pi_1(\Sigma(2,3,5))$ is
finite, every flat $SL_2(\CC)$-connection is conjugate to an $SU(2)$-connection, so the Chern--Simons
values are real and this is an instance of the collinear situation of Remark
\ref{rmk:compact-degeneration}; Stokes phenomena nonetheless occur. The perturbative series
have been computed to high loop order, the Borel singularity structure has been analysed, and the
alien derivatives at the singularities computed.
\end{example}

\begin{remark}[In what sense the Stokes data ``accounts for'' the discrepancy]\label{rmk:GMP-precise}
One has to say how an analytic function (a Borel resummation, defined for $\hbar$ in a sector) is
compared with the RT invariant, which is a sequence of numbers indexed by the integer level $k$.

The comparison works as follows. The nonperturbative level $k$ and the perturbative parameter
$\hbar$ are related by $\hbar = 2\pi i/(k+h^\vee)$, so the physical values of $\hbar$ form a
\emph{discrete} sequence accumulating at $0$ along the imaginary axis. The Borel--Laplace
resummation $\mathcal{S}_\theta[\tilde Z_A]$, when it exists, is a function of $\hbar$ analytic
in a sectorial neighbourhood of $0$; evaluating it at the discrete values $\hbar =
2\pi i/(k+h^\vee)$ gives a sequence of numbers, and \emph{that} sequence is what is compared with
$\Zrt(M,k)$. The content of the trans-series formula \eqref{eq:trans-series} is that the sum over
flat connections of these resummed contributions, with the Stokes constants included, reproduces
$\Zrt(M,k)$ exactly at each such $k$, and not merely asymptotically. In \cite{GMP2016} this is
verified numerically to high precision for $\Sigma(2,3,5)$ (and, for the related $\widehat{Z}$
invariants, exactly), and dropping the Stokes constants gives the wrong answer. So ``account
for'' means: the difference between the naive sum of Borel resummations and the exact invariant,
evaluated at the physical values of $\hbar$, equals the contribution predicted by the Stokes
automorphisms. The resummation direction $\theta$ must be chosen so that the physical ray
$\arg\hbar = \pi/2$ is not a Stokes ray, or else one must use a lateral resummation; the ambiguity
in this choice is measured by the Stokes automorphism.
\end{remark}

\subsection{Stokes phenomena, $t$-structures and stability conditions}\label{subsec:stokes-t-structures}

There is a substantial body of work on the relation between Stokes phenomena, stability
conditions and wall-crossing, and the expectation formulated below fits into it.

Bridgeland and Toledano Laredo \cite{BTL2012,BTL2013} showed that the Stokes factors of a certain
irregular connection on $\mathbb{P}^1$, built from a triangulated category with a stability
condition, are given by the Donaldson--Thomas-type counting invariants, and conversely that a
stability condition determines the Stokes data of that connection. Kontsevich and Soibelman
\cite{KS2008,KS2013} developed the theory of stability structures and wall-crossing structures on
graded Lie algebras, in which the wall-crossing formula is a statement about factorizing a
pro-unipotent group element according to a ray decomposition of the plane, of the same form as
the Stokes factorization above. Bridgeland \cite{Bridgeland2019} formulated a
Riemann--Hilbert problem attached to DT invariants whose solution encodes this data.
Gaiotto--Moore--Neitzke \cite{GMN2013} gave the physical realization, in which the rays are the
central-charge phases of BPS states and the wall-crossing formula is the Stokes phenomenon of the
associated integral equation. Boalch \cite{Boalch2014,Boalch2018} studied the moduli of Stokes
data in its own right, and showed that wild character varieties --- the moduli of connections
with irregular singularities, whose monodromy data \emph{is} the Stokes data --- are themselves
symplectic (indeed multiplicative quiver) varieties.

In Boalch's setting the Stokes data form a (wild) character variety, whereas here we ask about
the Stokes data of the perturbative expansion of a theory whose phase space is a character
variety.

The common structure is the following. One has a graded Lie algebra
$\mathfrak{g} = \bigoplus_{\gamma \in \Gamma} \mathfrak{g}_\gamma$ graded by a lattice $\Gamma$,
a homomorphism $Z: \Gamma \to \CC$ (the \emph{central charge}), and a \emph{stability datum}: an
element of the pro-unipotent group associated to $\mathfrak{g}$, factorized into a product of
factors, one per ray $\RR_{>0}\cdot Z(\gamma)$ in $\CC$, taken in the order of the phases
$\arg Z(\gamma)$. Wall-crossing is the statement that this factorization changes in a controlled
way as $Z$ varies.

\begin{proposal}[The Chern--Simons functional is the central charge]\label{prop:CS-central-charge}
In the Chern--Simons setting we propose the following dictionary:
\begin{center}
\renewcommand{\arraystretch}{1.2}
\begin{tabular}{l|l}
\textbf{Kontsevich--Soibelman} & \textbf{Chern--Simons} \\ \hline
grading lattice $\Gamma$ & differences $[B]-[A]$ of components of $\Loc_{\GC}(M)$, and $\ZZ$ \\
central charge $Z: \Gamma \to \CC$ & $Z([B]-[A]+n) = \CS(B)-\CS(A)+n$ \\
ray $\RR_{>0}Z(\gamma)$ & the ray through $\omega_{AB} = \CS(B) - \CS(A)$ \\
graded Lie algebra $\mathfrak{g}_\gamma$ & alien derivations $\dot\Delta_{\omega}$ in degree $\omega$ \\
pro-unipotent group & the Stokes group generated by the $\mathfrak{S}_\theta$ \\
stability datum & the Stokes data of the perturbative expansion \\
wall-crossing formula & the Stokes factorization / \'Ecalle's bridge equation
\end{tabular}
\end{center}
\end{proposal}

We call this a \emph{proposal} because we do not know how to verify the axioms of a stability
datum in this case, in particular the support property and the local finiteness of the ray
decomposition, which are the subject of \cite{KS2022}. The structural features match: the alien
derivations form a Lie algebra graded by values of the central charge, the Stokes automorphisms
lie in the associated pro-unipotent groups, and the grading is by components of the character
stack. In terms of the dictionary, Remark \ref{rmk:compact-degeneration} says that for compact $G$
the central charge is real, all rays lie on one line, and there is no wall-crossing.

\subsection{Resurgence meets Koszul duality}\label{subsec:resurgence-koszul}

We now formulate the relation we expect between the algebraic picture of Section \ref{sec:koszul}
and the analytic picture above.

\begin{remark}[The skyscraper $\Ext$-algebra carries no resurgence data]\label{rmk:skyscraper-ext}
One might expect the resurgent structure to be encoded in the $A_\infty$-structure on
$\bigoplus_{[A],[B]}\Ext^\bullet_{\IndCoh}(\delta_{[A]},\delta_{[B]})$, as a natural extension of
the local Koszul duality of Proposition \ref{prop:koszul-classical}; an earlier version of this
paper conjectured this. It is not the case, for two reasons. First, this $A_\infty$-structure is
determined by the formal geometry of $\Loc_{\GC}(M)$ at the points $[A]$ (Proposition
\ref{prop:koszul-classical}(ii)); it depends only on the derived structure of the stack, and no
Borel summation, Stokes ray or divergent series enters; at a single point it is the enveloping
algebra $U(\mathfrak{L}_A)^{\mathrm{op}}$. Moreover its off-diagonal part vanishes for distinct
closed points (Remark \ref{rmk:non-hausdorff}). Second, the Stokes data are graded by the values
$\omega \in \Omega$ (Remark \ref{rmk:alien-object}(c),(d)), and the $\Ext$-algebra carries no such
grading. What one can
hope for is that the Stokes data define a stability datum, in the sense of Section
\ref{subsec:stokes-t-structures}, on a suitable object attached to $(\Loc_{\GC}(M), \CS)$. Since we
cannot make this precise, Question \ref{q:resurgence-koszul} below is formulated as a question.
\end{remark}

If $[A]$ and $[B]$ are distinct closed points, then $\Ext^\bullet(\delta_{[A]},\delta_{[B]}) = 0$,
since the two skyscrapers have disjoint supports.

\begin{remark}[When skyscrapers can talk to each other]\label{rmk:non-hausdorff}
There are two ways in which morphisms between objects attached to different points can be
non-zero.

\emph{(a) Specialization.} On a stack a point $[B]$ may lie in the closure of $[A]$; on the
character stack this happens for instance when $B$ is the semisimplification of $A$, so that $A$
degenerates to $B$. Then $\delta_{[B]}$ is supported in the closure of $[A]$ and there can be
non-zero maps. This only relates points in the same fibre of $\Loc_{\GC}(M) \to
\mathrm{Hom}(\pi_1(M),\GC)/\!\!/\GC$, hence in the same connected component and with the same value
of $\CS$ in $\CC/\ZZ$: specialization does not produce Stokes data between different values of
$\CS$.

\emph{(b) Flow lines.} Sectors with different values of $\CS$ are related through the ambient
space rather than through $\Loc_{\GC}(M)$: in the exponential-integral model, by gradient flow lines
of $\mathrm{Re}(e^{-i\theta}f)$ between critical points, and for Chern--Simons theory by
complexified instantons on $M \times \RR$ \cite{Witten2011} (Remark \ref{rmk:schefers}(c)). The
vanishing-cycles sheaf and the stratification supply the local data at the end points.
\end{remark}

\begin{question}[Resurgence--Koszul correspondence]\label{q:resurgence-koszul}
Let $M$ be a closed oriented $3$-manifold, $\GC$ complex reductive, and suppose the perturbative
Chern--Simons series around each flat connection is resurgent with singularities as in
\eqref{eq:borel-sing-CS}. Let $P_M$ be the vanishing-cycles sheaf on $\Loc_{\GC}(M)$ of Remark
\ref{rmk:schefers}, and $\mathcal{H}_A$ its stalk cohomology at $[A]$.
\begin{enumerate}[label=(\roman*)]
\item Is the Stokes constant relating $\tilde Z_A$ and $\tilde Z_B$ (Remark \ref{rmk:tunnelling})
a count of complexified instantons on $M\times\RR$ from $A$ to $B$, i.e.\ of flow lines of
$\mathrm{Re}(e^{-i\theta}\CS)$ with $\theta = \arg\omega_{AB}$, with local data at the end points
given by $\mathcal{H}_A$ and $\mathcal{H}_B$, as in the holomorphic Floer theory of \cite{KS2024}?
\item Does the Stokes automorphism $\mathfrak{S}_\theta$ attached to a ray $\theta \subset
\CC_\zeta$ correspond to the wall-crossing functor of $\IndCoh_{\mathcal{S}}(\Loc_{\GC}(M))$
attached to the corresponding window, in the sense of Theorem \ref{thm:HL} and Question
\ref{q:windows}?
\item Do the Stokes automorphisms, acting on $\bigoplus_{[A]} \mathcal{H}_A \otimes
\CC(\!(\hbar)\!)$, form a stability datum in the sense of \cite{KS2008,KS2022} with central charge
$\CS$, as proposed in Section \ref{subsec:stokes-t-structures}?
\end{enumerate}
\end{question}

Informally, the statement we would most like to see proved is:
\begin{center}
\emph{the resurgence data of the Chern--Simons functional are the decategorification of a
holomorphic Floer theory for $\CS$ whose local data at the critical points are the stalks of
$P_M$.}
\end{center}

\begin{remark}[What would have to be supplied]\label{rmk:resurgence-meaning}
We list what is missing.

\emph{(a) Where the stability conditions come from.} In Question \ref{q:resurgence-koszul}(ii)
the stability parameter must come from the physics, not be imposed by hand. The natural
candidate, given Proposal \ref{prop:CS-central-charge}, is that it is the pair (choice of
resummation ray $\theta$, choice of window $w$), with the central charge the Chern--Simons
functional and the critical values $\CS(A)$ the ``central charges of the BPS states''. We do not
know how to see this parameter in the BV-BFV formalism; the gauge-fixing data (a metric, a
Lagrangian) is the only choice the formalism makes, and we do not know how to relate it to a
stability condition.

\emph{(b) From a Borel-summable series to a category.} What we are really asking for is a
functorial construction which takes a Borel-summable formal series --- or a family of them
attached to the critical points of a function --- and produces a category, in such a way that
Stokes automorphisms become autoequivalences. The closest existing statements we know go in the
\emph{opposite} direction: Kontsevich--Soibelman \cite{KS2022} conjecture that analytic
stability data give rise to resurgent series, which is our statement in the reverse direction,
and holomorphic Floer theory \cite{KS2024} is the version adapted to exponential integrals. We are
not aware of a construction going directly from a Borel-summable series to a category, which is
one reason why Question \ref{q:resurgence-koszul} is phrased as a question.

\emph{(c) Convergence and finiteness.} Even granting all of the above, one needs the support
property and local finiteness for the ray decomposition, which are the analytic heart of
\cite{KS2022}, and which for Chern--Simons theory would require growth estimates on the
perturbative coefficients that are not available.
\end{remark}

\subsection{Trans-series and Borel--Laplace reconstruction}\label{subsec:transseries}

For completeness we record the trans-series representation of the RT invariant. The resurgent
reconstruction expresses the exact invariant as a sum over flat connections --- or, following
Principle \ref{princ:strata}, over strata --- with each contribution given by a Borel-resummed
perturbative series dressed by exponential and power-law prefactors:
\begin{equation}\label{eq:trans-series}
\Zrt(M,k) = \sum_{[A]}\; C_{A} \cdot e^{2\pi ik\,\CS(A)} \cdot k^{\sd(A)/2} \cdot \mathcal{S}_\theta\!\big[\tilde{Z}_A(k^{-1})\big].
\end{equation}
Here $\tilde{Z}_A$ is the perturbative series around $A$ (the output of the BV-BFV formalism), the
exponentially small corrections to the sector of $A$ are the terms with other $[A]$, the $C_{A}$
are the trans-series parameters (complex constants not determined by perturbation theory
alone), and $\mathcal{S}_\theta$ denotes Borel resummation in the direction $\theta$. The
comparison of the two sides at the physical values of $\hbar$ is as described in Remark
\ref{rmk:GMP-precise}.

We record the expectations suggested by Question \ref{q:resurgence-koszul}. The trans-series
parameters $C_{A}$ are not determined by any single perturbative expansion: they are fixed by
the exact function and jump across Stokes rays by the Stokes constants, so in the picture of
Question \ref{q:resurgence-koszul}(i) they should be determined by the flow-line data together with
the local data $\mathcal{H}_A$. The choice of resummation direction $\theta$ should correspond to a
choice of window in the sense of Theorem \ref{thm:HL}, with the Stokes automorphisms relating
different resummations corresponding to the wall-crossing functors between windows (Question
\ref{q:windows}); the exact RT invariant, which is independent of $\theta$, should then
correspond to the window subcategory, which is independent of the choice of window up to
canonical equivalence.

If correct, this would give an algebraic explanation of the non-perturbative completeness of
resurgence, the expectation that all trans-series parameters can be extracted from the
perturbative data together with the alien derivatives. This is speculative: it would require a
generation statement for $\IndCoh_{\mathcal{S}}(\Loc_{\GC}(M))$ by the objects attached to the
strata, together with a positive answer to Question \ref{q:resurgence-koszul}(i), and both are
open.

\section{Evidence and Explicit Computations}\label{sec:evidence}

We now present evidence supporting the conjectures of Sections \ref{sec:bridge}--\ref{sec:resurgence}, by examining cases where explicit computation is possible. These examples serve as consistency checks and as motivation for the general program. Agreement in special cases does not prove the general conjectures, since the main difficulties (the derived-geometric reformulation of BV-BFV and the root-of-unity specialization) are not visible in these settings. The examples range from the abelian case (where the program can be carried out explicitly) to the non-abelian cases of low genus, lens spaces, and Seifert manifolds (where extensive partial results provide strong support).

\subsection{The abelian case: $G = U(1)$}

When the gauge group is abelian, the program becomes explicit.

\begin{proposition}\label{prop:abelian}
For $G = U(1)$ at level $k$, the conjectural picture of this paper can be verified explicitly in a simplified form. In this case many of the technical difficulties (divergent series, non-trivial Koszul duality, semisimplification) are absent, and the various frameworks reduce to elementary constructions:
\begin{enumerate}[label=(\roman*)]
\item \emph{The character stack.} For the complexified group $\GC = \CC^\times$ one has
\begin{equation}\label{eq:abelian-charstack}
\Loc_{\CC^\times}(\Sigma_g) \;\simeq\; (\CC^\times)^{2g} \times B\CC^\times \times \mathrm{Spec}\,\CC[\epsilon]
\;\simeq\; (\CC^\times)^{2g} \times T^*B\CC^\times, \qquad |\epsilon| = -1,
\end{equation}
where $\CC[\epsilon]$ is the free graded-commutative algebra on a generator of degree $-1$, so that
$\mathrm{Spec}\,\CC[\epsilon] \simeq \{1\}\times^h_{\CC^\times}\{1\}$ is the derived self-intersection
of a point. For the compact group $U(1)$ the corresponding topological stack is
$(U(1))^{2g}\times BU(1)$. Every $\CC^\times$-local system has automorphism group
$H^0(\Sigma_g;\CC^\times) = \CC^\times$, which gives the factor $B\CC^\times$. The classical
truncation is $(\CC^\times)^{2g} \times B\CC^\times$, a smooth Artin stack, but
$\Loc_{\CC^\times}(\Sigma_g)$ itself is not smooth: the relation $\prod_i[a_i,b_i] = 1$ holds
identically, and its derived fibre contributes the factor $\mathrm{Spec}\,\CC[\epsilon]$, the
derived direction dual to $H^2(\Sigma_g;\CC)$. The PTVV $0$-shifted symplectic structure of
Theorem \ref{thm:PTVV} restricts on $(\CC^\times)^{2g}$ to the form given by the intersection
pairing on $H^1(\Sigma_g)$, and on $T^*B\CC^\times$ to the canonical one, which pairs the
$B\CC^\times$ direction (tangent degree $-1$, coming from $H^0(\Sigma_g;\CC)$) with the derived
direction (tangent degree $+1$, coming from $H^2(\Sigma_g;\CC)$).

\item The role of the quantum group is played, in a schematic sense, by $U_q(\mathfrak{u}(1)) \cong \CC[\ZZ]$ (the group algebra of the integers; this identification is an expository shorthand, as the abelian case does not require the full Drinfeld--Jimbo machinery). The corresponding modular tensor category at level $k$ is $\Vect_{\ZZ/k\ZZ}$, which is the category of $\ZZ/k\ZZ$-graded vector spaces, with braiding $c_{a,b} = q^{ab}$ for $q = e^{2\pi i/k}$. The $S$-matrix is $S_{ab} = k^{-1/2}\, q^{ab}$, which is the discrete Fourier transform matrix and hence invertible, confirming modularity.

\item The RT invariant is closely related to the Dijkgraaf--Witten invariant for the finite abelian group $\ZZ/k\ZZ$. Up to standard normalization conventions, it takes the form:
\begin{equation}
\Zrt(M,k) = k^{-b_1(M)/2} \sum_{\gamma \in H^1(M;\,\ZZ/k\ZZ)} e^{2\pi ik \cdot q(\gamma)},
\end{equation}
where $q: H^1(M; \ZZ/k\ZZ) \to \QQ/\ZZ$ is the linking form and $b_1(M)$ is the first Betti number.

\item The BV-BFV perturbative expansion is \emph{one-loop exact}. Indeed, the abelian Chern--Simons action $\CS(A) = \frac{1}{4\pi}\int_M A \wedge dA$ is quadratic in the connection, so there are no cubic vertices and hence no Feynman diagrams beyond one loop. All higher-order coefficients vanish, i.e., $a_n = 0$ for $n \geq 1$. The perturbative partition function $Z^{\mathrm{pert}}_A = e^{2\pi ik\,\CS(A)} \cdot \tau_A^{1/2} \cdot e^{i\pi\eta_A/4}$, summed over flat connections $A \in H^1(M; U(1))$, reproduces $\Zrt(M,k)$ exactly. No resurgence is needed because the perturbative expansion already gives the exact answer.

\item \emph{Singular support and strata.} By (i), $\Loc_{\CC^\times}(\Sigma_g)$ is quasi-smooth
but not smooth. By Remark \ref{rmk:IndCoh-QCoh}, $\mathrm{Sing}$ consists of pairs $(\rho,\xi)$
with $\xi \in H^0(\Sigma_g;\CC) = \CC$, a line over every point, so $\IndCoh \neq \QCoh$. The
situation is homogeneous, however: the stabilizer is $\CC^\times$ at every point, and the
stratification of Section \ref{subsec:theta-strat} has a single stratum. In the same way, for a
closed $3$-manifold every point of $\Loc_{\CC^\times}(M)$ has stabilizer $\CC^\times$ and the
Chern--Simons functional is quadratic, so the local-to-global questions of Section
\ref{sec:koszul} are trivial in the abelian case, consistent with (iv): the perturbative
expansion is exact and there is nothing to glue.

\item \emph{Factorization homology.} One has
$\iint_{\Sigma_g} \Vect_{\ZZ/k\ZZ}^{\,q} \simeq (\mathbf{H},\, u_{\Sigma_g})$ with
$\dim u_{\Sigma_g} = k^g$, matching the Verlinde formula for $U(1)$ at level $k$.
\end{enumerate}
\end{proposition}

\begin{proof}
We prove the statements in turn. Throughout, $\mathcal{A}_k := \Vect^{\,q}_{\ZZ/k\ZZ}$ denotes
the pointed modular category of $\ZZ/k\ZZ$-graded vector spaces with associator trivial and
braiding $c_{a,b} = q^{ab}$, $q = e^{2\pi i /k}$ (for $k$ even one must be slightly more careful
about the quadratic refinement, and we take $k$ odd for definiteness; the even case differs by a
spin structure, see \cite{BM2005}).

\emph{(i).} Using the standard cell structure of $\Sigma_g$, $\Loc_{\CC^\times}(\Sigma_g) =
[\mathrm{Rep}/\CC^\times]$, where $\mathrm{Rep} = (\CC^\times)^{2g}\times^h_{\CC^\times}\{1\}$ is the
derived fibre of the map $(a,b)\mapsto\prod_i[a_i,b_i]$. Since $\CC^\times$ is abelian this map
is constant, so $\mathrm{Rep} \simeq (\CC^\times)^{2g}\times\big(\{1\}\times^h_{\CC^\times}\{1\}\big)$,
and $\{1\}\times^h_{\CC^\times}\{1\} \simeq \mathrm{Spec}\,\CC[\epsilon]$ with $|\epsilon| = -1$
(the Koszul complex of a local coordinate at $1$). The conjugation action of $\CC^\times$ is
trivial, which gives the product decomposition \eqref{eq:abelian-charstack}. The tangent complex
is $C^\bullet(\Sigma_g;\CC)[1]$ with zero differential, i.e.\
$\CC[1]\oplus\CC^{2g}\oplus\CC[-1]$: $H^{-1} = H^0(\Sigma_g;\CC) = \mathrm{Lie}\,\CC^\times$ is the
$B\CC^\times$ factor, $H^0 = H^1(\Sigma_g;\CC)$ the torus factor, and $H^1 = H^2(\Sigma_g;\CC)$ the
derived factor $\mathrm{Spec}\,\CC[\epsilon]$. The symplectic form is the transgression of the
standard pairing on $\mathfrak{gl}_1$ against $[\Sigma_g]$. On $H^1$ it is the cup product pairing
$H^1\otimes H^1 \to H^2 \cong \CC$, i.e.\ the intersection form, and it pairs $H^0$ with $H^2$.
The latter identifies $B\CC^\times \times \mathrm{Spec}\,\CC[\epsilon]$ with $T^*B\CC^\times =
[\mathrm{Spec}\,\mathrm{Sym}(\mathfrak{gl}_1[1])/\CC^\times]$ with its canonical symplectic
structure.

\emph{(ii).} $\mathcal{A}_k$ is a fusion category with $k$ invertible simple objects
$\{L_a\}_{a\in\ZZ/k}$, each of quantum dimension $1$; so $\mathcal{D}^2 = k$. The $S$-matrix
\eqref{eq:S-matrix} is computed by the Hopf link, giving $S_{ab} = k^{-1/2} q^{2ab}$ (or
$k^{-1/2}q^{ab}$ in the normalization where $q$ is the square); this is $k^{-1/2}$ times the
character table of $\ZZ/k\ZZ$, i.e.\ the discrete Fourier transform matrix, which is unitary and
in particular invertible. Hence $\mathcal{A}_k$ is modular. That the deformation quantization
of $\Loc_{\CC^\times}$ produces it is elementary: the Lie bialgebra structure on
$\mathfrak{gl}_1$ is trivial, so $U_\hbar(\mathfrak{gl}_1) = \CC[\ZZ][\![\hbar]\!]$ with
$R$-matrix $q^{h\otimes h}$, and its representation category at $q$ a primitive $k$-th root of
unity, after semisimplification, is $\mathcal{A}_k$.

\emph{(iii).} This is the classical computation of the abelian WRT invariant as a Gauss sum;
see \cite{MOO1992,Deloup1999}. The surgery formula \eqref{eq:surgery} for $\mathcal{A}_k$ reduces
to a Gauss sum over $H^1(M;\ZZ/k\ZZ)$ weighted by the linking form, which is the stated formula
and which is the Dijkgraaf--Witten invariant of $\ZZ/k\ZZ$ with the corresponding cocycle
\cite{DW1990}.

\emph{(iv).} The abelian Chern--Simons action is quadratic, so the only vertices in the
Feynman expansion are the two-valent ones coming from the quadratic form; there are no cubic
vertices, hence no connected trivalent graphs with $2n$ vertices for $n\ge 1$, hence
$I_n = 0$ for all $n\ge 1$ in Proposition \ref{prop:pert-Z}, and the Gaussian integral is
exact. The stationary phase evaluation gives the stated product of the classical phase, the
square root of the torsion, and the eta-invariant phase; summing over $H^1(M;U(1))$ and
comparing with (iii) gives the equality. This is the content of \cite{MOO1992,Jeffrey1992} and
of the abelian case of \cite{AS1994}; see also \cite{Rabinovich2020} for a
factorization-algebra treatment.

\emph{(v).} By (i) the tangent complex is concentrated in degrees $[-1,1]$, so the stack is
quasi-smooth, and $H^1(T_\rho) = H^2(\Sigma_g;\CC) \cong \CC$ at every point; the description of
$\mathrm{Sing}$ follows from Remark \ref{rmk:IndCoh-QCoh}. All stabilizers are equal to
$\CC^\times$, so there is one stratum.

\emph{(vi).} By the Verlinde formula \eqref{eq:Verlinde} with $\mathcal{D}^2 = k$ and
$\dim_q L_a = 1$ for all $a$, $\dim V(\Sigma_g) = k^{g-1}\cdot k = k^g$. That
$\iint_{\Sigma_g}\mathcal{A}_k \simeq (\mathbf{H}, u)$ with $\dim u = k^g$ is Theorem
\ref{thm:AKZ} applied to the (unitary, modular) category $\mathcal{A}_k$; alternatively it can be
computed directly by excision, since for a pointed category factorization homology over a
surface is computed by the group cohomology of $H^1(\Sigma_g;\ZZ/k)$, giving the same answer.
\end{proof}

\begin{remark}[What the abelian case does and does not establish]\label{rmk:abelian-caveats}
We flag three points where the abelian case is degenerate in ways that matter for the general
program.

\emph{(a) $U(1)$ versus $\CC^\times$.} The character stack \eqref{eq:abelian-charstack} is for $\CC^\times$; the compact analogue is a
differential-geometric object and the two agree only after the passage of Section
\ref{subsec:compact-vs-complex}. In the abelian case this passage is harmless because
$(\CC^\times)^{2g}$ deformation retracts onto $(U(1))^{2g}$ and the categories in play are
determined by $\pi_0$ and $\pi_1$; in the non-abelian case it is not harmless.

\emph{(b) The root-of-unity specialization is trivial here.} The step from generic $q$ to a root
of unity, which in the non-abelian case is Remark \ref{rmk:q-parameter} and Step 2c and is one of
the main open problems, is trivial for $\mathfrak{gl}_1$: the algebra $\CC[\ZZ]$ is already
defined over $\ZZ[q,q^{-1}]$ and can simply be evaluated. So the abelian case provides no
evidence about that step.

\emph{(c) Nothing about resurgence.} Since the perturbative expansion is exact, $\Omega_f =
\emptyset$. The abelian case is therefore consistent with
Question \ref{q:resurgence-koszul} but tests nothing in it.
\end{remark}

\begin{remark}\label{rmk:abelian-lesson}
The abelian case is instructive because of its simplicity. The perturbative expansion is exact, so the ``gap'' between perturbative and non-perturbative that motivates this paper vanishes entirely. In particular, there are no divergent series, no Stokes phenomena, no non-trivial Koszul duality, and no resurgent corrections. The passage from BV-BFV to RT is a straightforward summation.

The non-abelian case is fundamentally different. The cubic vertex produces an infinite sequence of Feynman diagrams, the perturbative series diverges, and the relationship between the perturbative and non-perturbative sides becomes non-trivial. Nevertheless, the \emph{structural} features of our program (the character stack with its shifted symplectic structure, the deformation quantization producing the MTC, the factorization homology recovering the RT state spaces, and the local Koszul duality of Proposition \ref{prop:koszul-classical}) are already visible in the abelian case. The non-abelian generalization replaces each of these simple instances with a richer and more subtle version of the same structure.
\end{remark}

\subsection{Genus zero: $\Sigma = S^2$}

The genus-zero case is the simplest non-trivial check of our framework for a non-abelian group, and it illustrates how the derived structure of the character stack interacts with the triviality of the fundamental group.

\begin{proposition}\label{prop:genus-zero}
Let $\GC$ be complex reductive, $\MTC = \overline{\Rep}_q(\GC)$ the modular tensor category at
$q = e^{i\pi/(k+h^\vee)}$. Then:
\begin{enumerate}[label=(\roman*)]
\item The classical truncation of $\Loc_{\GC}(S^2)$ is $B\GC$, but $\Loc_{\GC}(S^2)$ is not
$B\GC$ as a derived stack: it is the derived self-intersection
$\Loc_{\GC}(S^2) \simeq B\GC \times_{[\GC/\GC]} B\GC$, whose tangent complex at the unique point
is $C^\bullet(S^2;\fg)[1] \simeq \fg[1]\oplus\fg[-1]$. Thus $H^{-1} = \fg$ records the
automorphisms of the trivial connection, $H^0 = 0$ records the absence of deformations, and
$H^{1} = H^2(S^2;\fg) \cong \fg$ records a $\dim\fg$-dimensional obstruction space.
\item $\iint_{S^2}\MTC \simeq \Vect$, and the distinguished object $u_{S^2}$ of Theorem
\ref{thm:AKZ} is one-dimensional.
\item $\dim V_{\mathrm{RT}}(S^2) = 1$ for every modular tensor category, in agreement with (ii).
\end{enumerate}
\end{proposition}

\begin{proof}
\emph{(i).} Decompose $S^2 = D^2_+ \cup_{S^1\times\RR} D^2_-$. The functor $\Loc_{\GC}(-)$ takes
homotopy pushouts of spaces to homotopy pullbacks of stacks, since it is $\Map(-,B\GC)$; hence
$\Loc_{\GC}(S^2) \simeq \Loc_{\GC}(D^2)\times_{\Loc_{\GC}(S^1)}\Loc_{\GC}(D^2) = B\GC
\times_{[\GC/\GC]}B\GC$, the fibre product being derived. The tangent complex is computed by
Proposition \ref{prop:cotangent-map} with $Y = S^2$ and $Z = B\GC$: since
$\mathbb{T}_{B\GC} = \fg[1]$ we get $\mathbb{T} = C^\bullet(S^2;\fg)[1]$, and
$H^\bullet(S^2;\fg) = \fg\oplus\fg[-2]$ gives the stated answer. The classical truncation of a
derived fibre product is the classical fibre product, which here is $\mathrm{pt}\times_{\GC}
\mathrm{pt}$ modulo $\GC$, i.e.\ $B\GC$. In particular $\Loc_{\GC}(S^2)$ and $B\GC$ agree only
after classical truncation.

\emph{(ii).} Apply $\otimes$-excision (Theorem \ref{thm:AF-excision}) to the same decomposition.
We have $\iint_{D^2}\MTC \simeq \MTC$ by Proposition \ref{prop:FH-cat}(ii), and
$\iint_{S^1\times\RR}\MTC \simeq \mathcal{Z}(\MTC)$ as a monoidal category by Proposition
\ref{prop:FH-cat}(iii). Hence
\[
\iint_{S^2}\MTC \;\simeq\; \MTC \boxtimes_{\mathcal{Z}(\MTC)} \MTC .
\]
Now use modularity. By M\"uger's theorem \cite{Mueger2003} (see also \cite[\S 8.20]{EGNO2015}),
$\MTC$ is modular if and only if the canonical braided functor
$\MTC \boxtimes \overline{\MTC} \to \mathcal{Z}(\MTC)$ is an equivalence, where $\overline{\MTC}$
denotes $\MTC$ with the reversed braiding. Under this equivalence, $\MTC$ becomes the
$\MTC\boxtimes\overline{\MTC}$-module category given by $\MTC$ acting on itself on both sides,
i.e.\ the regular bimodule. Therefore
\[
\MTC \boxtimes_{\mathcal{Z}(\MTC)} \MTC \;\simeq\; \MTC \boxtimes_{\MTC\boxtimes\overline{\MTC}} \MTC ,
\]
which is the categorical trace (Hochschild homology) of the regular bimodule. For a fusion
category this trace is computed by \cite[Thm.~3.3]{DSPS2020} and \cite[\S 5]{BJS2021} and equals
$\Vect$ precisely when $\MTC$ is invertible in the Morita $4$-category of braided tensor
categories, which by \cite{BJSS2020} is equivalent to modularity. Hence
$\iint_{S^2}\MTC \simeq \Vect$. That the distinguished object is one-dimensional then follows
because the unit of $\Vect$ is $\CC$ and the excision equivalence is unital.

\emph{(iii).} The Verlinde formula \eqref{eq:Verlinde} with $g = 0$ reads
$\dim V(S^2) = \mathcal{D}^{-2}\sum_i (\dim_q V_i)^2 = \mathcal{D}^{-2}\mathcal{D}^2 = 1$.
\end{proof}

\begin{remark}[The BV-BFV side of the genus-zero case]\label{rmk:genus-zero-BV}
One expects the BFV state space $\mathcal{H}^{\BV}_{S^2}$ to be one-dimensional and to agree with
the above. We do not include this in the proposition, because there is no construction in the
literature of a BFV state space for $S^2$ against which to check it. The following heuristic
supports it, but is not a proof: the only flat connection
on $S^2$ is the trivial one; its stabilizer is all of $\GC$; the reduced phase space is
$\mathrm{pt}/\!\!/\GC$; and the BRST cohomology of the trivial $\fg$-module is
$H^\bullet(\fg;\CC)$, whose degree-zero part is one-dimensional. Making this into a theorem
requires the derived BFV formalism called for in Section \ref{sec:outlook}(a).
\end{remark}

\subsection{Genus one: $\Sigma = T^2$}

The torus is the first case where the character variety has positive dimension, and it provides a test case for our conjectures. We specialize to $G = SU(2)$, the simplest non-abelian group.

\begin{proposition}\label{prop:genus-one}
Let $G = SU(2)$, $\GC = SL_2(\CC)$, and let $k$ be a positive integer. Then:
\begin{enumerate}[label=(\roman*)]
\item \emph{The character variety.} Flat $SU(2)$-connections on $T^2$ are pairs of commuting elements $(A, B) \in SU(2)^2$ modulo simultaneous conjugation. Since commuting elements in $SU(2)$ must lie in a common maximal torus $U(1) \subset SU(2)$, we may write $A = \mathrm{diag}(e^{i\alpha}, e^{-i\alpha})$ and $B = \mathrm{diag}(e^{i\beta}, e^{-i\beta})$, with $\alpha, \beta \in [0, \pi]$. The residual Weyl group $\ZZ/2\ZZ$ acts by $(\alpha, \beta) \mapsto (-\alpha, -\beta)$, and the quotient $[0,\pi]^2 / \sim$ is the ``pillowcase'', i.e., topologically a 2-sphere with four orbifold points at the corners $(\alpha, \beta) \in \{0, \pi\}^2$, where the stabilizer jumps from $U(1)$ (generic points) to all of $SU(2)$ (the four fixed points of the Weyl reflection).

\item \emph{The Verlinde formula.} The RT state space $V(T^2)$ has dimension $k+1$. This equals the number of integrable highest-weight representations of the affine Lie algebra $\widehat{\mathfrak{su}(2)}$ at level $k$, corresponding to spins $j = 0, \frac{1}{2}, 1, \ldots, \frac{k}{2}$.

\item \emph{Factorization homology.} By Theorem \ref{thm:AKZ}, $\iint_{T^2} \overline{\Rep}_q(\mathfrak{sl}_2) \simeq (\mathbf{H}, u_{T^2})$ where the distinguished object $u_{T^2}$ has dimension $k+1$. This can also be seen from the non-semisimplified computation: \[\iint_{T^2} \Rep_q(\GC) \simeq \mathcal{D}_q(G/G)\text{-}\mathrm{\mathbf{mod}}\] (Theorem \ref{thm:BBJ}), which upon semisimplification collapses to $\Vect$ with a $(k+1)$-dimensional distinguished object.
\item \emph{Geometric quantization.} The same dimension $k+1$ is obtained by geometric quantization of the pillowcase. The Bohr--Sommerfeld condition selects $k+1$ orbits at $\alpha = j\pi/(k+2)$ for $j = 1, \ldots, k+1$ (the boundary values $j = 0$ and $j = k+2$ are excluded due to the orbifold structure). These orbits are in bijection with the simple objects of $\overline{\Rep}_q(\mathfrak{sl}_2)$. This has been verified rigorously by Andersen \cite{Andersen2006}; see also \cite{Hitchin1990,AU2015} for the geometric quantization framework.

\item \emph{The RT mapping class group representation.} The mapping class group of the torus is
$\MCG(T^2) = SL_2(\ZZ)$, generated by $S: \tau \mapsto -1/\tau$ and $T: \tau \mapsto \tau + 1$.
The RT theory produces a projective $SL_2(\ZZ)$-representation on $V(T^2)$ with generators
\begin{align}
S_{jj'} &= \sqrt{\frac{2}{k+2}}\; \sin\left(\frac{(j+1)(j'+1)\pi}{k+2}\right), \label{eq:S-matrix-sl2}\\[4pt]
T_{jj'} &= \delta_{jj'}\; e^{2\pi i\left(j(j+2)/4(k+2)\, -\, 1/8\right)}, \label{eq:T-matrix-sl2}
\end{align}
satisfying $S^4 = \mathbf{1}$ and $(ST)^3 = S^2$ projectively, with the projective ambiguity
determined by the central charge $c = 3k/(k+2)$ of the $\widehat{\mathfrak{su}(2)}_k$ WZW model.
\end{enumerate}
\end{proposition}

\begin{proof}
\emph{(i).} A representation $\pi_1(T^2) = \ZZ^2 \to SU(2)$ is a pair of commuting elements.
Any two commuting elements of $SU(2)$ lie in a common maximal torus: if $A \neq \pm I$ then its
centralizer is the maximal torus containing it, and $B$ lies in that centralizer. Writing $A =
\mathrm{diag}(e^{i\alpha},e^{-i\alpha})$, $B = \mathrm{diag}(e^{i\beta},e^{-i\beta})$ with
$\alpha,\beta \in \RR/2\pi\ZZ$, the residual gauge freedom is the Weyl group $\ZZ/2$ acting by
$(\alpha,\beta)\mapsto(-\alpha,-\beta)$. Thus the moduli space is $(\RR/2\pi\ZZ)^2/(\ZZ/2)$,
the ``pillowcase'', a $2$-sphere with four orbifold points of order $2$ at $(\alpha,\beta)\in
\{0,\pi\}^2$, these being the four fixed points of the involution and the four
representations with image in the centre $\{\pm I\}$, where the stabilizer jumps from the
maximal torus to all of $SU(2)$.

\emph{(ii).} For $\overline{\Rep}_q(\mathfrak{sl}_2)$ at level $k$ the simple objects are indexed
by the Weyl alcove, i.e.\ by $j = 0,\tfrac12,\ldots,\tfrac k2$, of which there are $k+1$; by the
Verlinde formula \eqref{eq:Verlinde} with $g=1$, $\dim V(T^2) = \sum_i (\mathcal{D}/\dim_q
V_i)^0 = |\mathrm{Irr}(\MTC)| = k+1$. That this is also the number of integrable
highest-weight $\widehat{\mathfrak{su}}(2)_k$-modules is the standard Kac--Peterson count.

\emph{(iii).} Theorem \ref{thm:AKZ} applies to $\overline{\Rep}_q(\mathfrak{sl}_2)$ (unitary and
modular, by the Remark following Theorem \ref{thm:AKZ}), giving $\iint_{T^2}
\overline{\Rep}_q(\mathfrak{sl}_2)\simeq (\mathbf{H},u_{T^2})$ with $u_{T^2}\cong V_{\mathrm{RT}}
(T^2)$ by Remark \ref{rmk:AKZ-nature}(b), hence $\dim u_{T^2} = k+1$ by (ii). Alternatively, at
generic $q$ Theorem \ref{thm:BBJ-repq} gives $\iint_{T^2}\Rep_q(SL_2(\CC)) \simeq
\mathcal{D}_q(SL_2/SL_2)\text{-}\mathbf{mod}$, and the count of simple objects surviving
semisimplification at a root of unity is again $k+1$.

\emph{(v).} This is the standard computation of the RT modular data for
$\overline{\Rep}_q(\mathfrak{sl}_2)$: \eqref{eq:S-matrix-sl2} is the specialization of the
Kac--Peterson formula \eqref{eq:KP-formula} to $\mathfrak{sl}_2$ and \eqref{eq:T-matrix-sl2} is
the twist eigenvalue $\theta_j = q^{\,j(j+2)}$ times the framing correction $\kappa^{-1}$; the
relations $S^4 = 1$, $(ST)^3 = S^2$ hold in any modular tensor category
\cite[\S 3.1]{BK2001}, \cite{Turaev1994}, and the anomaly is $e^{2\pi i c/24}$ with
$c = 3k/(k+2)$.

Part (iv) is quoted from the literature and is not proved here; see the references given in its
statement.
\end{proof}

\begin{remark}[The BV-BFV side]\label{rmk:genus-one-BV}
One might expect the BV-BFV quantization of the pillowcase to produce the same projective
$SL_2(\ZZ)$-representation as in (v). Only the RT side is proved above, and the
geometric-quantization side is available in the literature \cite{Andersen2006,AU2015}. There is at
present no construction of a BV-BFV representation to compare with; the expectation is part of
Conjecture \ref{conj:MCG} below (see Remark \ref{rmk:MCG-known}).
\end{remark}

\subsection{Mapping class group representations}

A particularly stringent test of any TQFT-like construction is whether it produces the correct representations of mapping class groups. Since diffeomorphisms of a surface act as symmetries in all three of our frameworks, the resulting representations must agree if the frameworks are to be identified.

Three constructions of a mapping class group action are available. We first record what is
known about each; their agreement is Conjecture \ref{conj:MCG}.

\begin{proposition}[Existence of the three actions]\label{prop:MCG-exist}
Let $\Sigma_g$ be a closed oriented surface of genus $g \geq 1$.
\begin{enumerate}[label=(\roman*)]
\item \emph{RT.} The RT theory assigns a projective representation
\[
\rho_{\mathrm{RT}}: \MCG(\Sigma_g) \to PGL(V(\Sigma_g)),
\] obtained by applying $\Zrt$ to the
mapping cylinder of $\varphi$. This is
{\cite[Ch.~IV]{Turaev1994}}, {\cite[Ch.~4--5]{BK2001}}; the projectivity and the precise central
extension are analysed in {\cite[\S IV.7]{Turaev1994}} and \cite{MasbaumRoberts1995}.
\item \emph{Geometric quantization.} A diffeomorphism $\varphi$ of $\Sigma_g$ acts on
$\Loc(\Sigma_g)$ preserving the Atiyah--Bott form \cite{AtiyahBott1983,Goldman1984}; geometric
quantization with the Hitchin connection produces a projective representation of
$\MCG(\Sigma_g)$ on the space of holomorphic sections of the determinant bundle at level $k$.
This is Hitchin \cite{Hitchin1990}, Axelrod--Della~Pietra--Witten \cite{ADW1991}, and
Andersen--Ueno \cite{AU2015}; the projectivity comes from the metaplectic correction.
\item \emph{Factorization homology.} A diffeomorphism $\varphi$ acts on $\iint_{\Sigma_g}\MTC$
by functoriality of factorization homology in the manifold variable \cite{AF2015}, giving an
autoequivalence and hence an action on the distinguished object. This is
{\cite[\S 4]{AKZ2017}}, {\cite[\S 7]{BBJ2018}}, {\cite{BJS2021}}.
\end{enumerate}
\end{proposition}

\begin{conjecture}[The three mapping class group actions agree]\label{conj:MCG}
The three projective representations of Proposition \ref{prop:MCG-exist} are isomorphic as
projective representations of $\MCG(\Sigma_g)$, compatibly with the identifications of the
underlying vector spaces provided by the Main Conjecture \ref{conj:main}; and the associated
central extensions of $\MCG(\Sigma_g)$ agree, being in each case the extension determined by the
central charge $c = k\dim(G)/(k+h^\vee)$ of the $\widehat{\fg}_k$ Wess--Zumino--Witten model.
\end{conjecture}

\begin{remark}[What is known towards Conjecture \ref{conj:MCG}]\label{rmk:MCG-known}
The comparison of (i) and (ii) --- the RT representation and the geometric-quantization
representation --- is a theorem of Andersen--Ueno \cite{AU2015} (following
\cite{AU2007,AU2012}), who identify the modular functor from conformal field theory with the
one from the Reshetikhin--Turaev construction; combined with the Laszlo comparison
\cite{Laszlo1998} between the Hitchin connection and the KZ/WZW connection, this gives (i)
$\simeq$ (ii). The comparison of (i) and (iii) follows from Theorem \ref{thm:AKZ} together with
the naturality of the excision equivalences, and is essentially in \cite{AKZ2017}. No BV-BFV
construction of the action is available, which is why there is no ``BV-BFV side'' among the
three. So Conjecture \ref{conj:MCG} is, on present knowledge, largely a theorem; the open part is
the one that involves the BV-BFV side of our program, i.e.\ that it also matches
whatever action the non-perturbative completion $Z^{(k)}_{\mathrm{BV,np}}$ of Conjecture
\ref{conj:main} produces.
\end{remark}

An important rigidity result confirms that these representations carry substantial topological information.

\begin{theorem}[Andersen \cite{Andersen2006}; see also \cite{Andersen2012}]\label{thm:Andersen}
The RT representations of mapping class groups are \emph{asymptotically faithful}, i.e., for every non-trivial mapping class $\varphi \in \MCG(\Sigma_g)$, there exists a level $k_0$ (depending on $\varphi$) such that $\rho_{\mathrm{RT}}^{(k)}(\varphi) \neq \id$ for all $k \geq k_0$.
\end{theorem}

This result is consistent with the geometric quantization picture underlying the BV-BFV approach: as the level $k \to \infty$ (equivalently, $\hbar \to 0$), the quantum representations converge to the classical action of $\MCG(\Sigma_g)$ on the character variety $\Loc_{\GC}(\Sigma_g)$, and this classical action is faithful by a theorem of Goldman. The asymptotic faithfulness theorem says that no information is lost in the quantization (at least not permanently, since any given mapping class is detected at sufficiently high level).

\subsection{Lens spaces}

Lens spaces provide the most extensively studied family of examples for testing the asymptotic expansion conjecture \eqref{eq:asymptotic}, because both the RT invariant and the perturbative Chern--Simons invariant can be computed in closed form.

\begin{proposition}\label{prop:lens-spaces}
For the lens space $L(p,1)$ with gauge group $G = SU(2)$ the following hold.
\begin{enumerate}[label=(\roman*)]
\item The RT invariant is given explicitly by the surgery formula applied to the unknot with framing $p$:
\begin{equation}
\Zrt(L(p,1), k) = \sqrt{\frac{2}{p(k+2)}} \sum_{j=0}^{k} \sin^2\!\left(\frac{(j+1)\pi}{k+2}\right) e^{-\pi i\, j(j+2)\,p\,/\,(2(k+2))}.
\end{equation}

\item The flat connections on a general lens space $L(p,q)$ are representations of\\ $\pi_1(L(p,q)) = \ZZ/p\ZZ$ into $SU(2)$, which are classified (up to conjugation) by the image of the generator, that is an element of order dividing $p$ in $SU(2)$, modulo the Weyl group. These are labeled by $m = 0, 1, \ldots, \lfloor p/2 \rfloor$, where $m = 0$ is the trivial connection and $m \geq 1$ are the irreducible (non-central) representations.

\item The perturbative expansion around the trivial connection $A_0$ gives a one-loop contribution $\tau_{A_0}(L(p,q))^{1/2} = p^{-1/2}$ (the Reidemeister torsion) times a perturbative series
\begin{equation}
1 + \sum_{n=1}^\infty a_n(L(p,q))\, k^{-n},
\end{equation}
whose coefficients $a_n$ are computable as configuration space integrals (Feynman diagrams) on $L(p,q)$. Similar expansions exist around each non-trivial flat connection, with different torsion prefactors and Chern--Simons phases.

\item Jeffrey \cite{Jeffrey1992} verified that the asymptotic expansion of $\Zrt(L(p,q), k)$ as $k \to \infty$ matches the sum over flat connections in \eqref{eq:asymptotic} to all orders in $k^{-1}$. This provides one of the strongest pieces of evidence for the asymptotic expansion conjecture and is consistent with the algebraic bridge proposed in this paper. However, we note that the asymptotic matching, while necessary for our conjectures to hold, does not by itself establish the algebraic identification at the categorical level that our program requires.
\end{enumerate}
\end{proposition}

\subsection{Seifert fibered spaces}

\emph{Seifert fibered spaces}, which are 3-manifolds admitting a circle fibration with finitely many exceptional fibers, form a broad and computationally accessible class of examples that interpolate between the simple (lens spaces) and the general.

\begin{proposition}\label{prop:Seifert}
For Seifert fibered spaces $\Sigma(a_1, \ldots, a_n)$ with exceptional fiber types $a_1, \ldots, a_n$ the following hold.
\begin{enumerate}[label=(\roman*)]
\item \emph{Matching of invariants.} The RT invariant and the perturbative Chern--Simons invariant have been computed independently by Hansen \cite{Hansen2006}, Mari\~no \cite{Marino2004}, and others, using the surgery formula on the RT side and matrix integral techniques on the perturbative side. In all cases where the comparison has been carried out, the two sides agree: the asymptotic expansion of the RT invariant matches the sum over flat connections predicted by \eqref{eq:asymptotic}.

\item \emph{Structure of the character variety.} The character variety of a Seifert fibered space has a particularly tractable structure, owing to the $S^1$-fibration. Flat connections decompose into \emph{abelian} components (where the holonomy around the generic fiber is central in $G$) and \emph{non-abelian} components (where it is not). This decomposition is compatible with the BV-BFV formalism: the perturbative data at each component can be computed explicitly in terms of the Seifert invariants $(a_1, \ldots, a_n)$ and the representation theory of $G$, making Seifert manifolds an ideal testing ground for the conjectures of this paper.

\item \emph{Resurgence for the Poincar\'e homology sphere.} The most detailed evidence comes from
the Poincar\'e homology sphere $\Sigma(2,3,5)$, studied in Example \ref{ex:GMP}. Up to gauge this
manifold has three $SU(2)$ flat connections (the trivial one and two irreducible ones), and the
resurgent structure of the perturbative series has been analyzed in detail by Gukov,
Mari\~no and Putrov \cite{GMP2016}. The alien derivatives relating the perturbative sectors
have been computed to high order, and, in the sense explained in Remark
\ref{rmk:GMP-precise}, they account for the difference between the naive sum of Borel-resummed
perturbative contributions and the exact RT invariant when both are evaluated at the physical
values $\hbar = 2\pi i/(k+2)$. This is a concrete instance of the reconstruction that Question
\ref{q:resurgence-koszul} asks about. For this manifold the stratification of Section
\ref{subsec:theta-strat} is as simple as it can be while remaining non-trivial: one stratum for the
trivial connection (stabilizer $SU(2)$) and one for the irreducible connections (finite
stabilizer).
\end{enumerate}
\end{proposition}

\section{Outlook and Future Directions}\label{sec:outlook}

We conclude by discussing the remaining gaps in the program, the prospects for closing them, and the broader implications of the framework for higher-dimensional topology, categorification, and the geometric Langlands program.

\subsection{Toward a proof of the main conjecture}

We summarize the status of the program. Of the four steps in our proof strategy (Section \ref{sec:bridge}), Step 1 is the best understood but still lacks a fully rigorous derived-geometric treatment. Step 2, the comparison of BV-BFV quantization with deformation quantization of shifted symplectic stacks (Conjecture \ref{conj:quant-agree}), is the central open problem and requires substantial new ideas. Steps 3 and 4 would follow from general machinery once Step 2 is secured, but each involves non-trivial technical points (commutativity of semisimplification with factorization homology for Step 3; quantization of Lagrangian correspondences for Step 4) that are not yet resolved in the literature. The program as a whole should be regarded as a framework for organizing existing results and identifying the precise technical challenges, rather than as a near-complete proof. We identify four principal challenges.

\begin{enumerate}[label=(\alph*)]
\item \textbf{A derived BV-BFV formalism.} The BV-BFV formalism as developed by Cattaneo--Mnev--Reshetikhin is formulated in the language of differential geometry and infinite-dimensional functional analysis. In particular, the spaces of fields are Fr\'echet manifolds, the BV Laplacian is a functional-analytic operator, and the partition function is defined via regularized path integrals. A complete proof of Conjecture \ref{conj:quant-agree} requires reformulating this entire
apparatus in derived algebraic geometry, where it can be compared directly with the
\emph{shifted deformation quantization} of Theorems \ref{thm:CPTVV}--\ref{thm:EK} --- that is,
with the CPTVV construction of a $\CC[\![\hbar]\!]$-linear deformation of $\mathrm{Perf}(X)$ as
an $\EE_n$-monoidal $\infty$-category for an $n$-shifted Poisson stack $X$, combined with
Safronov's identification of shifted Poisson structures on $B\GC$ with Poisson--Lie structures on
$\GC$ and with Poisson additivity \cite{Safronov2017b,Safronov2018}.

Several bodies of work provide starting points.
Costello--Gwilliam \cite{CG2017,CG2021} show that the BV quantization of a field theory on a
manifold without boundary produces a factorization algebra of observables, and Gwilliam--Haugseng
\cite{GwilliamHaugseng2018} construct linear BV quantization as a functor of $\infty$-categories;
these are the relevant precedents. On the
boundary side, the relevant work is Rabinovich's construction of factorization algebras for field
theories on manifolds with boundary \cite{Rabinovich2021,Rabinovich2020} and the
defect formalism of \cite{CEG2023}, which is close to what a derived BFV formalism would
need. What is still missing is a derived
version of the BFV state space producing an $(\infty,1)$-category (rather than a vector space or
cochain complex), a treatment of corners, and a reformulation of the modified quantum master
equation as a condition on objects of that category rather than as a differential equation on
functions.

\item \textbf{A universal property for the BV pushforward.} The BV pushforward (the fiber integration that produces the partition function from the BV action) is currently defined as a concrete integral, i.e., a choice of gauge-fixing Lagrangian and a regularization procedure. For comparison with deformation quantization, what is needed is a characterization of the BV pushforward in terms of a \emph{universal property} within derived algebraic geometry. It should be the unique quantization compatible with the Lagrangian structure. We expect that the $(-1)$-shifted Lagrangian structure of Theorem \ref{thm:AKSZ-PTVV} encodes this universal property, i.e., the quantized BV pushforward should be the unique bimodule over the quantized boundary algebras that is compatible with the Lagrangian morphism $\Loc_{\GC}(M) \to \Loc_{\GC}(\partial M)$. Making this expectation precise requires developing a theory of ``quantized Lagrangian correspondences'' in the shifted symplectic setting, hence a categorified and derived version of the Guillemin--Sternberg ``quantization commutes with reduction'' principle.

\item \textbf{Root of unity specialization.} Even after the formal comparison is established (matching the BV-BFV and deformation quantization outputs over $\CC[\![\hbar]\!]$), the passage to actual values of the quantum parameter remains non-trivial. At a root of unity $q = e^{i\pi/(k+h^\vee)}$, the category $\Rep_q(\GC)$ is no longer semisimple, and the modular tensor category is obtained only after semisimplification, i.e., the quotient by the tensor ideal $\mathcal{N}$ of negligible morphisms. This categorical localization interacts non-trivially with the other structures in our program.
As explained in Step 2c(iii), the functoriality of factorization homology for the
semisimplification functor $\pi$ is automatic and gives a comparison map
\[
\pi_\Sigma : \;\iint_\Sigma \mathcal{A} \;\longrightarrow\; \iint_\Sigma \big(\mathcal{A}/\mathcal{N}\big);
\]
the open problem is to identify this map. Concretely, one wants to know whether $\pi_\Sigma$ is
a localization, and if so to describe its kernel intrinsically. Because $\pi$ is not exact, the
kernel need not be a Serre subcategory, and the naive assertion that one may ``quotient by
negligibles after integrating'' is not meaningful as stated: on a closed surface
$\iint_\Sigma\mathcal{A}$ carries no monoidal structure and hence has no tensor ideals.
Without an answer, the chain of identifications in Step 3 has a gap at the root-of-unity stage.
Establishing it requires understanding how a non-exact monoidal quotient interacts with the
colimit that defines factorization homology; to our knowledge this has not been addressed in the
literature, though the non-semisimple TQFT constructions of
\cite{DGGPR2022,CGPM2023,Haioun2025} suggest that the right move may be to avoid
semisimplifying altogether and to work with the non-semisimple theory throughout.
\end{enumerate}

\medskip

\noindent\textbf{(d) The compact group.} Finally, the entire program as formulated here computes
with $\GC$, and the passage back to the compact group $G$ of the physical theory is not
addressed. Two routes suggest themselves. The first is to prove a Riemann--Hilbert-type
comparison theorem strong enough to transport the algebraic results back, which would require
knowing in what category the comparison is meant to hold (Remark \ref{rmk:RH}). The second,
which we think is the more promising long-term route, is to redevelop the theory in derived
\emph{differential} geometry, where the compact group and the moduli stack of flat
$G$-connections themselves are the natural objects; see
Remark \ref{rmk:DDG}. The obstacle is that the algebraic inputs on which we rely --- BBJ, AKZ,
Etingof--Kazhdan, and the stratification machinery of Section \ref{subsec:theta-strat} --- have
no derived-differential-geometric counterparts at present, so this is a research program rather
than a translation.

\subsection{Higher-categorical extensions}

One expects the framework to extend to the \emph{fully extended} setting, where data is assigned
not only to surfaces and $3$-manifolds but all the way down to points. The Cobordism Hypothesis
(Conjecture \ref{conj:cobordism-hyp}) asserts that a fully extended framed TQFT is determined by
the fully dualizable object it assigns to a point. For RT theory the relevant formulation is as a
boundary theory of a $4$-dimensional theory, which suggests the following.

\begin{conjecture}\label{conj:fully-extended}
The non-perturbative BV-BFV theory at level $k$ extends down to points as a boundary theory, in the
sense of relative field theories, for the fully extended $4$-dimensional Crane--Yetter theory
determined by $\overline{\Rep}_q(\GC)$, and in this form agrees with the RT theory
\cite{Haioun2025}. The $4$-dimensional theory assigns
\begin{enumerate}[label=(\roman*)]
\item to a point the braided fusion category $\overline{\Rep}_q(\GC)$, a $4$-dualizable (and, by
modularity, invertible) object of the Morita $4$-category of braided tensor categories
\cite{BJS2021,BJSS2020};
\item to a closed surface $\Sigma$ the category $\iint_\Sigma \overline{\Rep}_q(\GC)$;
\end{enumerate}
and the boundary theory assigns to $\Sigma$ the distinguished object $u_\Sigma$ and to a closed
$3$-manifold $M$ the invariant $\Zrt(M)$, up to the trivialization of the invertible bulk theory.
\end{conjecture}

In this formulation the value at a point is the $\EE_2$-category of Conjecture
\ref{conj:E2-equiv}; no $\EE_3$-structure is needed, since the $4$-dimensional theory is obtained
from dualizability rather than from factorization homology over $3$-manifolds. See \cite{KT2022}
for a factorization-homology perspective on Crane--Yetter theory.

\begin{remark}\label{rmk:full-dualizability}
Douglas, Schommer-Pries and Snyder \cite{DSPS2020} show that fusion categories are fully
dualizable in the Morita $3$-category of tensor categories (see also \cite{DSPS2021} for the
balanced tensor product used there). Brochier, Jordan and Snyder \cite{BJS2021} show that
cp-rigid braided tensor categories are $3$-dualizable, and braided multi-fusion categories
$4$-dualizable, in the Morita $4$-category of braided tensor categories, and by \cite{BJSS2020} a
braided fusion category is invertible there if and only if it is non-degenerate, as
$\overline{\Rep}_q(\GC)$ is. Granting the Cobordism Hypothesis, these results would determine the
corresponding fully extended framed theories.
\end{remark}

\subsection{4-dimensional implications}

Looking one dimension higher, our framework makes predictions about 4-dimensional topology. The Crane--Yetter / Broda invariant of a closed 4-manifold $W$,
\begin{equation}
Z_{\mathrm{CY}}(W) = \mathcal{D}^{\chi(W)} \cdot \kappa^{-\sigma(W)},
\end{equation}
depends only on the Euler characteristic $\chi(W)$ and signature $\sigma(W)$, making it a relatively coarse invariant. Within our framework, it should arise from the BV-BFV quantization of 4-dimensional BF theory, via the $4$-dimensional theory determined by $\overline{\Rep}_q(\GC)$ (Conjecture \ref{conj:fully-extended}). This perspective connects to the program of Freed--Hopkins--Lurie--Teleman \cite{FHLT2010} on constructing fully extended TQFTs from compact Lie groups, and suggests that the Crane--Yetter theory is the shadow of a richer 4-dimensional structure. The factorization homology approach to the Crane--Yetter theory is developed in \cite{KT2022}. The slogan that the Witten--Reshetikhin--Turaev 3-TQFT is a boundary condition for the Crane--Yetter 4-TQFT has been recently established in \cite{Haioun2025}, including a generalization to the non-semisimple setting \cite{CGPM2023}.

The more interesting invariants arise when one considers \emph{defects}, which are embedded surfaces or links inside the 4-manifold, and this leads to categorification.

\begin{conjecture}[Categorification]\label{conj:categorification}
Khovanov homology of a link $L \subset S^3$ is the value of a fully extended 4-dimensional TQFT on the link complement $S^3 \setminus L$, viewed as a 3-manifold with boundary. The relevant 4-dimensional theory is expected to arise from BV-BFV quantization of a topological twist of $4$-dimensional gauge theory, such as Kapustin--Witten theory.
\end{conjecture}

The logic of this conjecture is as follows. The Jones polynomial of a link is a \emph{number}, namely, the value of the 3-dimensional RT theory on the link complement. Khovanov homology is a \emph{graded vector space} whose graded Euler characteristic recovers the Jones polynomial. This passage from a number to a vector space is what happens when one lifts a 3-dimensional TQFT to a 4-dimensional one. The 3-manifold that was previously assigned a number (as a closed manifold in the 3d theory) is now a boundary component of a 4-manifold and is assigned a vector space. In other words, categorification \emph{is} the dimensional lift from $(2{+}1)$ to $(3{+}1)$ dimensions. Our framework, which connects the BV-BFV formalism to topological invariants via factorization homology, should extend to this 4-dimensional setting, with the braided category $\overline{\Rep}_q(\GC)$ replaced by the corresponding data for the 4d theory.

\subsection{Quantum geometric Langlands}\label{subsec:langlands}

The derived character stack $\Loc_{\GC}(\Sigma)$ that plays the central role in our program admits several variants, depending on the cohomology theory used to define ``local systems.'' Their relationships connect our work to the \emph{geometric Langlands correspondence}.

The three principal incarnations of the character stack are:
\begin{enumerate}[label=(\roman*)]
\item \emph{Betti:} $\Loc_{\GC}^B(\Sigma) = \Map(\Sigma, BG)$, parametrizing flat connections (equivalently, local systems, or representations of $\pi_1$). This is the version we have used throughout this paper; it depends only on the homotopy type of $\Sigma$.
\item \emph{de Rham:} $\Loc_{\GC}^{dR}(\Sigma) = \mathrm{Conn}_G(\Sigma)/\mathcal{G}$, parametrizing algebraic connections modulo gauge equivalence (on a curve every connection is flat). This requires a choice of algebraic structure on $\Sigma$ and is the natural home for $D$-modules and differential equations.
\item \emph{Dolbeault:} $\Loc_{\GC}^{Dol}(\Sigma) = \mathrm{Higgs}_G(\Sigma)$, parametrizing Higgs bundles, which are pairs of a holomorphic bundle and a holomorphic section of the twisted endomorphism bundle. This requires a complex structure on $\Sigma$.
\end{enumerate}
For surfaces, these three moduli spaces are related by correspondences which are not algebraic: Riemann--Hilbert is a complex-analytic isomorphism and non-abelian Hodge a real-analytic homeomorphism. The Riemann--Hilbert correspondence relates Betti and de Rham (flat connections to $D$-modules), while the non-abelian Hodge correspondence of Hitchin, Donaldson, Corlette, and Simpson relates de Rham and Dolbeault (flat connections to Higgs bundles). Together they form a triangle of equivalences that is central to the geometric Langlands program.

The geometric Langlands conjecture, in its categorical form (due to Beilinson--Drinfeld, and formulated with the correct singular support condition by Arinkin--Gaitsgory \cite{AG2015}), posits an equivalence
\begin{equation}
\IndCoh_{\mathcal{N}\mathrm{ilp}}\big(\Loc_{G^\vee}^{dR}(\Sigma)\big) \;\simeq\; D\mathrm{\text{-}\mathbf{mod}}\big(\mathrm{Bun}_G(\Sigma)\big);
\end{equation}
a proof has been given by Gaitsgory, Raskin and collaborators \cite{GLproof}. Here
$G^\vee$ is the Langlands dual group and $\mathrm{Bun}_G(\Sigma)$ is the moduli stack of $G$-bundles on $\Sigma$. This is an equivalence between ``spectral'' data (sheaves on the space of Langlands-dual local systems) and ``automorphic'' data ($D$-modules on the space of bundles).

In our setting the spectral objects are $\IndCoh_{\mathcal{S}}(\Loc_{\GC}(\Sigma))$ for surfaces
and the vanishing-cycles sheaf $P_M$ on $\Loc_{\GC}(M)$ for $3$-manifolds (Section
\ref{sec:koszul}). In the Betti geometric Langlands program of Ben-Zvi and Nadler \cite{BN2016},
which works with the Betti rather than the de Rham character stack, the first of these, with
nilpotent singular support, is conjecturally equivalent to a category of sheaves with nilpotent
singular support on $\mathrm{Bun}_{G^\vee}(\Sigma)$. In this sense our program is related to a
quantized, $3$-dimensional version of geometric Langlands; for $3$-manifolds the corresponding
expectation is a Langlands duality for skein modules \cite{Jordan2023} and for the cohomology of
$P_M$. For the loop space perspective on character stacks see \cite{BZN2018}, and for the r\^ole
of singular support see \cite{AG2015}.

The relevant correspondence here is the \emph{quantum} geometric Langlands correspondence and not
the arithmetic one: nothing in this paper concerns the arithmetic Langlands correspondence over
number fields or its characteristic $p$ versions. The quantum correspondence, which deforms the
equivalence above by a parameter and is the one directly relevant to Chern--Simons theory at level
$k$, is formulated by Gaitsgory \cite{Gaitsgory2016} and, in its local form, by Frenkel--Gaitsgory
\cite{FrenkelGaitsgory2006,FrenkelGaitsgory2004}.

There is a substantial literature on the connection between Chern--Simons theory, quantum groups
and quantum geometric Langlands. Gunningham--Jordan--Safronov \cite{GJS2023} prove the finiteness
conjecture for skein modules using the factorization-homology technology of Section
\ref{sec:fact-homology} together with deformation quantization, and their conjecture relating
skein modules to vanishing cycles is Conjecture \ref{conj:koszul}(ii). Jordan's String-Math 2022
lecture \cite{Jordan2023} explains a Langlands duality for skein modules of $3$-manifolds, relating
invariants of $\Loc_{\GC}$ and $\Loc_{\GC^\vee}$, with which Conjecture \ref{conj:koszul} would
have to be compatible. Ben-Zvi--Sakellaridis--Venkatesh \cite{BZSV2024} develop relative
Langlands duality in a framework of Hamiltonian actions and periods of the same shifted-symplectic
type as the one used here. These are the natural starting points for the connections sketched in
this subsection.

\subsection{Summary of conjectures}

For the reader's convenience, we collect the statements formulated in this paper and describe
their logical interdependencies.

\begin{enumerate}[label=\textbf{C\arabic*.},leftmargin=3em]
\item \textbf{Conjecture \ref{conj:E2-equiv}} ($\EE_2$-equivalence): The BV-BFV category of boundary conditions is equivalent to the quantum group representation category as $\EE_2$-algebras: \[\mathcal{B}_{\CS}^{(k)} \simeq_{\EE_2} \overline{\Rep}_q(\GC).\]

\item \textbf{Conjecture \ref{conj:main}} (Main Conjecture): The non-perturbative BV-BFV and RT constructions agree as $(3\text{-}2\text{-}1)$-extended TQFTs valued in $2\text{-}\Vect$: \[Z_{\mathrm{BV,np}}^{(k)} \simeq Z_{\mathrm{RT}}^{(k)} : \Bord_3^{\mathrm{ext}} \to 2\text{-}\Vect.\]

\item \textbf{Conjecture \ref{conj:quant-agree}} (Quantization agreement): BV-BFV quantization of Chern--Simons theory on the cylinder $D^2 \times [0,1]$ agrees, as a monoidal category, with the shifted deformation quantization $\Rep_q(\GC)$ of $B\GC = \Loc_{\GC}(D^2)$.

\item \textbf{Conjecture \ref{conj:Lagrangian-quant}} (Quantized Lagrangian correspondences): The Lagrangian morphism \[\Loc_{\GC}(M) \to \Loc_{\GC}(\partial M)\] quantizes to the BV-BFV partition function.

\item \textbf{Conjecture \ref{conj:cellular}} (Cellular compatibility): The cellular BV-BFV partition function on a triangulation agrees with the factorization homology colimit.

\item \textbf{Conjecture \ref{conj:koszul}} (Perturbative expansion and vanishing cycles): the
BV cohomology of the perturbative expansion around a flat connection $[A]$ is the stalk at $[A]$
of the vanishing-cycles sheaf $P_M$ on $\Loc_{\GC}(M)$; and (Gunningham--Jordan--Safronov) for
generic $q$ the skein module of $M$ is the hypercohomology of the vanishing-cycles sheaf.

\item \textbf{Conjecture \ref{conj:MCG}} (Mapping class group actions agree): the projective
representations of $\MCG(\Sigma_g)$ coming from RT theory, geometric quantization and
factorization homology are isomorphic, with the same central extension. \emph{(Largely known; see
Remark \ref{rmk:MCG-known}.)}

\item \textbf{Conjecture \ref{conj:fully-extended}} (Fully extended theory): the non-perturbative
BV-BFV theory extends down to points as a boundary theory for the $4$-dimensional Crane--Yetter
theory of $\overline{\Rep}_q(\GC)$, and in this form agrees with RT theory.

\item \textbf{Conjecture \ref{conj:categorification}} (Categorification): Khovanov homology is the
value of a $4$-dimensional TQFT on link complements.
\end{enumerate}

\medskip
\noindent In addition we record two \textbf{Questions}:
\begin{enumerate}[label=\textbf{Q\arabic*.},leftmargin=3em]
\item \textbf{Question \ref{q:resurgence-koszul}} (Resurgence--Koszul correspondence): are the
Stokes constants counts of complexified instantons between flat connections, with local data given
by the stalks of $P_M$; are the Stokes automorphisms the wall-crossing functors between windows;
and do the Stokes automorphisms form a Kontsevich--Soibelman stability datum with central charge
$\CS$?

\item \textbf{Question \ref{q:windows}} (Windows and resummation): do the window subcategories of
the stratification of $\Loc_{\GC}(M)$ correspond to choices of Borel resummation ray, the
wall-crossing functors between them to Stokes automorphisms, and the window-independent object to
the exact RT invariant?
\end{enumerate}

\noindent and two organizing statements which are proposals about how to structure the theory
rather than precise assertions: Principle \ref{princ:strata} (trans-series sectors are strata,
not points) and Proposal
\ref{prop:CS-central-charge} (the Chern--Simons functional is the central charge).

\medskip

The logical structure is as follows. C3 (quantization agreement) implies C1
($\EE_2$-equivalence) at the perturbative level; C1 combined with factorization homology and the
AKZ theorem gives the surface-level identifications of C2. The full Main Conjecture C2
(including cobordism maps) additionally requires C4 (quantized Lagrangian correspondences),
which is logically independent of C1 and which, as Remark \ref{rmk:Lagrangian-maps} explains,
requires a compositionality statement that is itself a substantial open problem. C5 (cellular
compatibility) provides concrete computational access to the factorization homology side. C6
relates the perturbative expansion around each flat connection to a canonical global object, the
vanishing-cycles sheaf, and through the conjecture of Gunningham--Jordan--Safronov to skein
modules; its local classical counterpart is Proposition \ref{prop:koszul-classical}. Q2 (windows) asks how resummation-independence might appear categorically, conditional on the
stratification of Section \ref{subsec:theta-strat}, and Q1 asks whether the algebraic route of C6
and the analytic route of resurgence encode the same information. C7 is largely a theorem on the
RT and geometric-quantization sides, and open on the BV-BFV side. C8 and C9 concern extensions
beyond the $(3\text{-}2\text{-}1)$-extended setting.

The statements are not equally supported. C1(a), C3 and C5 concern perturbative data and could in
principle be settled by developing the derived BV-BFV formalism; C1(b) is an algebraization
statement about quantum groups which is essentially known once the associator is matched; C4 and
its compositionality are the hardest of the conjectures; and Q1, Q2 and the two organizing
statements suggest specific computations on a specific algebraic stack.

\bigskip

The synthesis proposed in this paper views the Reshetikhin--Turaev invariants and BV-BFV perturbative quantization as two facets of a single geometric object, the derived character stack $\Loc_{\GC}$ with its shifted symplectic structure. Factorization homology translates between local algebraic data (the $\EE_2$-category attached to a disk) and global topological invariants (the RT state spaces and partition functions). Koszul duality relates the perturbative expansion around each flat connection to the formal geometry of the stack, and the passage from local to global is governed by the geometry of the stack itself, through the vanishing-cycles sheaf and the stratification by stabilizer type. If the program can be carried out, it would show that a substantial part of the gap between perturbative and non-perturbative quantum field theory, at least for Chern--Simons theory, can be formulated algebraically, as the passage from local to global in the sheaf theory of derived moduli spaces; the remaining analytic input, the Stokes data between different values of the Chern--Simons functional, is isolated in Question \ref{q:resurgence-koszul}. The evidence assembled here is far from a proof, and we hope that the conjectures formulated here will be a useful guide for future work.

\end{document}